\documentclass[11pt,twoside,a4paper,notitlepage]{article}

\usepackage{ifthen}

\newboolean{publisherversion}
\setboolean{publisherversion}{false}

\ifthenelse{\boolean{publisherversion}}
{}
{
  \usepackage[a4paper,margin=2.5cm]{geometry}
  \usepackage{tgtermes}
  \usepackage[scale=.92]{tgheros}
  \usepackage{tgcursor}
}

\ifthenelse{\boolean{publisherversion}}
{}
{
}

\ifthenelse{\boolean{publisherversion}}
{}
{
  \usepackage[english]{babel}
}

\ifthenelse{\boolean{publisherversion}}
{}
{
  \usepackage{fancyhdr}
}

\ifthenelse{\boolean{publisherversion}}
{}
{
  
}

\usepackage{ifthen}

\usepackage{ifthen}

\provideboolean{detectedSTOC}
\provideboolean{detectedFOCS}
\provideboolean{detectedElsevier}
\provideboolean{detectedNOW}
\provideboolean{detectedLMCS}
\provideboolean{detectedIEEE}
\provideboolean{detectedPoster}
\provideboolean{detectedSIAM}
\provideboolean{detectedLNCS}
\provideboolean{detectedACM}
\provideboolean{detectedACMconf}
\provideboolean{detectedSigplanconf}
\provideboolean{detectedToC}
\provideboolean{detectedLIPIcs}
\provideboolean{detectedAAAI}
\provideboolean{detectedIJCAI}
\provideboolean{detectedCompCplx}
\provideboolean{detectedEasyChair}
\provideboolean{detectedJAIR}
\provideboolean{detectedArticle}
\provideboolean{detectedReport}
\provideboolean{detectedThesis}

\makeatletter

\@ifclassloaded{sig-alternate}
{\setboolean{detectedSTOC}{true}}
{\setboolean{detectedSTOC}{false}}

\@ifclassloaded{elsarticle}
{\setboolean{detectedElsevier}{true}}
{\setboolean{detectedElsevier}{false}}

\@ifclassloaded{now}
{\setboolean{detectedNOW}{true}}
{\setboolean{detectedNOW}{false}}

\@ifclassloaded{lmcs}
{\setboolean{detectedLMCS}{true}}
{\setboolean{detectedLMCS}{false}}

\@ifclassloaded{IEEEtran} {
  \setboolean{detectedIEEE}{true}
  \ifCLASSOPTIONconference {
    \setboolean{detectedFOCS}{true}
  }
  \else {
    \setboolean{detectedFOCS}{false}
  }
  \fi
}
{
  \setboolean{detectedFOCS}{false}
  \setboolean{detectedIEEE}{false}
}

\@ifclassloaded{siamart171218}
{\setboolean{detectedSIAM}{true}}
{\setboolean{detectedSIAM}{false}}

\@ifclassloaded{llncs}
{\setboolean{detectedLNCS}{true}}
{\setboolean{detectedLNCS}{false}}

\@ifclassloaded{acmsmall}
{\setboolean{detectedACM}{true}}
{\setboolean{detectedACM}{false}}

\@ifclassloaded{acmart}
{\setboolean{detectedACMconf}{true}
 \setboolean{detectedACM}{true}}
{\setboolean{detectedACMconf}{false}}

\@ifclassloaded{sigplanconf}
{\setboolean{detectedSigplanconf}{true}}
{\setboolean{detectedSigplanconf}{false}}

\@ifclassloaded{toc}
{\setboolean{detectedToC}{true}}
{\setboolean{detectedToC}{false}}

\@ifclassloaded{lipics}
{\setboolean{detectedLIPIcs}{true}}
{\@ifclassloaded{lipics-v2019}
  {\setboolean{detectedLIPIcs}{true}}
  {\@ifclassloaded{oasics-v2019}
    {\setboolean{detectedLIPIcs}{true}}
    {\setboolean{detectedLIPIcs}{false}}
  }
}

\@ifclassloaded{cc}
{\setboolean{detectedCompCplx}{true}}
{\setboolean{detectedCompCplx}{false}}

\@ifclassloaded{easychair}
{\setboolean{detectedEasyChair}{true}}
{\setboolean{detectedEasyChair}{false}}

\@ifpackageloaded{jair}
{\setboolean{detectedJAIR}{true}}        
{\setboolean{detectedJAIR}{false}}        

\@ifpackageloaded{aaai}
{\setboolean{detectedAAAI}{true}}        
{\@ifpackageloaded{aaai18}
  {\setboolean{detectedAAAI}{true}}
  {\@ifpackageloaded{aaai20}       
    {\setboolean{detectedAAAI}{true}}
    {\setboolean{detectedAAAI}{false}}}}

\@ifpackageloaded{ijcai18}
{\setboolean{detectedIJCAI}{true}}        
{\@ifpackageloaded{ijcai19}
  {\setboolean{detectedIJCAI}{true}}        
  {\setboolean{detectedIJCAI}{false}}}

\@ifclassloaded{sciposter}
{\setboolean{detectedPoster}{true}}
{\setboolean{detectedPoster}{false}}

\@ifclassloaded{article}
{\setboolean{detectedArticle}{true}}
{\setboolean{detectedArticle}{false}}

\makeatother

\ifthenelse{\not \isundefined{\examen} 
  \and \not \isundefined{\disputationsdatum} 
  \and \not \isundefined{\disputationslokal}}   
  {\setboolean{detectedThesis}{true}}
  {\setboolean{detectedThesis}{false}}

\ifthenelse{\boolean{detectedArticle} \or \boolean{detectedThesis}
  \or \boolean{detectedSTOC}    \or \boolean{detectedFOCS}
  \or \boolean{detectedSIAM}    \or \boolean{detectedIEEE}
  \or \boolean{detectedACMconf} \or \boolean{detectedACM}
  \or \boolean{detectedPoster}}
{\setboolean{detectedReport}{false}}
{\setboolean{detectedReport}{true}}

\ifthenelse{\boolean{detectedLIPIcs}}
{}
{\usepackage[utf8]{inputenc}}

\ifthenelse{\boolean{detectedIEEE}}
{\usepackage[cmex10]{amsmath}
  \interdisplaylinepenalty=2500}
{\usepackage{amsmath}}
\ifthenelse{\boolean{detectedACM}}{}
{\usepackage{amssymb}}
\usepackage{amsfonts}

\usepackage{mathtools}

\usepackage{bussproofs}

\usepackage[linesnumbered,ruled,vlined]{algorithm2e}

\usepackage{microtype}

\usepackage{bm}  

\usepackage{subcaption} 

\ifthenelse
{     \boolean{detectedElsevier} \or \boolean{detectedSIAM}
  \or \boolean{detectedSIAM}     \or \boolean{detectedLIPIcs}}
{}
{\usepackage{varioref}}

\usepackage{xspace}

\ifthenelse    
{\boolean{detectedSTOC}      \or \boolean{detectedSIAM}         \or 
  \boolean{detectedLMCS}     \or \boolean{detectedNOW}          \or 
  \boolean{detectedLNCS}     \or \boolean{detectedACM}          \or
  \boolean{detectedToC}      \or \boolean{detectedLIPIcs}       \or
  \boolean{detectedACMconf}  \or \boolean{detectedAAAI}         \or
  \boolean{detectedCompCplx} \or \boolean{detectedSigplanconf}}
{}
{
  \usepackage{amsthm}
  \usepackage{amsthm-boldfont}
}

\ifthenelse        
{\boolean{detectedElsevier}  \or \boolean{detectedFOCS}         \or 
  \boolean{detectedSTOC}     \or \boolean{detectedPoster}       \or
  \boolean{detectedSIAM}     \or \boolean{detectedLMCS}         \or
  \boolean{detectedIEEE}     \or \boolean{detectedNOW}          \or
  \boolean{detectedToC}      \or \boolean{detectedThesis}       \or
  \boolean{detectedLNCS}     \or \boolean{detectedACM}          \or 
  \boolean{detectedLIPIcs}   \or \boolean{detectedAAAI}         \or
  \boolean{detectedACMconf}  \or \boolean{detectedIJCAI}        \or 
  \boolean{detectedCompCplx} \or \boolean{detectedSigplanconf}}
{}
{\usepackage{jncaptionsheadings}}

\usepackage{ifthen}
\usepackage{xspace}
\ifthenelse
{\boolean{detectedElsevier} \or \boolean{detectedSIAM} 
  \or \boolean{detectedLIPIcs}}
{}
{\usepackage{varioref}}

\DeclareMathAlphabet{\mathsfsl}{OT1}{cmss}{m}{sl}

\DeclareRobustCommand{\BibTeX}{%
  {\normalfont B\kern-.05em{\scshape i\kern-.025em b}\kern-.08em \TeX}%
}

\newcommand{\formuladots}{\cdots}

\newcommand{\bigoh}[1]{\mathrm{O} ( #1 )}

\newcommand{\bigtheta}[1]{\Theta ( #1 )}

\newcommand{\Bigomega}[1]{\Omega \bigl( #1 \bigr)}
\newcommand{\bigomega}[1]{\Omega ( #1 )}

\ifthenelse{\boolean{detectedToC}}{}
{

  \newcommand{\N}         {\mathbb{N}}
  \newcommand{\Nplus}     {\mathbb{N}^{+}}

}

\providecommand{\abs}[1]{\lvert#1\rvert}
\providecommand{\Abs}[1]{\bigl\lvert#1\bigr\rvert}

\newcommand{\floor}[1]{\lfloor #1 \rfloor}

\newcommand{\MAXOFEXPR}[2][]{\max_{#1} \left\{ #2 \right\}}
\newcommand{\MINOFEXPR}[2][]{\min_{#1} \left\{ #2 \right\}}
\newcommand{\Maxofexpr}[2][]{\max_{#1} \bigl\{ #2 \bigr\}}
\newcommand{\Minofexpr}[2][]{\min_{#1} \bigl\{ #2 \bigr\}}
\newcommand{\maxofexpr}[2][]{\max_{#1} \{ #2 \}}
\newcommand{\minofexpr}[2][]{\min_{#1} \{ #2 \}}

\newcommand{\MAXOFSET}[3][:]%
     {\ifthenelse{\equal{#1}{;}}%
     {\MAXOFEXPR{ #2 \,;\, #3 }}
     {\ifthenelse{\equal{#1}{:}}%
     {\MAXOFEXPR{ #2 \,:\, #3 }}
     {\max \twincommandJN{\left\{}{#2}{\left#1}{\right}{\,#3}{\right\}}}}}
\newcommand{\MINOFSET}[3][:]%
     {\ifthenelse{\equal{#1}{;}}%
     {\MINOFEXPR{ #2 \,;\, #3 }}
     {\ifthenelse{\equal{#1}{:}}%
     {\MINOFEXPR{ #2 \,:\, #3 }}
     {\min \twincommandJN{\left\{}{#2}{\left#1}{\right}{\,#3}{\right\}}}}}

\newcommand{\Maxofset}[3][:]%
     {\ifthenelse{\equal{#1}{;}}%
     {\Maxofexpr{ #2 \,;\, #3 }}
     {\ifthenelse{\equal{#1}{:}}%
     {\Maxofexpr{ #2 \,:\, #3 }}
     {\max \twincommandJN{\bigl\{}{#2}{\bigl#1}{\bigr}{\,#3}{\bigr\}}}}}
\newcommand{\Minofset}[3][:]%
     {\ifthenelse{\equal{#1}{;}}%
     {\Minofexpr{ #2 \,;\, #3 }}
     {\ifthenelse{\equal{#1}{:}}%
     {\Minofexpr{ #2 \,:\, #3 }}
     {\min \twincommandJN{\bigl\{}{#2}{\bigl#1}{\bigr}{\,#3}{\bigr\}}}}}

\newcommand{\F}{\mathbb{F}}

\DeclareMathOperator{\Expop}{E}

\ifthenelse{\boolean{detectedLMCS}}
{}
{}

\newcommand{\twincommandJN}[6]%
    {#1#2#3\vphantom{#2#5}\mspace{-2.05mu}#4.#5#6}

\newcommand{\CondExp}[2]%
    {\Expop\twincommandJN{\bigl[}{#1}{\bigl|}{\bigr}{\,#2}{\bigr]}}
\newcommand{\CONDEXP}[2]%
     {\Expop\twincommandJN{\left[}{#1}{\left|}{\right}{\,#2}{\right]}}

\newcommand{\Condprob}[3][]%
    {\Pr_{#1}\twincommandJN{\bigl[}{#2}{\bigl|}{\bigr}{\,#3}{\bigr]}}
\newcommand{\CONDPROB}[3][]%
    {\Pr_{#1}\twincommandJN{\left[}{#2}{\left|}{\right}{\,#3}{\right]}}

\newcommand{\funcdescr}[3]{\ensuremath{ #1 : #2 \to #3}}

\newcommand{\edges}[1]{E( #1 )}

\newcommand{\vertices}[1]{V( #1 )}

\newcommand{\set}[1]{\{ #1 \}}
\newcommand{\Set}[1]{\bigl\{ #1 \bigr\}}

\newcommand{\setdescr}[3][\mid]{\set{ #2 #1 #3 }}
\newcommand{\Setdescr}[3][|]%
     {\ifthenelse{\equal{#1}{;}}%
     {\Set{ #2 \,;\, #3 }}
     {\ifthenelse{\equal{#1}{:}}%
     {\Set{ #2 \,:\, #3 }}
     {\twincommandJN{\bigl\{}{#2\,}{\bigl#1}{\bigr}{\,#3}{\bigr\}}}}}
\newcommand{\SETDESCR}[3][|]%
     {\twincommandJN{\left\{}{#2\,}{\left#1}{\right}{\,#3}{\right\}}}

\newcommand{\Setdescrbrackets}[3][|]%
     {\twincommandJN{\bigl[}{#2}{\bigl#1}{\bigr}{\,#3}{\bigr]}}
\newcommand{\SETDESCRBRACKETS}[3][|]%
     {\twincommandJN{\left[}{#2}{\left#1}{\right}{\,#3}{\right]}}

\newcommand{\Setsize}[1]{\bigl\lvert#1\bigr\rvert}
\newcommand{\setsize}[1]{\lvert#1\rvert}

\newcommand{\intersection}{\cap}

\newcommand{\union}{\cup}

\newcommand{\olnot}[1]{\overline{#1}}
\newcommand{\stdnot}[1]{\olnot{#1}}

\newcommand{\nvar}{n}

\newcommand{\nclause}{m}

\newcommand{\clwidth}{k}

\newcommand{\xcnf}[1]{\mbox{\ensuremath{#1}-CNF}\xspace}

\newcommand{\randkcnfnclwrepl}[3][\clwidth]%
        {\ensuremath{\mathcal{F}^{#2, #3}_{#1}}}
\newcommand{\randkcnfnclwreplstd}%
        {\randkcnfnclwrepl{\clwidth}{\nvar}{\nclause}}

\newcommand{\complclassformat}[1]%
        {\textrm{\upshape{\textsf{#1}}}\xspace}
\newcommand{\cocomplclass}[1]%
        {\textrm{\upshape{\textsf{co#1}}}\xspace}

\newcommand{\DTIMEadviceclass}[2]%
    {\ensuremath{\complclassformat{DTIME}\bigl(#1\bigr)/{#2}}}

\newcommand{\PCPalph}[5]%
    {\ensuremath{\complclassformat{PCP}_{{#1},{#2}}[{#3}, {#4}, {#5}]}}
\newcommand{\PCP}[4]%
    {\ensuremath{\complclassformat{PCP}_{{#1},{#2}}[{#3}, {#4}]}}

\newcommand{\eqperiod}{\enspace .}
\newcommand{\eqcomma}{\enspace ,}
\renewcommand{\eqperiod}{\, .}
\renewcommand{\eqcomma}{\, ,}

\ifthenelse{\boolean{detectedToC}}{}
  { %
    }
\ifthenelse{\boolean{detectedToC}}{}{
\newcommand{\ie}{i.e.,\ }

}

\ifthenelse{\boolean{detectedLIPIcs} \or \boolean{detectedIJCAI}}
{\renewcommand{\st}{\errmessage{Please do not use st}}}
{\ifthenelse{\isundefined{\st}}
  {\newcommand{\st}{such that\xspace}}
  {}}      

\ifthenelse{\boolean{detectedIEEE}}{}{}

\ifthenelse
{\isundefined{\refeq}}
{\newcommand{\refeq}[1]{\eqref{#1}}}
{\renewcommand{\refeq}[1]{\eqref{#1}}}

\usepackage[dvipsnames]{xcolor}
\definecolor{newcolor}{hsb}{0.6,1,0.75}

\ifthenelse{\isundefined{\DONOTLOADHYPERREFPACKAGE}
  \and \not \boolean{detectedSTOC}        \and \not \boolean{detectedFOCS}
  \and \not \boolean{detectedPoster}      \and \not \boolean{detectedElsevier} 
  \and \not \boolean{detectedSIAM}        \and \not \boolean{detectedACM}
  \and \not \boolean{detectedIEEE}        \and \not \boolean{detectedNOW}
  \and \not \boolean{detectedToC}         \and \not \boolean{detectedThesis}
  \and \not \boolean{detectedLIPIcs}      \and \not \boolean{detectedSIAM}
  \and \not \boolean{detectedAAAI}        \and \not \boolean{detectedIJCAI}
  \and \not \boolean{detectedSigplanconf} \and \not \boolean{detectedACMconf}   
  \and \not \boolean{detectedCompCplx} \and \not \boolean{detectedEasyChair}}
{\usepackage[colorlinks=true, linkcolor  = newcolor, citecolor=newcolor, urlcolor   = newcolor]{hyperref}}
{}

\ifthenelse{\boolean{detectedSTOC}}                  
{\input{definitions-STOC.tex}}
{\ifthenelse{\boolean{detectedSIAM}}
  {\input{definitions-SIAM.tex}}
  {\ifthenelse{\boolean{detectedNOW}}
    {\input{definitions-NOW.tex}}
    {\ifthenelse{\boolean{detectedLMCS}}
      {\input{definitions-lmcs.tex}}
      {\ifthenelse{\boolean{detectedIEEE}}
        {\input{definitions-IEEE.tex}}
        {\ifthenelse{\boolean{detectedLNCS}}
          {\input{definitions-LNCS.tex}}
          {\ifthenelse{\boolean{detectedACM} \or \boolean{detectedACMconf}}
            {\input{definitions-ACM.tex}}
            {\ifthenelse{\boolean{detectedToC}}
              {\input{definitions-ToC.tex}}
              {\ifthenelse{\boolean{detectedLIPIcs}}
                {\input{definitions-LIPIcs.tex}}
                {\ifthenelse{\boolean{detectedCompCplx}}
                  {\input{definitions-CompCplx.tex}}
                  {\ifthenelse{\boolean{detectedSigplanconf}}
                    {\input{definitions-Sigplanconf.tex}}
                    {\ifthenelse{\boolean{detectedAAAI}}
                      {}
                      {\ifthenelse{\boolean{detectedArticle} 
                          \or \boolean{detectedElsevier}
                          \or \boolean{detectedEasyChair}}
                        {%
\newtheorem{standardlocalcounter}{Dummy}[section]

\usepackage{thmtools}
\theoremstyle{plain}    

\declaretheorem[sibling=standardlocalcounter,name=Theorem]{theorem}
\declaretheorem[sibling=standardlocalcounter,name=Lemma]{lemma}
\declaretheorem[sibling=standardlocalcounter,name=Proposition]{proposition}
\declaretheorem[sibling=standardlocalcounter,name=Corollary]{corollary}
\declaretheorem[sibling=standardlocalcounter,name=Fact]{fact}

\theoremstyle{definition}

\declaretheorem[sibling=standardlocalcounter,name=Observation]{observation}

\declaretheorem[sibling=standardlocalcounter,name=Property]{property}
\declaretheorem[sibling=standardlocalcounter,name=Definition]{definition}
\declaretheorem[sibling=standardlocalcounter,name=Claim]{claim}

\theoremstyle{remark}

\declaretheorem[sibling=standardlocalcounter,name=Remark]{remark}

}
                        {\input{definitions-report.tex}}}}}}}}}}}}}}

\ifthenelse{\boolean{detectedThesis}}
{}
{\usepackage{ifpdf}}

\ifthenelse{\boolean{detectedToC}}
{}
{\usepackage{graphicx}}

\ifpdf         
\fi

\ifthenelse
{\boolean{detectedSTOC}   \or \boolean{detectedThesis} \or 
 \boolean{detectedLNCS}   \or \boolean{detectedToC}    \or 
 \boolean{detectedLIPIcs} \or \boolean{detectedAAAI}   \or
 \boolean{detectedIJCAI}  \or \boolean{detectedSIAM}}
{}
{
  \setcounter{secnumdepth}{3}
  \setcounter{tocdepth}{3}
}

\newcount\hours
\newcount\minutes
\def\SetTime{\hours=\time
\global\divide\hours by 60
\minutes=\hours
\multiply\minutes by 60
\advance\minutes by-\time
\global\multiply\minutes by-1 }
\SetTime
\def\now{\number\hours:\ifnum\minutes<10 0\fi\number\minutes}

\newcommand{\proofstd}{\pi}

\newcommand{\refpi}{\pi}

\newcommand{\derivof}[4][\derives]
        {{\ensuremath{{#2} : {#3} \, {#1}\, {#4}}}}
\newcommand{\refof}[2]{\derivof{#1}{#2}{\bot}}

\newcommand{\emptycl}{\bot}

\newcommand{\formf}{\ensuremath{F}}

\newcommand{\varx}{\ensuremath{x}}
\ifthenelse{\boolean{detectedACMconf}}
{}
{}

\ifthenelse{\boolean{detectedACMconf}}
{
 }
{
 }

\newcommand{\clc}{\ensuremath{C}}
\newcommand{\cld}{\ensuremath{D}}

\newcommand{\SETSOFVARSORLIT}[2]%
        {\mathit{#1}\left({#2}\right)}
\newcommand{\setsofvarsorlit}[2]%
        {\mathit{#1}({#2})}
\newcommand{\Setsofvarsorlit}[2]%
        {\mathit{#1}\bigl({#2}\bigr)}

\newcommand{\vars}[1]{\setsofvarsorlit{Vars}{#1}}

\newcommand{\restr}{\rho}

\newcommand{\restrict}[2]{{{#1}\!\!\upharpoonright_{#2}}}

\newcommand{\derivabbrev}[2]{\bigl( #1 \vdash #2 \bigr)}
\newcommand{\derivabbrevsmall}[2]{( #1 \vdash #2 )}
\newcommand{\derivabbrevcompact}[2]{\bigl( #1 \vdash #2 \bigr)}

\newcommand{\refutabbrevsmall}[1]{\derivabbrevsmall{#1}{\!\bot}}
\newcommand{\refutabbrevcompact}[1]{\derivabbrevcompact{#1}{\!\bot}}

\newcommand{\genericformsmall}[2]{\mathit{#1}( #2 )}

\newcommand{\genericrefsmall}[3]%
    {{\mathit{#1}}_{#2}\refutabbrevsmall{#3}}
\newcommand{\genericrefcompact}[3]%
    {{\mathit{#1}}_{#2}\refutabbrevcompact{#3}}
\newcommand{\genericderiv}[4]%
    {{\mathit{#1}}_{#2}\derivabbrev{#3}{#4}}
\newcommand{\genericderivsmall}[4]%
    {{\mathit{#1}}_{#2}\derivabbrevsmall{#3}{#4}}
\newcommand{\genericderivcompact}[4]%
    {{\mathit{#1}}_{#2}\derivabbrevcompact{#3}{#4}}
\newcommand{\generictaut}[3]%
    {{\mathit{#1}}_{#2}\derivabbrev{}{#3}}
\newcommand{\generictautcompact}[3]%
    {{\mathit{#1}}_{#2}\derivabbrevcompact{}{#3}}
\newcommand{\generictautsmall}[3]%
    {{\mathit{#1}}_{#2}\derivabbrevsmall{}{#3}}

\newcommand{\sizestd}{S}
\newcommand{\size}[1]{\genericformsmall{S}{#1}}

\newcommand{\sizeofarg}[1]{\genericformsmall{S}{#1}}

\newcommand{\widthofarg}[2][]{\genericformsmall{W_{#1}}{#2}}

\newcommand{\formulaformat}[1]{\mathit{#1}}

\newcommand{\extendedversion}[1]{\widetilde{#1}}

\newcommand{\TSEITINFORM}{Tseitin Formula\xspace}

\newcommand{\epopnot}[1]%
    {\extendedversion{\formulaformat{POP}}_{#1}}

\newcommand{\elopnot}[1]%
    {\extendedversion{\formulaformat{LOP}}_{#1}}

\newcommand{\ephpnot}[2]%
    {\vphantom{\extendedversion{\formulaformat{PHP}}}
      {\smash{\extendedversion{\formulaformat{PHP}}}
        \vphantom{\formulaformat{PHP}}}^{#1}_{#2}}

\newcommand{\efphpnot}[2]%
    {\vphantom{\extendedversion{\formulaformat{FPHP}}}
      {\smash{\extendedversion{\formulaformat{FPHP}}}
        \vphantom{\formulaformat{FPHP}}}^{#1}_{#2}}

\newcommand{\ontophpnot}[2]%
    {\formulaformat{Onto}\text{-}\formulaformat{PHP}^{#1}_{#2}}
\newcommand{\ontofphpnot}[2]%
    {\formulaformat{Onto}\text{-}\formulaformat{FPHP}^{#1}_{#2}}

\newcommand{\graphontophpnot}[1][G]%
    {\text{$\formulaformat{Onto}$-$\formulaformat{PHP}$}({#1})}

\newcommand{\perfectmatchingnot}[1][G]%
    {\formulaformat{PM}({#1})}
\newcommand{\general}{general\xspace}

\newcommand{\Ts}{\mathrm{Ts}}
\newcommand{\lcm}{\mathrm{lcm}}

\newcommand{\earlength}{r}
\newcommand{\midlength}{L}
\newcommand{\modulus}[1]{{m_{#1}}}
\newcommand{\vertexeq}{{\equiv_V}}
\newcommand{\edgeeq}{{\equiv_E}}
\newcommand{\unirow}{\mathrm{row}_{\mathrm{unique}}}
\newcommand{\diam}{\mathrm{diam}}
\newcommand{\basegraph}{{G_{\mathrm{cyl}}}}
\newcommand{\coppos}{\mathbb C}

\newcommand{\depthofarg}[1]{\genericformsmall{D}{#1}}

\newcommand{\clausesizeofarg}[1]{{\vert{#1}\vert}}

\newcommand{\szs}{s}

\newcommand{\xeq}[1]{{x_{#1/\edgeeq}}}

\newcommand{\vclassset}[1]{{#1_{\vertexeq}}}
\newcommand{\eclass}[1]{{#1/_{\edgeeq}}}
\newcommand{\eclassset}[1]{{#1_{\edgeeq}}}

\newcommand{\MSET}[1]{\Big\{\!\!\Big\{#1\Big\}\!\!\Big\}}

\newcommand{\CFI}[2]{{\mathrm{CFI}(#1,#2)}}

\newcommand{\eps}{{\varepsilon}}

\newcommand{\indexing}[2]{\mathrm{IND}_{#1}\!\left({#2}\right)}
\newcommand{\indexlength}{m}
\newcommand{\indexingnolength}[1]{\indexing{\indexlength}{#1}}

\newcommand{\indexingformulawithoutlength}{\IND(\formf)}

\newcommand{\IND}{\mathrm{IND}}

\newcommand{\XOR}[2]{\mathrm{XOR}_{#2}\!\left({#1}\right)}

\newcommand{\var}{\mathit{Vars}}

\newcommand{\RowInSet}[2]{{#2}{[{#1}]}}

\newcommand{\widthbndsize}{\floor{\log_{(\indexlength +1)/2} \sizestd}}
\newcommand{\widthlem}{w}
\newcommand{\numberofblocks}{{r}}
\newcommand{\sumindext}{t}

\newcommand{\proofbigpi}{\pi}

\newcommand{\lowercasewidthofarg}[1]{w(#1)}

\newcommand{\poly}{\mathrm{poly}}

\newcommand{\Cabs}{{C_{\mathrm{abs}}}}

\newcommand{\Sum}{\mathop\sum\limits}

\newcommand{\E}{{\mathbb E}}

\newcommand{\sat}{\mathrm{sat}}

\usepackage{enumitem}

\usepackage{tikz, thmtools}
\usetikzlibrary{decorations.pathreplacing}

\newcommand{\x}{{\bf x}}
\newcommand{\Y}{{\bf Y}}
\newcommand{\y}{{\bf y}}

\newcommand{\inn}{\text{\upshape in}}
\newcommand{\err}{\text{\upshape err}}

\newcommand{\spcnoindent}{\vspace{1em} \noindent}

\theoremstyle{plain}

\declaretheorem[name={Triangle Lemma},numbered=no]{triangleLemma}
\newcommand{\triLemma}{\hyperref[lem:triangle]{Triangle Lemma}\xspace}

\declaretheorem[name={Full Image Lemma},numbered=no]{fullImageLemma}
\newcommand{\fullImg}{\hyperref[lem:fullImage]{Full Image Lemma}\xspace}
\newcommand{\fullImgApp}{\hyperref[lem:fullImage_app]{Full Image Lemma}\xspace}

\newcommand{\bit}{b}
\newcommand{\xorText}{\mathrm{XOR}}
\newcommand{\xorsize}{m}
\newcommand{\indexsize}{m}
\newcommand{\relation}{\mathcal{S}}
\newcommand{\Output}{\mathcal{O}}
\newcommand{\gDomain}{\mathcal{D}}
\newcommand{\Search}{{\mathsf{Search}}}
\newcommand{\mKW}{{\mathsf{mKW}}}
\newcommand{\fix}{{\mathsf{fix}}}
\newcommand{\hati}{\hat{\imath}}
\newcommand{\hatj}{\hat{\jmath}}
\newcommand{\hatx}{\widehat{X}}
\newcommand{\hatS}{\widehat{S}}

\newcommand{\superlinear}{super-linear\xspace}

\usepackage{numname}
\newcounter{authorcount}
\newcommand{\newauthor}[3]{%
\newcounter{#1comment}
\expandafter\newcommand\csname #1comment\endcsname[1]{%
\ifdefined\DONOTINSERTCOMMENTS\relax\else%
\medskip\par\noindent
{\bfseries \scshape \footnotesize #2's comment
\stepcounter{#1comment}\csname the#1comment\endcsname}:
{\sffamily \itshape \scriptsize\textcolor{#3}{##1}\par}
\medskip%
\fi}
\stepcounter{authorcount}
\expandafter\edef\csname #1ordinal\endcsname{\theauthorcount}
\expandafter\newcommand\csname theauthor#1\endcsname{%
the \ordinaltoname{\csname #1ordinal\endcsname} author\xspace}
\expandafter\newcommand\csname Theauthor#1\endcsname{%
The \ordinaltoname{\csname #1ordinal\endcsname} author\xspace}
}

\definecolor{airforceblue}{rgb}{0.36, 0.54, 0.66}
\definecolor{amethyst}{rgb}{0.6, 0.4, 0.8}
\definecolor{asparagus}{rgb}{0.53, 0.66, 0.42}
\definecolor{brass}{rgb}{0.71, 0.65, 0.26}
\definecolor{brown}{rgb}{0.59, 0.29, 0.0}
\definecolor{darkolivegreen}{rgb}{0.33, 0.42, 0.18}
\definecolor{darkorange}{rgb}{1.0, 0.55, 0.0}
\definecolor{darkorange}{rgb}{1.0, 0.55, 0.0}
\definecolor{maroon}{rgb}{0.5, 0, 0}
\definecolor{darkgreen}{rgb}{0, 0.6, 0}
\definecolor{betterYellow}{RGB}{255,255,0}
\definecolor{betterGreen}{RGB}{102,224,102}

\usepackage{thmtools,thm-restate}
\theoremstyle{plain}

\declaretheorem[%
sibling=standardlocalcounter,
name=Theorem
]%
{restatabletheorem}
\declaretheorem[%
sibling=standardlocalcounter,
name=Lemma
]%
{restatablelemma}

\makeatletter
\renewenvironment{thmt@restatable}[3][]{%
  \thmt@toks{}%
  \stepcounter{thmt@dummyctr}%
  \long\def\thmrst@store##1{%
    \@xa\gdef\csname #3\endcsname{%
      \@ifstar{%
        \thmt@thisistheonefalse\csname thmt@stored@#3\endcsname
      }{%
        \thmt@thisistheonetrue\csname thmt@stored@#3\endcsname
      }%
    }%
    \@xa\long\@xa\gdef\csname thmt@stored@#3\@xa\endcsname\@xa{%
      \begingroup
      \ifthmt@thisistheone
      \else
        \@xa\protected@edef\csname the#2\endcsname{%
          \thmt@trivialref{thmt@@#3}{??} (Restated)}%
        \ifcsname r@thmt@@#3\endcsname\else
          \G@refundefinedtrue
        \fi
        \@xa\let\csname c@#2\endcsname=\c@thmt@dummyctr
        \@xa\let\csname theH#2\endcsname=\theHthmt@dummyctr
        \let\label=\@gobble
        \let\ltx@label=\@gobble%
        \def\thmt@restorecounters{}%
        \@for\thmt@ctr:=\thmt@innercounters\do{%
          \protected@edef\thmt@restorecounters{%
            \thmt@restorecounters
            \protect\setcounter{\thmt@ctr}{\arabic{\thmt@ctr}}%
          }%
        }%
        \thmt@trivialref{thmt@@#3@data}{}%
      \fi
      \ifthmt@restatethis
        \thmt@restatethisfalse
      \else
        \csname #2\@xa\endcsname\ifx\@nx#1\@nx\else[{#1}]\fi
      \fi
      \ifthmt@thisistheone
         \thmt@rst@storecounters{#3}%
        \label{thmt@@#3}%
      \fi
      ##1%
      \csname end#2\endcsname
      \ifthmt@thisistheone\else\thmt@restorecounters\fi
      \endgroup
    }%
    \csname #3\@xa\endcsname\ifthmt@thisistheone\else*\fi
    \@xa\end\@xa{\@currenvir}
  }%
  \thmt@collect@body\thmrst@store
}{%
}
\makeatother

\makeatletter
\renewcommand{\paragraph}{%
  \@startsection{paragraph}{4}%
  {\z@}{1.25ex \@plus 1ex \@minus .2ex}{-1em}%
  {\normalfont\normalsize\bfseries}%
}
\makeatother

\ifthenelse{\boolean{publisherversion}}
{}
{
  \numberwithin{equation}{section}
}

\usepackage{comment}
\usepackage{tcolorbox}

\addto\extrasenglish{%
}

\setlist{listparindent=\parindent,parsep=0pt,itemsep=1em}
\setlist[itemize]{label=$-$,noitemsep}
\setlist[enumerate]{itemsep=1mm}

\begin{document}

\title{Truly Supercritical Trade-offs for Resolution, Cutting Planes,
  Monotone Circuits, and Weisfeiler--Leman%
}

\author{
  Susanna F.\ de Rezende \\
  Lund University
  \and
  Noah Fleming \\
  Memorial University
  \and
  Duri Andrea Janett \\
  University of Copenhagen \\
  and Lund University
  \and  
  Jakob Nordström \\
  University of Copenhagen \\and Lund University
  \and
  Shuo Pang\\
  University of Copenhagen
}

\date{\today}

\maketitle

\ifthenelse{\boolean{publisherversion}}
{}
{
  \thispagestyle{empty}

  \pagestyle{fancy}
  \fancyhead{}
  \fancyfoot{}
  \renewcommand{\headrulewidth}{0pt}
  \renewcommand{\footrulewidth}{0pt}

  \fancyhead[CE]{\slshape
    TRULY SUPERCRITICAL TRADE-OFFS%
  }
  \fancyhead[CO]{\slshape \nouppercase{\leftmark}}
  \fancyfoot[C]{\thepage}

  \setlength{\headheight}{13.6pt}
}
\begin{abstract}
We exhibit supercritical trade-off for monotone circuits, showing that there are functions computable by small circuits for which any circuit must have depth super-linear or even super-polynomial in the number of variables, far exceeding the linear worst-case upper bound. We obtain similar trade-offs in proof complexity, where we establish the first size-depth trade-offs for cutting planes and resolution that are truly supercritical, i.e., in terms of formula size rather than number of variables, and we also show supercritical trade-offs between width and size for treelike resolution.

Our results build on a new supercritical width-depth trade-off for resolution, obtained by refining and strengthening the compression scheme for the Cop-Robber game in [Grohe, Lichter, Neuen \& Schweitzer 2023]. This yields robust supercritical trade-offs for dimension versus iteration number in the Weisfeiler--Leman algorithm, which also translate into trade-offs between number of variables and quantifier depth in first-order logic. Our other results follow from improved lifting theorems that might be of independent interest.

\end{abstract}

   \newpage
   \thispagestyle{empty}
   \tableofcontents

 \pagenumbering{gobble}

    \newpage
    \pagenumbering{arabic}

\section{Introduction}
\label{sec:intro}

Computational complexity aims to understand the amount of resources---such as running time or 
memory---required in order to solve computational problems. An important regime is to understand 
how sets of resources interact: can they be optimized simultaneously or 
are there problems where there is necessarily a trade-off, when 
optimizing one resource leads to a substantial increase in the others? 
Traditionally, the strongest trade-offs between two complexity measures, say~$\mu$ 
and~$\nu$, have been of the following form: it is possible to solve the problem with a small value 
for~$\mu$ and with a small value for~$\nu$, but optimizing $\mu$ causes $\nu$ to increase to 
nearly the value obtained from the brute-force worst-case algorithm (see \autoref{fig:nonsuper}). 
In this setting,  \emph{robust} trade-offs have been established, where we cannot even approximately 
optimize $\mu$ without a blow-up for $\nu$ (corresponding to a tall infeasible region
in \autoref{fig:supercritical}).

Razborov~\cite{Razborov16NewKind} and the earlier works~\cite{BBI16TimeSpace,BNT13SomeTradeoffs} 
show that trade-offs exist which go far beyond this regime,
where optimizing one measure causes the other to increase \emph{beyond its 
worst-case value} (see \autoref{fig:super}). 
These \emph{supercritical} trade-offs have mostly appeared in proof 
complexity~\cite{BBI16TimeSpace,Berkholz12ComplexityNarrowProofs,BNT13SomeTradeoffs,%
Razborov16NewKind,Razborov17WidthSemialgebraic,Razborov18SpaceDepth,BN20Supercritical,%
fleming2022extremely,buss2024simple,CD24ResLin}
and finite model theory~\cite{BN23QuantifierDepth,grohe2023iteration}.
Recent
papers~\cite{GGKS20MonotoneCircuit,FGIPRTW21BranchAndCut,fleming2022extremely}
have raised the question of whether there are supercritical trade-offs in circuit complexity,
which brings us to the main contribution of this paper.

\begin{figure}[h]
\centering
    \begin{subfigure}[b]{0.4\textwidth}
    \centering
    \begin{tikzpicture}
  \filldraw[draw = red!80, fill = red!12] (0 ,0) -- (0 ,1) -- (1.7 ,1) -- (1.7,0) -- cycle ;
  \draw[->, thick] (0, 0) node[left] {} -- (4, 0) node[right] {$\nu$};
  \draw[->, thick] (0, 0) -- (0, 2.6) node[above] {$\mu$};
 \draw[dashed] (0, 1.6) node[left] {$\mu_{\textrm{worst}}$} -- (4, 1.6);
  \draw[dashed] (2.4, 0) node[below] {$\nu_{\textrm{worst}}$} -- (2.4, 2.6);
  \node[shape=circle,draw=blue,fill=blue, inner sep=0pt, minimum width=2.5pt] at (0.2,1.3) {};
  \node[shape=circle,draw=blue,fill=blue, inner sep=0pt, minimum width=2.5pt] at (2.1, 0.2) {};
\end{tikzpicture}
    \caption{\label{fig:nonsuper}}
    \end{subfigure}
\quad
    \begin{subfigure}[b]{0.4\textwidth}
    \centering
    \begin{tikzpicture}
  \filldraw[draw = red!80, fill = red!12] (6 ,0) -- (6 ,1) -- (9.6 ,1) -- (9.6,0) -- cycle ;
  \draw[->, thick] (6, 0) node[left] {} -- (10.5, 0) node[right] {$\nu$};
  \draw[->, thick] (6, 0) -- (6, 2.6) node[above] {$\mu$};
 \draw[dashed] (6, 1.6) node[left] {$\mu_{\textrm{worst}}$} -- (10, 1.6);
  \draw[dashed] (8.4, 0) node[below] {$\nu_{\textrm{worst}}$} -- (8.4, 2.6);  
  \node[shape=circle,draw=blue,fill=blue, inner sep=0pt, minimum width=2.5pt] at (6.2,1.3) {};
  \node[shape=circle,draw=blue,fill=blue, inner sep=0pt, minimum width=2.5pt] at (9.9, 0.2) {};
\end{tikzpicture}
    \caption{\label{fig:super}}
    \end{subfigure}
\quad
    \begin{subfigure}[b]{1\textwidth}
    \centering
    \begin{tikzpicture}
  \filldraw[draw = red!80, fill = red!12, label=right:{not feasible}] (3.8,.4) -- (3.8 ,0) -- (4.6,0) -- 
  (4.6,.4) -- cycle (4.6,.2) node[right] {infeasible}; 
\end{tikzpicture}
    \end{subfigure}
\caption{An illustration of trade-offs. Blue dots represent provable upper bounds on
measures $\mu$ and $\nu$. 
Proofs with measures in 
the shaded region are ruled out by 
    the trade-off, where $\mu_{\textrm{worst}}$ and 
    $\nu_{\textrm{worst}}$ are the worst-case upper bound on $\mu$ and $\nu$, respectively. 
    \autoref{fig:nonsuper} illustrates a non-supercritical trade-off and \autoref{fig:super} illustrates
    a supercritical one.}    \label{fig:supercritical}
\end{figure}
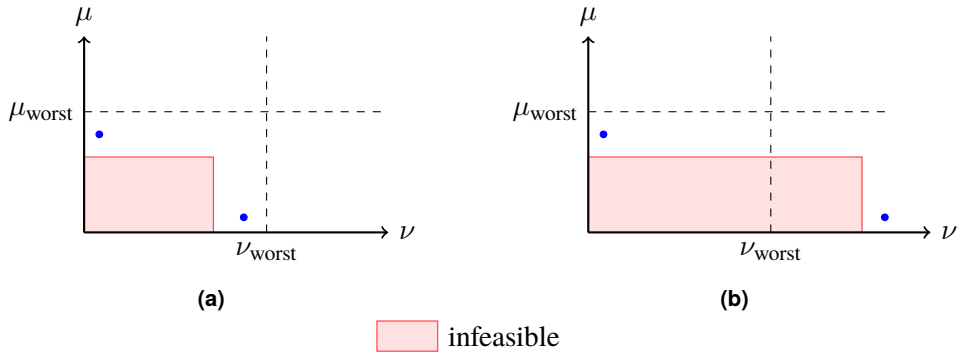

\subsection{Supercritical Trade-offs in Circuit Complexity}

Perfect matching is one of the classical problem in complexity theory.
Although it has been known to be solvable in polynomial time 
for nearly 70 years~\cite{Edmonds65Paths}, many 
questions about its
computational 
complexity remain unresolved,
in particular regarding its monotone complexity.
In a breakthrough result in 1985, Razborov~\cite{Razborov85LowerBounds} 
proved the first super-polynomial
size lower bound for monotone circuits---Boolean circuits with only AND and OR 
gates---for two functions: $k$-clique and bipartite perfect matching. 
(Independently, Andreev~\cite{Andreev85Monotone} showed
an exponential size lower bound for a different function.)
A few years later, Alon and Boppana~\cite{AB87Monotone} improved the %
lower bound for $k$-clique to exponential for large~$k$. 

For bipartite perfect matching, however, we still have no better size lower bounds.
Raz and Wigderson~\cite{RW92MonotoneCircuits} proved a
depth-$\bigomega{n}$ lower bound,
where $n$ is the number of vertices of the graph and the function has $\bigtheta{n^2}$ inputs.
This lower bound is tight, as there are monotone circuits that compute bipartite perfect 
matching in depth $\bigoh{n}$ (and size $2^{\bigoh{n}}$). 
In fact, this rather straightforward upper bound remains the best known to this day.
Are there monotone circuits computing bipartite perfect matching in size $n^{\bigoh{\log n}}$?
If so, why have we not yet been able to find them?
And if not, why have we not been able to prove a stronger lower bound?

One possible answer to these questions could be that we have not
been able to prove exponential lower bounds because they are simply not 
true, and that we have not been able to find smaller monotone
circuits computing perfect matching because \emph{they look different}.
We know already that if there is a monotone circuit of size $n^{\bigoh{\log n}}$
that computes bipartite perfect matching then it must have depth at least $\bigomega{n}$. 
But what if any monotone circuit of size $n^{\bigoh{\log n}}$ 
requires even larger depth, say depth $n^{\bigomega{\log n}}$? 
This could sound like an absurd hypothesis---how can a small circuit
require \superlinear depth? 
It is natural to ask, as was done 
in~\cite{GGKS20MonotoneCircuit,FGIPRTW21BranchAndCut,fleming2022extremely},
if there are any monotone functions
that exhibit this kind of supercritical trade-off behavior,
where small circuits exist but any small circuit requires
\superlinear depth.
We prove this is the case, even for the stronger model of monotone real circuits,
where gates can compute any monotone function from two real numbers
to a real number.

\begin{restatable}[Monotone circuit trade-offs]{restatabletheorem}{monotonerealtradeoff}
	\label{thm:monotonerealtrade-off}
	There are $N$-variate Boolean functions $f_N$ with either of the following properties:
  \begin{enumerate}[label=(\arabic*)]
      \item  $f_N$ is computable by monotone circuits of size $s=\poly(N)$, but any monotone (real) 
      circuit computing $f_N$ of size at most 
      $s^{1.4}$ must have depth at least $N^{2.4}$. \label{example:1-ckt} 
      \item $f_N$ is computable by monotone circuits of size 
      $s=\mathrm{quasipoly}(N)$, but any monotone  (real) circuit computing $f_N$ 
      of size at most $s\cdot \exp\left((\log N)^{1.9}\right)$ must have depth super-polynomial in $N$.
     \label{example:2-ckt} 
  \end{enumerate}
\end{restatable}

The functions we present that exhibit this %
behavior
actually come from new supercritical trade-offs in 
the neighboring field of proof complexity.

\subsection{Supercritical Trade-offs in Proof Complexity}

The \emph{Tseitin formulas}---unsatisfiable systems of $\bmod \, 2$ linear equations---were used 
to prove the first proof complexity lower bounds in \cite{Tseitin68ComplexityTranslated}. Since then, 
these formulas 
have played a central role in establishing lower bounds for, and 
understanding the reasoning power of %
proof systems; see \cite{FlemingP21} for a 
survey.
A notable exception is the %
\emph{cutting planes} proof system,
which captures the types of reasoning achievable when memory is limited to halfspaces.
The original paper on cutting planes~\cite{Chvatal73EdmondPolytopes} conjectured 
that the Tseitin formulas were %
hard to prove in cutting planes, 
and this was reiterated %
in~\cite{BP98Propositional, Jukna2012boolean}.
While lower bounds on the size of cutting planes proofs have appeared for a variety of 
formulas~\cite{Pudlak97LowerBounds, BPR95LowerBoundsCP,%
FPPR21RandomCNFS,HP17RandomCNFs,Sokolov24RandomCNFs,GGKS20MonotoneCircuit}, 
this conjecture remained open.

In a surprising turn of events, 
Dadush and Tiwari \cite{DT20CplxBranchingProofs} %
exhibited short
(quasi-polynomial size) cutting planes proofs of the Tseitin formulas. Notably, these 
proofs also have quasi-polynomial depth, far exceeding the linear worst-case upper bound.
This raised the question of whether the depth of \emph{any} small cutting planes 
proof of the Tseitin formulas must be supercritical \cite{BFIKPPR18Stabbing, FGIPRTW21BranchAndCut, 
fleming2022extremely}, which would give a partial explanation as to 
why these proofs took so long to find. 

Progress on this question was made
in~\cite{BGHMP06RankBounds, FGIPRTW21BranchAndCut}, 
by showing that any cutting planes proof of the 
Tseitin formulas on $n$ variables requires depth $\Omega(n)$;
and in~\cite{fleming2022extremely, buss2024simple}, 
by constructing families of CNF formulas which exhibit 
supercritical size-depth trade-offs for  cutting planes. 
The latter result is somewhat unsatisfactory, however, as the trade-off is supercritical only in the 
number of variables and not in the size of the formula.
This differs from the upper bound %
in~\cite{DT20CplxBranchingProofs}, which is supercritical in terms of the 
formula size as well.
We refer to trade-offs that are supercritical in the input size---rather than in %
the number of variables---as \emph{truly} supercritical.
In this work, we give the first size-depth trade-offs for cutting planes that are truly supercritical.

\begin{restatable}[Cutting planes trade-offs]{restatabletheorem}{scptradeoff}\label{thm:CPtrade-off} 
There are $3$-CNF formulas $F_N$ of size $S(F_N)$ over $N$ variables with either of the following properties:
\begin{enumerate}[label=(\arabic*)]
	\item Resolution refutes $F_N$ in size $S(F_N\vdash\perp)=\poly(N)$, but any cutting planes refutation of 
	size at most %
  $S(F_N\vdash\perp)^{1.4}$ has depth at least $S(F_N)^{2.4}$. \label{example:1-cp} 
	\item Resolution refutes $F_N$ in size $S(F_N\vdash\perp)=\mathrm{quasi}\text{-}\mathrm{poly}(N)$, but 
	any cutting planes refutation of size at most $S(F_N\vdash\perp)\cdot \exp\bigl((\log N)^{1.9}\bigr)$ has 
	depth at least super-polynomial in $N$. \label{example:2-cp} 
\end{enumerate}
\end{restatable}

As cutting planes simulates resolution, this also implies the first truly
supercritical size-depth trade-offs for resolution. 
In fact, most trade-offs in proof complexity so far 
are not truly supercritical (with the exception 
of~\cite{BBI16TimeSpace,Berkholz12ComplexityNarrowProofs,BNT13SomeTradeoffs}). 
In this work, we also obtain 
truly supercritical trade-offs for other combinations of %
complexity measures, 
which we state next. 

The %
seminal work~\cite{Razborov16NewKind} %
provides formulas for which any low-width treelike resolution proof must have size that is 
\emph{doubly-exponential} in the number of variables. Again, the lower bound is not 
supercritical in terms of the formula size.
We establish a truly supercritical width-size trade-off for treelike resolution. %

\begin{theorem}[Width-size trade-offs]\label{thm:treelike_trade-off_concrete} There are CNF formulas $F_N$ of size 
$S(F_N)=\poly(N)$ over $N$ variables with either of the following properties:
\begin{enumerate}[label=(\arabic*)]
	\item  Resolution refutes $F_N$ in width $W(F_N\vdash\perp)=o\left( \log N\right)$, but any 
	treelike refutation of width at most %
	$1.4 W(F_N\vdash\perp)$ has size at least $\exp\left(S(F_N)^{2.4}\right)$.
   \label{example:1-width-size} 

	\item Resolution refutes $F_N$ in width $W(F_N\vdash\perp)=o\left( (\log N)^{3/2} \right)$, but 
	any treelike resolution refutation of width at most $W(F_n\vdash\perp) + 40\frac{\log N}{\log\log 
	N}$ has size at least $\exp\left( S(F_N)^{\omega(1)} \right)$. \label{example:2-width-size} 
\end{enumerate}
\end{theorem}

Underlying each of these results is the first truly supercritical width-depth 
trade-off that is non-trivially robust; the aforementioned results then follow by applying several (new or improved) lifting theorems. Prior to our work, the only truly supercritical trade-off for width versus depth was due to Berkholz~\cite{Berkholz12ComplexityNarrowProofs}; however, this trade-off 
has no robustness---it holds only for the minimum width---and it therefore cannot be used
to obtain other supercritical trade-offs.

\begin{restatable}[Width-depth 
trade-offs]{restatabletheorem}{depthwidthinformal}\label{thm:tradeoff_informal}
  For any constants $C$ and $\delta\in(0,1)$, there are $4$-CNF formulas $\formf_N$ of size 
  $S(\formf_N)=\bigtheta{N}$ over $N$ variables which have resolution refutations of width 
  $w = \floor{\frac{n}{2\ln n}}+3$, with either of the following properties:
  \begin{enumerate}[label=(\arabic*)]
  \item \label{example:1} 
  $N=\poly(n)$ and any refutation of width at most $w+C$ has depth
  exponential in $\poly(S(\formf_N))$. 

  \item\label{example:2} 
  $N=o(2^{n/2})$ and any refutation of width at most 
  $(1+\delta)w$ has depth \superlinear in $S(\formf_N)$. 
  \end{enumerate}
\end{restatable}

\subsection{Trade-offs for the Weisfeiler--Leman Algorithm}

Surprisingly, all of the results above are obtained by studying the well-known \emph{Weisfeiler--Leman algorithm} for 
classifying
graphs and, more generally, relational structures. 
This algorithm 
appears as a subroutine in Babai's celebrated graph isomorphism result~\cite{Babai16GraphIsomorphism}, and has also been connected to machine learning \cite{Morris19WLgoNeural,Grohe21LogicGNN,Morris23WLgoML} and many other areas \cite{kiefer2020weisfeiler,GLNS23CompressingCFI}. 
The $1$\nobreakdash-dimensional version of the algorithm applied to graphs, known as \emph{color refinement}, starts by coloring all vertices according to their degree. This coloring is then iteratively refined by distinguishing vertices if their multisets of neighborhood colors differ. The process stops when a \emph{stable} coloring is reached, i.e., no further pair of vertices of the same color gets different colors. 
The \emph{$k$-dimensional} version of the algorithm ($k$-WL) 
instead performs 
colorings of $k$-tuples of vertices, or of elements in more general relational structures. Another parameter of interest is the \emph{iteration number}, which is the number of refinement steps until the coloring stabilizes.

It is easy to see that the iteration number of $k$-WL is at most $n^k-1$, and this can be slightly improved~\cite{KS19UpperBounds,LPS19WalkRefinement,grohe2023iteration}.
For a long time, the best lower bound was linear \cite{Furer01WeisfeilerLeman} until the sequence of works \cite{BN23QuantifierDepth,grohe2023iteration}
showed that $n^{\Omega(k)}$ iterations can be necessary. 
These results are actually slightly stronger in that they 
provide robust trade-offs between dimension and iteration number,
but they
only hold for relational structures of much higher arity than graphs.
A stronger $\Omega(n^{k/2})$ lower bound was finally proven in \cite{GLNS23CompressingCFI} 
for pairs of graphs distinguishable in dimension $k$,
but the authors left it as an open problem to turn this into a robust trade-off. 
Such a result is the starting point of our work.

\begin{restatable}[Weisfeiler--Leman trade-offs]{restatabletheorem}{wlinformal}\label{thm:wl_informal}
  For all $c,k$ with $1 \leq c \leq k-1$, and for $n$ large enough, 
  there are pairs of graphs of size $n$ that 
  can be distinguished by $k$-dimensional Weisfeiler--Leman, but for which even~$(k+c-1)$-dimensional Weisfeiler--Leman requires $\Bigomega{n^{k/(c+1)}}$ 
  iterations.
\end{restatable}

By the well-known equivalence between 
Weisfeiler--Leman and %
fragments of first order logic with counting~\cite{CFI92OptimalLowerBound},
our result also implies trade-offs between %
variable number and quantifier depth there. %

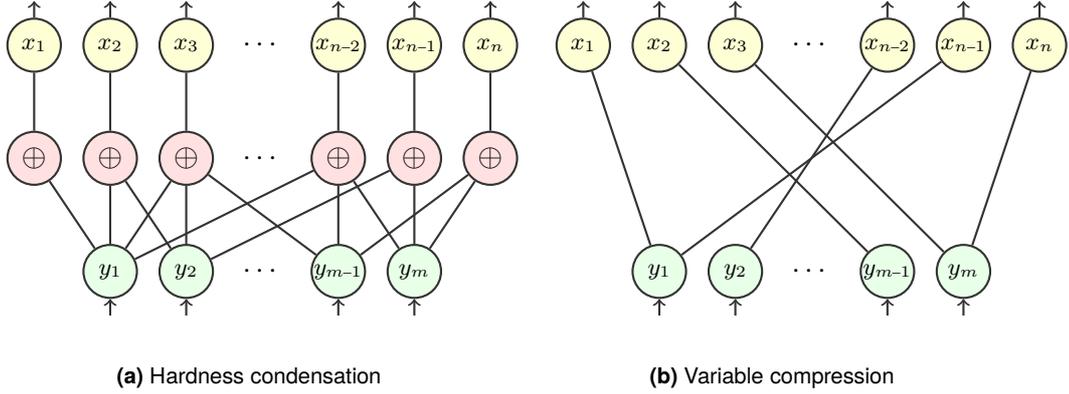
\begin{figure}
\centering
    \begin{subfigure}[b]{0.4\textwidth}
    \centering
    \begin{tikzpicture}
    	\node[circle, draw=black!80, fill=green!9, minimum size=20pt, thick] at (0,0)(l11){\small \phantom{a}};
    	
    	\node[circle, draw=black!80, fill=green!9, minimum size=20pt, thick] at (1,0)(l12){\small \phantom{a}};
    	
    	\node[circle, draw=black!80, fill=green!9, minimum size=20pt, thick] at (3,0)(l13){\small \phantom{a}};
    	
    	\node[circle, draw=black!80, fill=green!9, minimum size=20pt, thick] at (4,0)(l14){\small \phantom{a}};
    	
    	\node[circle, draw=black!80, fill=red!12, minimum size=20pt, thick] at (-1,1.5)(l21){\small \phantom{a}};
    	
    	\node[circle, draw=black!80, fill=red!12, minimum size=20pt, thick] at (0,1.5)(l22){\small \phantom{a}};
    	
    	\node[circle, draw=black!80, fill=red!12, minimum size=20pt, thick] at (1,1.5)(l23){\small \phantom{a}};
    	
    	\node[circle, draw=black!80, fill=red!12, minimum size=20pt, thick] at (3,1.5)(l24){\small \phantom{a}};
    	
    	\node[circle, draw=black!80, fill=red!12, minimum size=20pt, thick] at (4,1.5)(l25){\small \phantom{a}};
    	
    	\node[circle, draw=black!80, fill=red!12, minimum size=20pt, thick] at (5,1.5)(l26){\small \phantom{a}};

    	\node[circle, draw=black!80, fill=betterYellow!16, minimum size=20pt, thick] at (-1,3)(l31){\small \phantom{a}};
    	
    	\node[circle, draw=black!80, fill=betterYellow!16, minimum size=20pt, thick] at (0,3)(l32){\small \phantom{a}};
    	
    	\node[circle, draw=black!80, fill=betterYellow!16, minimum size=20pt, thick] at (1,3)(l33){\small \phantom{a}};
    	
    	\node[circle, draw=black!80, fill=betterYellow!16, minimum size=20pt, thick] at (3,3)(l34){\small \phantom{a}};
    	
    	\node[circle, draw=black!80, fill=betterYellow!16, minimum size=20pt, thick] at (4,3)(l35){\small \phantom{a}};
    	
    	\node[circle, draw=black!80, fill=betterYellow!16, minimum size=20pt, thick] at (5,3)(l36){\small \phantom{a}};
    	
    	\node at (0,0)(y1){\footnotesize $y_1$};
    	\node at (1,0)(y2){\footnotesize $y_2$};
    	\node at (3,0)(y3){\footnotesize $y_{m\text{--}1}$};
    	\node at (4,0)(y4){\footnotesize $y_m$};
    	
    	\node at (2,0)(dot1){$\cdots$};
    	\node at (2,1.5)(dot2){$\cdots$};
    	\node at (2,3)(dot3){$\cdots$};
    	
    	\node at (-1,1.5)(xor1){\large $\oplus$};
    	\node at (0,1.5)(xor2){\large $\oplus$};
    	\node at (1,1.5)(xor3){\large $\oplus$};
    	\node at (3,1.5)(xor4){\large $\oplus$};
    	\node at (4,1.5)(xor5){\large $\oplus$};
    	\node at (5,1.5)(xor6){\large $\oplus$};
    	
    	\node at (0,-0.8)(v1){\small \phantom{a}};
    	\node at (1,-0.8)(v2){\small \phantom{a}};
    	\node at (3,-0.8)(v3){\small \phantom{a}};
    	\node at (4,-0.8)(v4){\small \phantom{a}};
    	
    	\node at (-1,3)(x1){\footnotesize $x_1$};
    	\node at (0,3)(x2){\footnotesize $x_2$};
    	\node at (1,3)(x3){\footnotesize $x_3$};
    	\node at (3,3)(x4){\footnotesize $x_{n\text{--}2}$};
    	\node at (4,3)(x5){\footnotesize $x_{n\text{--}1}$};
    	\node at (5,3)(x6){\footnotesize $x_n$};
    	
    	\node at (-1,3.8)(o1){\small \phantom{a}};
    	\node at (0,3.8)(o2){\small \phantom{a}};
    	\node at (1,3.8)(o3){\small \phantom{a}};
    	\node at (3,3.8)(o4){\small \phantom{a}};
    	\node at (4,3.8)(o5){\small \phantom{a}};
    	\node at (5,3.8)(o6){\small \phantom{a}};
    	
    	\draw[thick, color=black!80] (l11) -- (l21);
    	\draw[thick, color=black!80] (l11) -- (l22);
    	\draw[thick, color=black!80] (l11) -- (l23);
    	\draw[thick, color=black!80] (l11) -- (l24);
    	
    	\draw[thick, color=black!80] (l12) -- (l23);
    	\draw[thick, color=black!80] (l12) -- (l22);
    	\draw[thick, color=black!80] (l12) -- (l25);
    	
    	\draw[thick, color=black!80] (l13) -- (l23);
    	\draw[thick, color=black!80] (l13) -- (l24);
    	\draw[thick, color=black!80] (l13) -- (l26);     	
    	\draw[thick, color=black!80] (l14) -- (l24);
    	\draw[thick, color=black!80] (l14) -- (l25);
    	\draw[thick, color=black!80] (l14) -- (l26);
    	
    	\draw[thick, color=black!80] (l21) -- (l31);
    	\draw[thick, color=black!80] (l22) -- (l32);
    	\draw[thick, color=black!80] (l23) -- (l33);
    	\draw[thick, color=black!80] (l24) -- (l34);
    	\draw[thick, color=black!80] (l25) -- (l35);
    	\draw[thick, color=black!80] (l26) -- (l36);
    	
    	\draw[thick, color=black!80,->] (v1) -- (l11);
    	\draw[thick, color=black!80,->] (v2) -- (l12);
    	\draw[thick, color=black!80,->] (v3) -- (l13);
    	\draw[thick, color=black!80,->] (v4) -- (l14);
    	
    	\draw[thick, color=black!80,->] (l31) -- (o1);
    	\draw[thick, color=black!80,->] (l32) -- (o2);
    	\draw[thick, color=black!80,->] (l33) -- (o3);
    	\draw[thick, color=black!80,->] (l34) -- (o4);
    	\draw[thick, color=black!80,->] (l35) -- (o5);
    	\draw[thick, color=black!80,->] (l36) -- (o6);

    \end{tikzpicture}

    \caption{\label{fig:xor}Hardness condensation}
    \end{subfigure}
\quad
    \begin{subfigure}[b]{0.4\textwidth}
    \begin{tikzpicture}
    	\node at (-1.5,0)(sep){\small \phantom{a}};
    
    	\node[circle, draw=black!80, fill=green!9, minimum size=20pt, thick] at (0,0)(l11){\small \phantom{a}};
    	
    	\node[circle, draw=black!80, fill=green!9, minimum size=20pt, thick] at (1,0)(l12){\small \phantom{a}};
    	
    	\node[circle, draw=black!80, fill=green!9, minimum size=20pt, thick] at (3,0)(l13){\small \phantom{a}};
    	
    	\node[circle, draw=black!80, fill=green!9, minimum size=20pt, thick] at (4,0)(l14){\small \phantom{a}};

    	\node[circle, draw=black!80, fill=betterYellow!16, minimum size=20pt, thick] at (-1,3)(l31){\small \phantom{a}};
    	
    	\node[circle, draw=black!80, fill=betterYellow!16, minimum size=20pt, thick] at (0,3)(l32){\small \phantom{a}};
    	
    	\node[circle, draw=black!80, fill=betterYellow!16, minimum size=20pt, thick] at (1,3)(l33){\small \phantom{a}};
    	
    	\node[circle, draw=black!80, fill=betterYellow!16, minimum size=20pt, thick] at (3,3)(l34){\small \phantom{a}};
    	
    	\node[circle, draw=black!80, fill=betterYellow!16, minimum size=20pt, thick] at (4,3)(l35){\small \phantom{a}};
    	
    	\node[circle, draw=black!80, fill=betterYellow!16, minimum size=20pt, thick] at (5,3)(l36){\small \phantom{a}};
    	
    	\node at (0,0)(y1){\footnotesize $y_1$};
    	\node at (1,0)(y2){\footnotesize $y_2$};
    	\node at (3,0)(y3){\footnotesize $y_{m\text{--}1}$};
    	\node at (4,0)(y4){\footnotesize $y_m$};
    	
    	\node at (2,0)(dot1){$\cdots$};
    	\node at (2,3)(dot3){$\cdots$};

    	\node at (0,-0.8)(v1){\small \phantom{a}};
    	\node at (1,-0.8)(v2){\small \phantom{a}};
    	\node at (3,-0.8)(v3){\small \phantom{a}};
    	\node at (4,-0.8)(v4){\small \phantom{a}};
    	
    	\node at (-1,3)(x1){\footnotesize $x_1$};
    	\node at (0,3)(x2){\footnotesize $x_2$};
    	\node at (1,3)(x3){\footnotesize $x_3$};
    	\node at (3,3)(x4){\footnotesize $x_{n\text{--}2}$};
    	\node at (4,3)(x5){\footnotesize $x_{n\text{--}1}$};
    	\node at (5,3)(x6){\footnotesize $x_n$};
    	
    	\node at (-1,3.8)(o1){\small \phantom{a}};
    	\node at (0,3.8)(o2){\small \phantom{a}};
    	\node at (1,3.8)(o3){\small \phantom{a}};
    	\node at (3,3.8)(o4){\small \phantom{a}};
    	\node at (4,3.8)(o5){\small \phantom{a}};
    	\node at (5,3.8)(o6){\small \phantom{a}};
    	
    	\draw[thick, color=black!80] (l11) -- (l31);
    	\draw[thick, color=black!80] (l11) -- (l35);
    	\draw[thick, color=black!80] (l12) -- (l34);
    	\draw[thick, color=black!80] (l13) -- (l32);
    	\draw[thick, color=black!80] (l14) -- (l33);
    	\draw[thick, color=black!80] (l14) -- (l36);
    	
    	\draw[thick, color=black!80,->] (v1) -- (l11);
    	\draw[thick, color=black!80,->] (v2) -- (l12);
    	\draw[thick, color=black!80,->] (v3) -- (l13);
    	\draw[thick, color=black!80,->] (v4) -- (l14);
    	
    	\draw[thick, color=black!80,->] (l31) -- (o1);
    	\draw[thick, color=black!80,->] (l32) -- (o2);
    	\draw[thick, color=black!80,->] (l33) -- (o3);
    	\draw[thick, color=black!80,->] (l34) -- (o4);
    	\draw[thick, color=black!80,->] (l35) -- (o5);
    	\draw[thick, color=black!80,->] (l36) -- (o6);
    \end{tikzpicture}
    \centering
    \caption{\label{fig:ident}Variable compression}
    \end{subfigure}
\caption{
  Hardness condensation in \autoref{fig:xor} substitutes the $x$\nobreakdash-variables with an XOR over some $y$\nobreakdash-variables, while variable compression in \autoref{fig:ident} substitutes with $y$-variables directly. Note that $m\ll n$.
  }
\end{figure}

\subsection{Techniques}\label{subsec:techniques}

Most of the previously known supercritical trade-offs are based on \emph{hardness condensation}~\cite{Razborov16NewKind}, %
which works by substituting the variables of a problem instance with XOR gadgets over a much smaller set of variables (cf.\ \autoref{fig:xor}),
and then showing
that the substituted instance 
remains essentially as hard, although the number of variables has decreased substantially.
This technique transferred from proof complexity to finite model theory in~\cite{BN23QuantifierDepth} to prove the Weisfeiler--Leman trade-offs discussed above. %

The recent result %
~\cite{GLNS23CompressingCFI} instead relies on a new technique of \emph{graph compression},
where vertices %
are identified via an equivalence relation,
together with the standard approach of analyzing Weisfeiler--Leman via the \emph{Cop-Robber 
game}~\cite{seymour1993graph}.
Here, dimension corresponds to number of Cops in play, and iteration number to (game-)rounds. 
Lower bounds for Weisfeiler--Leman %
follow from strong Robber strategies, and the bounds become
supercritical when these strategies continue to work even when 
the game is played on the compressed graph. 
In proof complexity, the number of Cops and rounds approximately correspond to resolution width and depth for Tseitin formulas \cite{galesi2020cops}. Because the correspondence is not exact,~\cite{GLNS23CompressingCFI} 
does not give proof complexity results.
Using a %
refined graph compression {and analysis}, we obtain (\autoref{thm:wl_informal}) Weisfeiler--Leman trade-offs which are robust. 
Thanks to this robustness, we are able to translate %
these results into truly supercritical width-depth trade-offs for resolution,
exporting the technique of~\cite{GLNS23CompressingCFI} to proof complexity, as advocated in that 
paper.
In contrast to hardness condensation, the resulting \emph{compressed} Tseitin formula
is obtained by substituting each variable with one of the new variables in a structured way (see \autoref{fig:ident}).
We believe that the tool of \emph{variable compression}, interesting in its own right, may find more applications in proof complexity.
The remaining trade-offs in this paper are obtained
by proving and applying new \emph{lifting theorems}
to our width-depth trade-off. 
Lifting is a framework for deriving lower bounds for stronger computation models from those for weak ones.
In particular, the lifting theorems of \cite{GGKS20MonotoneCircuit, lmmpz2022liftingWithSunflowers} convert lower bounds on resolution width to lower bounds on the 
size of monotone (real) circuits, which in turn imply lower bounds for cutting planes.
However, 
the parameters of these theorems are insufficient to obtain supercritical trade-offs from \autoref{thm:tradeoff_informal}. 
We therefore establish an improved,  
\emph{tight} lifting theorem 
for both monotone circuits and cutting planes. %
The key to the proof is a new 
way of %
approximating %
a combinatorial triangle by structured rectangles, from which we can extract clauses.
We also provide an even tighter lifting for resolution size, which has a simple proof based on %
random restriction. %
Lastly, we prove a %
lifting theorem for treelike resolution that turns a depth lower bound into a size lower bound 
and simultaneously increases the width. %
We believe that these lifting results should be of independent interest.

\subsection{Related Work}
\label{sec:relatedwork}

In concurrent work, G\"{o}\"{o}s, Maystre, Risse and Sokolov \cite{GMRS2024Supercritical} report
supercritical size-depth trade-offs for monotone circuits, resolution and cutting planes.
Their approach is similar in that they also start with a truly supercritical width-depth trade-off
and apply lifting to obtain size-depth trade-offs,
but their width-depth trade-off is very 
different from ours, and relies on a novel, interesting formula construction.

In terms of parameters,
their formulas have resolution proofs in
width $\bigoh{\log n}$, 
but any proof in width up to $n^{\epsilon}$ %
has supercritical depth,
making
their width-depth trade-off extremely robust. This robustness allows them to apply existing 
lifting theorems as a black box
to obtain functions that are computable by monotone circuits of size $n^{\bigoh{\log n}}$
but where any monotone circuit of polynomial depth has exponential size.
While our results are not nearly as robust, 
we obtain a blow-up in size even for circuits with depth polynomial in the size upper bound,
and our proof complexity trade-offs apply for constant-width proofs.
In this sense, %
{the two are incomparable.} 
In addition, %
{we give results for the Weisfeiler--Leman algorithm and prove tight lifting theorems.}

Our Weisfeiler--Leman and resolution width-depth trade-offs were announced at the Oberwolfach workshop \emph{Proof Complexity and Beyond} in March 2024. 
Building on %
that work, Berkholz, Lichter and Vinall-Smeeth~\cite{BerkholzLV2024} also obtained a truly 
supercritical width-size trade-off for treelike resolution. 
Our treelike resolution trade-off came %
afterward, with a different technique and improved parameters.

\subsection{Organisation of This Paper}

This rest of this paper is structured as follows. The preliminaries and a proof overview
are presented in \autoref{sec:prelim}. 
In \autoref{sec:construction}, we define the compressed Cop-Robber game
and show how it relates to resolution. 
In \autoref{sec:lbd}, we prove the round lower bound on the game. 
We then prove our lifting theorem for treelike resolution in \autoref{sec-WidthVsTreesize}
and for resolution in \autoref{sec-SizeVsDepth}.
In \autoref{sec-MonSizeVsDepth} we prove our tight general lifting theorem and
finally, in \autoref{sec:conclusion}, we 
discuss some open problems.

\section{Preliminaries and Proof Overview}
\label{sec:prelim}

In this section, we present an overview of the components needed to obtain our trade-off results
stated in \autoref{sec:intro} and explain how they fit together. We begin with some notation and general definitions.

For $a,b\in\Nplus$ with $a\leq b$, we use the notation $[a,b]\coloneqq\set{a, a+1, \dots , b}$, and $[a]=\set{1, \dots, a}$. 
Given $k\in \Nplus$ and $a,b\in [k]$ with $a>b$, we write $[a,b]\coloneqq\set{a, a+1,\dots, k}\cup \set{1, 2, \dots, b}$. We call sets $[a,b]$, where $a,b\in[k]$, \emph{cyclic intervals modulo $k$}. In this paper, $\log(\cdot)$ are base 2 and $\ln(\cdot)$ are base~$e$.

All graphs in this paper are simple. 
For a graph $G=(V,E)$ and a vertex subset $W\subseteq V$, $G|_W$ denotes the induced subgraph 
on $W$. Given $F\subseteq E$, we write $V(F)$ for the set of all vertices incident to an edge in $F$.
By a path in a graph, we always mean \emph{simple} paths, i.e., a sequence of distinct vertices where consecutive ones are connected by an edge.

We begin by introducing the Cop-Robber game~\cite{seymour1993graph}, which underpins the proof of our width-depth and Weisfeiler--Leman trade-offs.

\subsection{The Cop-Robber Game} 
We describe the %
rect\emph{$k$-Cop-Robber game}. The Cops and the Robber stay on vertices of a graph $G$ and can see each other. Initially, the Robber is at a %
vertex, and all $k$ Cops are lifted from the graph (in a helicopter).
A game round unfolds as follows:
\begin{enumerate}[label=(\arabic*)]
	\item If there is no lifted Cop, the Cops choose and lift one. Then they signal a 
	vertex $v$ to the 
	Robber.
	\item The Robber uses a path in $G$ to move from his position $w_1$ to another 
	vertex 
	$w_2$ while avoiding the Cops on the graph.
	\item A lifted Cop lands at the signaled vertex $v$.
\end{enumerate}
\noindent The game ends when a Cop lands at the Robber's position.

In \autoref{sec:construction} and \autoref{sec:lbd}, 
we analyze a variant called the \emph{compressed Cop-Robber game} \cite{GLNS23CompressingCFI}. 
It is played on a graph added with 
equivalence relations on vertices and edges, which are obtained as follows.
First,  
a vertex equivalence relation is
chosen.
For equivalent vertices, the incident edges are then identified one by one, and we take the transitive closure of this identification to get an equivalence relation.
In the specific instance we analyze, the graph is a \emph{cylinder} (a grid where columns are cycles) with $k$ rows and roughly $n^k$ columns %
where vertices on the same row are identified periodically,
using a different period in every row.
We %
generalize the construction in~\cite{GLNS23CompressingCFI} by selecting the row periods based on a parameter $c\leq k-1$%
; %
see \autoref{def:compression} and \autoref{def:parameters}. 
The original construction corresponds to~$c=2$.  

The rules in the compressed game are subtler (see \autoref{def:cgame}). In particular, the Robber must avoid \emph{all} vertices equivalent to those occupied by Cops. 
On the uncompressed cylinder, %
$k+1$ Cops have an obvious strategy: 
block off the middle of the graph---forming a police cordon of sorts---and then march towards the robber in lockstep. 
With more Cops, and with the compression providing Cop copies %
on the equivalent vertices, they can potentially do better. %
Despite that, we prove the following theorem.
\begin{restatable}[Cop-Robber]{restatabletheorem}{coprobinformal}
    \label{thm:main}
    For any parameters $k=k(n)$ and $c=c(n)$ where $1\leq c\leq k-1$ and $2\leq k<{n/(2\ln n)}$, there are degree-4 graphs $\set{G_n}$ and a compressed Cop-Robber game on $G_n$ where $k+1$ Cops can win, but the Robber can survive %
    $\bigomega{n^{k}}$ rounds against $k+c$ Cops.
  \end{restatable}

The novelty of our analysis, in comparison to \cite{GLNS23CompressingCFI}, lies in having the 
Robber play against a virtual stronger opponent, %
formalized in %
the notion of \emph{virtual 
cordons} associated with the Cops (\autoref{def:virtual}). 
The transition of (the set of) virtual 
cordons 
over a game round is easier to analyze (see, for example, \autoref{lem:key}), which allows us to 
show a strategy for the Robber to survive for $\Bigomega{n^k}$ rounds.

\subsection{The Weisfeiler--Leman Algorithm}\label{sec:prelim_wl}
We define the Weisfeiler--Leman algorithm on graphs; see the survey \cite{kiefer2020weisfeiler} for further explanations.
A graph $G=(V,E,c)$ where $c\colon V \rightarrow \mathbb{N}$ is \emph{vertex colored}. Given $k\geq 2$ and a vertex colored graph $G=(V,E,c)$, 
the \emph{$k$-dimensional Weisfeiler--Leman algorithm} \cite{weisfeiler1968reduction,IL90GraphCanonization} iteratively refines a coloring of the $k$-tuples of vertices. We denote the coloring after the $i$th round by $\chi^{(i)}\colon V^k\rightarrow C$, where  $C$ is a finite set. %
In the initial round, the color $\chi^{(0)}(\vec{u})$ of a tuple $\vec{u}=(u_1,\dots,u_k)$ is its own 
isomorphism class, 
where we say $(u_1,\dots,u_k)$ is isomorphic to $(v_1,\dots,v_k)$ if the map $u_i\mapsto v_i$ 
preserves vertex colors and is an isomorphism between the induced subgraphs of the two tuples. 
We use $\vec{u}\ [v/u_j]$ to denote the $k$-tuple obtained by substituting $u_j$ with $v$ in $\vec{u}$, i.e., $(u_1,\dots, u_{j-1},v, u_{j+1}, \dots, u_k)$. 
In round $i$, the coloring $\chi^{(i)}(\vec{u})$ of a tuple $\vec{u}$ is obtained by appending a {\it multiset} of tuples to $\chi^{(i-1)}(\vec{u})$:
\begin{equation*}
\chi^{(i)}(\vec{u}):= \Big{(} \chi^{(i-1)}(\vec{u}),\ \MSET{ \Big{(} \chi^{(i-1)}(\vec{u}\ [v/u_1]),\ \dots,\ 
\chi^{(i-1)}(\vec{u}\ [v/u_k]) \Big{)} \mid v\in V(G)} \Big{)}.
\end{equation*}
The algorithm %
\emph{stabilizes} after round $t$ if any two tuples that have the same 
color in round~$t$, i.e., $\chi^{(t)}(\vec{u})=\chi^{(t)}(\vec{v})$, get the same color in round $t+1$, i.e., $\chi^{(t+1)}(\vec{u})=\chi^{(t+1)}(\vec{v})$. The minimum such $t$ is called the \emph{iteration number} on $G$.

The algorithm can be used to distinguish a pair of colored graphs $G,H$ by comparing the colorings $\chi^{(i)}(G)$ and $\chi^{(i)}(H)$.%
We say that $k$-dimensional Weisfeiler--Leman \emph{distinguishes $G$ and $H$ in $t$ rounds} if 
for some color $c$, the number of tuples that have color $c$ in $\chi^{(t)}(G)$ is different from the 
number of such tuples in $\chi^{(t)}(H)$. 

By applying standard translations (see \autoref{appendix}), 
\autoref{thm:main} gives the following trade-off for Weisfeiler--Leman 
algorithms, 
which is an explicit version of \autoref{thm:wl_informal}.

\begin{theorem}[Weisfeiler--Leman trade-offs, explicit]\label{thm:main_wl}
    For all %
    $c$ and $k$ with ${1 \leq c \leq k-1}$, 
    if $n$ is large enough, there are $n$-vertex 
    graph pairs distinguished by $k$-dimensional Weisfeiler--Leman, but for which $(k+c-1)$-dimensional Weisfeiler--Leman requires at least~$\bigl(2^{-(c+10)}k^{-3} n \bigr)^{k/(c+1)}$ iterations. 
\end{theorem}

Using %
the equivalence between the $k$-dimensional Weisfeiler--Leman algorithm and the $(k+1)$-variable fragment of first order 
logic with counting \cite[Theorem 5.2]{CFI92OptimalLowerBound}, %
\autoref{thm:main_wl} also implies a trade-off between the number of variables and quantifier depth. Namely, there is a graph pair distinguishable in the {${(k+1)}$-variable} fragment of first order logic with counting, but a lower bound of $\bigl(2^{-(c+10)}k^{-3} n \bigr)^{k/(c+1)}$ on the quantifier depth applies up to the $(k+c)$-variable fragment.

\subsection{Proof Complexity Basics and Resolution}
Let us review some standard definitions from proof complexity. For a more comprehensive 
presentation of this material, see, e.g., 
\cite{Krajicek19ProofComplexity,BN21ProofCplxSAT}. A \emph{literal} is a Boolean variable 
$\varx$ or its negation~$\stdnot{\varx}$. 
It will sometimes be convenient to use the notations 
$\varx^1=\varx$ and $\varx^0=\stdnot{\varx}$. 
A \emph{clause} is a disjunction of literals $\cld = x_1 \lor \formuladots \lor x_\clwidth$, 
which we require to be over pairwise disjoint variables. 
We call the number of literals appearing in a clause~$\cld$ the \emph{width} $\widthofarg{\cld}$ 
of~$\cld$. 
We call a clause of width at most $\clwidth$ a \emph{$\clwidth$-clause}. 
A CNF $\formf = \cld_1 \land \formuladots\land \cld_\nclause$ is a conjunction of clauses, the 
\emph{formula width} $W(F)$ is the maximal width of clauses in $F$, the \emph{clause size} 
$\clausesizeofarg{\formf}$ is the number of clauses in $F$ (viewed as a set of clauses), and \emph{formula size} $S(F)$ is the 
sum of width over the clauses in $F$. 
We call $F$ a \emph{\xcnf{\clwidth}} if all clauses are $k$-clauses. 
We denote by $\vars{\formf}$ the set of variables appearing in a formula $\formf$. 

A \emph{resolution refutation} $\refof{\refpi}{\formf}$ of an unsatisfiable CNF formula $\formf$ is an 
ordered sequence of clauses $\refpi = (\cld_1, \dots, \cld_\szs)$, where $\cld_\szs$ is the empty 
clause containing no literals denoted by $\emptycl$, and each $\cld_i$ is a clause in $\formf$, or 
derived from some specified $\cld_j$ and $\cld_k$, where $j,k<i$, using the \emph{resolution rule}
\begin{equation}
  \AxiomC{$\clc \lor \varx$}
  \AxiomC{$\cld \lor \stdnot{\varx}$}
  \BinaryInfC{$\clc \lor \cld$}
  \DisplayProof  \eqperiod
\end{equation}
We associate a DAG $G_\refpi$ with every resolution refutation $\refpi$ as follows. There is a vertex $v_i\in\vertices{G_\refpi}$ for every $i \in [s]$, and directed edges $(v_j,v_i),(v_k,v_i)\in \edges{G_\refpi}$ if and only if $\cld_i$ was derived from $\cld_j$ and $\cld_k$ by resolution. 

The \emph{size} (or \emph{length}) $\size{\refpi}$ of a refutation~$\refpi$ is the number of clauses $s$ in it. By width $\widthofarg{\refpi}$ of a refutation~$\refpi$, we mean the width of the largest clause in $\refpi$.
Lastly, the \emph{depth} $\depthofarg{\refpi}$ of a refutation $\refpi$ is the number of edges in the longest path in its associated DAG $G_\refpi$.
We also consider the above measures for refuting a CNF formula $\formf$, by taking the minimum over all refutations of~$\formf$. That is, $S(F\vdash\perp):=\minofexpr[\refof{\refpi}{\formf}]{\size{\refpi}}$, 
$W(F\vdash\perp):=\minofexpr[\refof{\refpi}{\formf}]{\widthofarg{\refpi}}$, 
and $D(F\vdash\perp):=\minofexpr[\refof{\refpi}{\formf}]{\depthofarg{\refpi}}$ are the size, width, 
and depth of refuting~$\formf$, respectively.

\subsection{Supercritical Width-Depth Trade-off for Resolution}

Our %
{first} technical contribution is the following truly supercritical width-depth trade-off
for resolution.

\begin{theorem}[Width-depth trade-offs, \general]\label{thm:main_tradeoff}
\label{thm:main_trade-off}
Let $k=k(n)$, $c=c(n)$ be any integer parameters such that $3\leq c\leq k-1<\frac{n}{2\ln n}$. 
Then for all $n$, there is a linear-size 4-CNF formula~$\formf$ 
with between $2k^2n^{c+1}$ and $40k^2(2n)^{c+1}$ variables, 
which has a resolution refutation of width $k+3$ and size $O(k^2(4n)^k)$, but for which any refutation of width at most $k+c$ has depth at least 
$\bigomega{n^{k}}$. 
\end{theorem}

This theorem will be proven through the connection to the compressed Cop-Robber game which we 
make formal in \autoref{sec:construction}. 
The formula $\formf_N$ in the theorem is a \emph{\TSEITINFORM} 
\cite{Tseitin68ComplexityTranslated} after a variable projection operation, also 
defined in the next section. 
The \TSEITINFORM is defined for any simple graph~$G$ where each vertex $v\in\vertices{G}$ is 
labeled $0$ or $1$ so that the labels sum to an odd number. A vertex labeled $1$ is said to have an 
$\emph{odd charge}$. The formula has a variable $x_e$ for 
every edge $e\in \edges{G}$ and is defined to be the CNF containing, for all $v \in \vertices{G}$, the 
clauses expressing that the sum of the edge variables incident to $v$ has parity equal to the label of 
$v$.

The %
two examples in \autoref{thm:tradeoff_informal} follow from \autoref{thm:main_tradeoff} by taking 
$k(n)\coloneqq \floor{n/(2\log n)}$ and setting: 
(1) $c(n)$ to be a large constant, %
and 
(2) $c(n)$ to be %
$\floor{\frac{1+\delta}{2}k}$.

\subsection{Tight Lifting %
and Supercritical Trade-offs for Resolution}

The framework for obtaining our other proof and circuit complexity trade-offs
from the width-depth trade-off
is \emph{lifting} which is based on composition
with functions, which we refer to as \emph{gadgets}. 
For CNF formulas, there can be multiple ways of representing
its composition with a gadget as a CNF formula.
Therefore, 
for the gadgets~$g$
we are interested in, 
we will denote by $g(\formf)$ 
a specific CNF encoding of the composition of the CNF formula~$\formf$ 
with the gadget $g$.

In this paper, we consider two gadgets:
$\xorText_{\xorsize} : \set{0,1}^{\xorsize} \to \set{0,1}$,
defined as $\xorText_{\xorsize}(x_1, \ldots, x_{\xorsize}) = \bigoplus_{i \in [\xorsize]} x_i$,
and $\IND_m : [m] \times \set{0,1}^m \to \set{0,1}$,
defined as $\IND_m(x,y) = y_x$.
Given a CNF formula~$\formf$ over variables $x_1, \ldots, x_n$,
we denote by $\XOR{\formf}{\xorsize}$ the CNF formula obtained by 
substituting each $x_i$ by 
$y_{i,1}\oplus\cdots\oplus y_{i,\xorsize}$ where $y_{i,j}$ is a new propositional variable,
and then expanding it out in CNF. 
For instance, if $m=2$ then the clause $x_4\lor \stdnot{x}_5$ yields $4$ clauses: 
\begin{align}
y_{4,1}\lor y_{4,2} \lor y_{5,1}\lor \stdnot{y}_{5,2} \eqcomma
&\hspace{22pt}
y_{4,1}\lor \stdnot{y}_{4,2} \lor y_{5,1}\lor y_{5,2} \eqcomma \\
\stdnot{y}_{4,1} \lor y_{4,2}  \lor \stdnot{y}_{5,1}\lor \stdnot{y}_{5,2} \eqcomma
&\hspace{22pt}%
\stdnot{y}_{4,1} \lor y_{4,2} \lor y_{5,1}\lor  y_{5,2} \eqperiod 
\end{align}
Note that the width of $\XOR{\formf}{\xorsize}$ is  $\xorsize \cdot W(\formf)$ and 
the number of clauses is $\clausesizeofarg{\XOR{\formf}{\xorsize}}\leq 2^{(\xorsize-1)\cdot \widthofarg{\formf}} \clausesizeofarg{\formf}$.

Our lifting theorem for treelike resolution, which we prove in \autoref{sec-WidthVsTreesize},
uses composition with the $\xorText_{\xorsize}$ gadget.
Observe that the resolution refutation in its conclusion has small depth 
and simultaneously smaller width. This decrease in width 
is essential for obtaining our width-size trade-off.

\begin{restatable}[Lifting for treelike resolution]{restatabletheorem}{treelikeReslifting}
\label{thm:treelikeReslifting}
    Let $\formf$ be a CNF formula and let $\xorsize\geq 2$. If there is a width-$w$, size-$s$ treelike 
    resolution refutation for $\XOR{\formf}{\xorsize}$, then there is a 
    width-$\left({\frac{w}{\xorsize-1}}\right)$, depth-$\log s$ resolution refutation of $\formf$.
\end{restatable}

We can now apply this theorem to our width-depth trade-off to obtain
the supercritical trade-offs for treelike resolution.

\begin{theorem}[Width-size trade-offs, \general]\label{thm:treelike_trade-off}
    For any %
    $\xorsize=\xorsize(n)\geq 3$, $k=k(n)\in[4,\frac{n}{2\ln n}]$, and $\eps=\varepsilon(n)\in(\frac{4}{k},1-\frac{1}{k})$, 
    there are $4\xorsize$-CNF formulas $\formf_N$ 
    with $N$ variables and formula size $\bigoh{16^\xorsize \cdot N}$,
    where 
    $2k^2n^{\floor{\varepsilon k}} \xorsize \leq N \leq 40k^2(2n)^{\floor{\varepsilon k}} \xorsize$,
    which are refutable in width $\xorsize(k+3)$ resolution, 
    but for which any treelike refutation of width at most 
    $(\xorsize-1)(1+\varepsilon)k$ has size at least $2^{\Omega(n^k)}$.
\end{theorem}
\begin{proof}%
Let $F$ be the
4-CNF formulas from \autoref{thm:main_trade-off} with 
parameter $c:=\floor{\varepsilon k} -1 \in[3,k-1]$, and define $\formf_N:=\XOR{F}{\xorsize}$. 
Then 
$\sizeofarg{\formf_N} = \bigoh{  2^{4(\xorsize-1)} \cdot  4\xorsize \cdot
\sizeofarg{F}} = \bigoh{ 2^{4\xorsize} \cdot \xorsize \cdot
\setsize{\vars{F}}} = \bigoh{ 16^\xorsize \cdot N}$,
and since $F$ is refutable in width $k+3$, a line-by-line simulation via 
$x_i=y_{i,1}\oplus\ldots\oplus y_{i,\xorsize}$ gives a 
refutation of $\formf_N$ in width $\xorsize(k+3)$. Now suppose $\proofbigpi$ is a treelike 
refutation of $\formf_N$ in width 
$(\xorsize-1)(1+\varepsilon)k$ and size $s$, then by \autoref{thm:treelikelifting}, there is a 
refutation 
of $F$ in width 
$(1+\varepsilon)k$ and depth $\log s$.  
The theorem follows since \autoref{thm:main_trade-off} implies that 
$\log s=\Omega(n^k)$.
\end{proof}

Note that \autoref{thm:treelike_trade-off_concrete} follows immediately from 
\autoref{thm:treelike_trade-off} by taking 
$k(n)\coloneqq \floor{n/(2\log n)}$ and setting: 
(1) %
$\xorsize\coloneqq256$, $\varepsilon\coloneqq0.41$,
and
(2) $\xorsize\coloneqq\floor{\sqrt{n}}$, 
$\eps\coloneqq\frac{100}{\sqrt{n}}$. %
\smallskip

Now, as a warm up for the lifting theorems for monotone circuits and cutting planes
in \autoref{subsec:liftingforMCandCP}, we prove an even tighter %
result for resolution. 
For this theorem, we consider the following composition of a CNF formula
with the indexing gadget.\footnote{Other standard encodings work as well,
but this one ensures the formula width increase by at most a factor~$2$.}
Let~$\formf$ be a CNF formula over variables $z_1, \ldots, z_{n}$.
To obtain the CNF formula $\indexingnolength{\formf}$,
we start with substituting in $\formf$ every occurrence of $z_i$ by 
\begin{equation}
(x_{i,1} \rightarrow y_{i,1}) \land \ldots \land (x_{i,\indexsize} \rightarrow y_{i,\indexsize}) \eqcomma
\end{equation}
where $x_{i,j}$ and $y_{i,j}$ are new propositional variables, and we expand it out to CNF.
Moreover, we would like to include $x_{i,1} \lor \ldots \lor x_{i,\indexsize}$ for each $i$ to ensure that 
$x_{i,j} = 1$ for at least one $j\in[\indexsize]$; 
but to keep the width of the formula small, 
we instead use extension variables to encode each of these clauses as
a $3$-CNF formula with $\leq m$ clauses. 
Note that the width of $\indexingnolength{\formf}$ is  
$2 W(\formf)$
and 
the number of clauses is $\clausesizeofarg{\indexingnolength{\formf}} \leq 
\indexsize^{\widthofarg{\formf}} 
\clausesizeofarg{\formf} + n\indexsize$.
Using this gadget, we obtain our lifting theorem for resolution.
\begin{restatable}[Lifting for resolution]{restatabletheorem}{sizelifting}
\label{thm:size_lifting}
    For any $\indexsize, n\geq 1$ and $n$-variate CNF formula $\formf$, 
    if $\indexingnolength{\formf}$ has a 
    resolution refutation of size $\sizestd$ and depth $d$, 
    then $\formf$ has a resolution refutation of width $\widthbndsize$ and depth $d$. 
\end{restatable}

In this theorem, the size-width relation is nearly tight (see \autoref{lem:res-easydirection} 
below), and there is no increase in depth. 
Moreover, the theorem holds for any gadget size, and the proof, which we defer 
to \autoref{sec-SizeVsDepth}, is simple---based on a random restriction argument.

By a standard step-by-step simulation we obtain the following upper bound
for refuting $\indexingnolength{\formf}$. We include the proof 
for the sake of
completeness.

\begin{lemma}\label{lem:res-easydirection}
    For any $\indexsize, n\geq 1$ and $n$-variate CNF formula $\formf$, 
    if $\formf$ has a resolution refutation of width $w$ and size $s \geq n$,  
    then $\indexingnolength{\formf}$ has a 
    resolution refutation of size $\bigoh{s \cdot \indexsize^{w+1}}$.
\end{lemma}

\begin{proof}
    The proof is a standard step-by-step simulation.
    Let $\formf$ be a CNF formula over variables $z_1, \ldots, z_n$ and
    let $\Pi$ be a resolution refutation of $\formf$ in width $w$ and size $s$.
    We start by deriving $\bigvee_{j\in [m]}x_{i,j}$ for all $i\in [n]$ from the
    axioms in $\indexingnolength{\formf}$, which can be done in $\bigoh{nm}$ steps.
    We then simulate $\Pi$ step by step, keeping the invariant that for every %
    clause $C=\bigvee_{\ell\in [w']} z_{i_\ell}^{\beta_\ell}$ %
    in $\Pi$, we derive, for each $J = (j_{1}, \ldots, j_{w'})\in [m]^{w'}$, the clause
    $C_J = \bigvee_{\ell\in [w']} (\stdnot{x}_{i_\ell,j_\ell} \lor y_{i_\ell,j_\ell}^{\beta_\ell})$.
    This holds for the axioms by definition of $\indexingnolength{\formf}$.
    Suppose it holds for clause $C \lor z_i$ and $D\lor \stdnot{z}_i$, and let $w'$
    be the width of $D\lor C$.
    Then for any $J = (j_{1}, \ldots, j_{w'}) \in [m]^{w'}$
    and any $j\in [m]$ we can derive
    $(D\lor C)_J \lor \stdnot{x}_{i,j}$ in one step by
    resolving over variable $y_{i,j}$.
    Finally, we can derive $(D\lor C)_J$ in $m$ steps by resolving
    $(D\lor C)_J \lor \stdnot{x}_{i,j}$ for all $j\in [m]$
    with $\bigvee_{j\in [m]}x_{i,j}$.
    This give a total of $m^{w'+ 1} + m = \bigoh{m^{w+1}}$
    steps per new clause in $\Pi$. 
\end{proof}

We can now apply \autoref{thm:size_lifting} to our width-depth trade-off to obtain 
supercritical size-depth trade-offs for resolution.

\begin{theorem}[{Resolution size-depth trade-offs, \general}]\label{thm:size_tradeoff}
    For any %
    $m = m(n)$, $k=k(n)$, and $c=c(n)$ such that $3 \leq c \leq k-1<\frac{n}{2\ln n}$, 
    there are $8$-CNF formulas %
    $F_n$ with $O(m k^2(2n)^{c+1})$ variables and 
    formula size $S(F_n)=\bigoh{m^4 k^2(2n)^{c+1}}$ which resolution can refute in size $\bigoh{m^{k+4}k^2(4n)^k}$,
    but for which any refutation of size at most 
    $\left(\frac{m+1}{2}\right)^{k+c}$ has depth at least $\Omega(n^k)$.
\end{theorem}
\begin{proof}%
Let $F_N = {\IND}_m(F)$, where $F$ is the formula obtained from 
\autoref{thm:main_tradeoff}, our 
supercritical width-depth trade-off, for the parameters $c,k$ and $n$.
Note that $F_N$ is a $8$-CNF formula of size $\bigoh{m^4k^2(2n)^{c+1}}$.
Since by \autoref{thm:main_tradeoff} 
$F$ has a resolution refutation of width $k+3$ and size $O(k^2(4n)^k)$,
we have by \autoref{lem:res-easydirection} that 
$\IND_m(F)$ has a resolution refutation of size $\bigoh{m^{k+4}k^2(4n)^k}$.
The lower bounds follows from combining the lifting theorem (\autoref{thm:size_lifting})
and the width-depth trade-off (\autoref{thm:main_tradeoff}).
\end{proof}

\subsection{Tight Lifting and Supercritical Trade-offs for Monotone Circuits and 
Cutting Planes}\label{subsec:liftingforMCandCP}

A \emph{monotone real circuit} is a Boolean circuit whose gate-set includes all monotone functions of 
the form $f:\mathbb{R} \times \mathbb{R} \rightarrow \mathbb{R}$. It has $n$ input gates $x_1,\ldots, 
x_n$ and must output a bit in $\{0,1\}$. 
Note that monotone real circuits are an extension of traditional monotone circuits. 

We define the more general (semantic) version of cutting planes, to which our lower bounds
also apply.
A \emph{semantic cutting planes refutation} of a system of linear 
inequalities $Ax \geq b$ is a sequence of inequalities $\{c_ix \geq d_i\}_{i \in [s]}$, with $c_i \in 
\mathbb{Z}^n, d_i \in \mathbb{Z}$, such that the final inequality %
is the contradiction 
$0 \geq 1$, and  for every $i \in [s]$, $c_ix \geq d_i$ either belongs to $Ax\geq b$ or follows from 
two previous inequalities %
by a semantic deduction step, that is, 
from $ax \geq b$ and $a'x \geq b'$ we can derive any $cx \geq d$ which 
	satisfies $(ax \geq b) \wedge (a'x \geq b') \implies cx \geq d$ for every $x \in \{0,1\}^n$.
The \emph{size} of a semantic cutting planes refutation is $s$, the number of inequalities in 
the sequence. One 
may view a semantic cutting planes proof as a DAG with one vertex per inequality such that the 
leaves are the inequalities belonging to $Ax \geq b$, the root is $0 \geq 1$, and every non-leaf 
vertex has two incoming edges the vertices from which it was derived. The \emph{depth} of a 
semantic cutting planes proof is the longest root-to-leaf path in this DAG.  

Like previous DAG lifting theorems, it will be convenient to work with the following 
top-down definitions of these models---\emph{rectangle}- and \emph{triangle}-DAGs solving
\emph{(total) search problems}. 
A \emph{search problem} is a relation $\relation \subseteq {\cal D} \times \Output$ where
for every input $x \in  {\cal D}$, there is at least one 
output~$o \in \Output$ such that $(x,o) \in {\cal S}$.
We start by defining \emph{shape-DAGs}~\cite{GGKS20MonotoneCircuit}, 
which are a generalisation of 
{rectangle}-DAGs introduced in~\cite{Razborov95unprovanility}
and simplified in~\cite{Pudlak10Extracting,Sokolov17Daglike}. 

\begin{definition}[Shape-DAG]\label{def:shapedag}
    Let $\cal F \subseteq {\cal D}$ be a family of sets, which we call the ``shapes'' of  the DAG, and 
    ${\cal S} \subseteq {\cal D} \times  {\cal O}$ be a search problem. An ${\cal F}$-DAG solving $\cal 
    S$ is a fan-in $\leq 2$ rooted directed acyclic graph where each vertex $v$ is labeled with a 
    shape $S_v \in {\cal F}$ such that the following hold: 
    \begin{enumerate}
        \item \emph{Root.}~ The distinguished root $r$ is labelled with the ``full'' shape $S_r = {\cal 
        D}$.
        \item \emph{Non-Leaves.}~ If $u$ has children $v,w$ then $S_u \subseteq S_v \cup S_w$.
        \item \emph{Leaf.}~ If $\ell$ is a leaf of the DAG then there is some $o \in {\cal O}$ such that 
        $S_\ell \subseteq {\cal S}^{-1}(o)$.
    \end{enumerate}
    The \emph{size} of an ${\cal F}$-DAG is the number of nodes it contains, and the 
    \emph{depth} is the length of the longest root-to-leaf path in the DAG.
\end{definition}

For a bipartite input domain $X \times Y$, a 
\emph{rectangle} $R~=~R^X \times R^Y$ is a product set, where $R^X \subseteq X$ and $R^Y 
\subseteq Y$. A \emph{triangle} is a subset $T \subseteq X \times Y$ that can be written as 
$T=\{(x,y) \mid a_T(x) < b_T(y)\}$ for some labeling of the rows $a_T:X \rightarrow \mathbb{R}$ and 
columns $b_T:Y \rightarrow \mathbb{R}$ by real numbers.  
A \emph{rectangle-DAG} is a shape-DAG where the set of shapes ${\cal F}$ is the set of all rectangles 
over the  input domain. Similarly, a \emph{triangle-DAG} is a shape-DAG where ${\cal F}$ is the set of all 
triangles. Note that because any rectangle is also a triangle, a rectangle-DAG is a special case of a 
triangle-DAG.

We now introduce the two types of search problems that allow us to relate triangle- and rectangle-DAGs
to cutting planes and monotone circuits.
Let $F = C_1 \wedge \cdots \wedge C_m$ be an unsatisfiable CNF formula on $n$ variables. 
The \emph{falsified clause search problem for $F$} is the following total search problem: 
given $z \in \set{0,1}^n$, find an $i \in [m]$ such that the clause $C_i$ is falsified by $z$. 
Formally, we define the relation 
$\Search(F) \subseteq \set{0,1}^n \times [m]$ by 
\begin{equation} 
(z, i) \in \Search(F) \Longleftrightarrow C_i(z) = 0 \eqperiod
\end{equation}
We are sometimes interested in bipartite input domains, 
so given a partition of the variables of $F$, where we define the relation 
$\Search^{X,Y}(F) \subseteq (X\times Y) \times [m]$ by 
$((x,y), i) \in \Search(F) \Longleftrightarrow C_i(z) = 0$.
It is not difficult to see that for any CNF formula~$F$ and any partition of its variables,
a semantic cutting planes refutation of~$F$ 
implies, for any partition of the variables of~$F$, a triangle-DAG for $\Search^{X,Y}(F)$
of the same size and depth; indeed, any halfspace 
$az \geq b$ defines a triangle $H:=\{z \in \{0,1\}^n \mid az < b\}$. 
Similarly, a resolution refutation of~$F$ 
implies a rectangle-DAG for $\Search^{X,Y}(F)$
of the same size and depth.

Given a total or partial monotone function~$f : \set{0,1}^n \rightarrow \set{0,1}$,  
the \emph{monotone Karchmer--Wigderson search problem}~\cite{KW90Monotone}
$\mKW(f) \subseteq (f^{-1}(1) \times f^{-1}(0)) \times [n]$ is defined as
\begin{equation} 
((x,y) , i) \in \mKW(f) \Longleftrightarrow x_i > y_i \eqperiod
\end{equation}
The DAG-like version of the monotone Karchmer--Wigderson 
relation~\cite{Razborov95unprovanility,Pudlak10Extracting,Sokolov17Daglike}
implies that there is a monotone circuit (respectively, monotone real circuit) 
computing $f$ if and only if there is a 
rectangle-DAG (respectively, triangle-DAG) solving $\mKW(f)$
of the same size and depth.

For our lifting theorems we need to compose search problems with gadgets.
Given a search problem $\relation \subseteq \set{0,1}^n \times \Output$ 
and a gadget $g : \gDomain \to \set{0,1}$, 
we can define $\relation \circ g^n \subseteq \gDomain^{n} \times \Output$
to be the relation where $(x,o) \in \relation \circ g^n$ if and only if
$(z,o) \in \relation$, where $z_i = g(x_i)$ for $i\in [n]$. 
We also consider the search problem $\Search^{X,Y}(\IND_m(F))$,
where $X$ corresponds to the $x$-variables, and $Y$ to the $y$-variables
of $\IND_m(F)$. 
By a standard reduction~\cite{Gal01Characterization,Razborov90Applications},
there is a way of translating between the composed search problems; 
see e.g. \cite{GGKS20MonotoneCircuit} for a proof. 
\begin{fact}
\label{fact:searchToFunc}
	Let $F$ be an unsatisfiable $k$-CNF on $\ell$ clauses and $n$ variables, let 
	$m=m(n)$ be a 
	parameter and $N=\ell \cdot (2m)^k$. There is a partial monotone function 
    $f:\{0,1\}^N \rightarrow \{0,1\}$ such that 
	\begin{enumerate}
		\item $\Search(F) \circ \IND_m^n$ reduces to $\mKW(f)$. In particular, an 
		$\mathcal{F}$-DAG solving $\mKW(f)$ implies an $\mathcal{F}$-DAG solving 
		$\Search(F) \circ \IND_m^n$ of the same size and depth. 
		\item $\mKW(f)$ reduces to $\Search^{X,Y}(\IND_m(F))$. In particular, an 
		$\mathcal{F}$-DAG solving
		$\Search^{X,Y}(\IND_m(F))$ implies an $\mathcal{F}$-DAG solving $\mKW(f)$ of the 
		same size and depth.
	\end{enumerate}
\end{fact}

We now state our lifting theorem from resolution to triangle-DAGs.

\begin{restatable}[Lifting for triangle-DAGs]{restatabletheorem}{dagLifting}
		\label{thm:dagLifting}
		Let $F$ be an $n$-variate unsatisfiable CNF formula, 
        and let $m,w\in\N, \delta>0$ be arbitrary parameters satisfying $w\leq n$, 
       $0<\delta < 1-  \frac{1}{\log m}$ and $m\geq (\frac{50n}{\delta})^{2/\delta}$.
        If there is a triangle-DAG of size $\frac{1}{2}m^{(1-\delta)w}$ and depth $d$ solving 
        $\Search(F)\circ \IND^n_m$, 
        then $F$ has a resolution refutation of width $w$ and depth $dw$.
\end{restatable}

We prove this theorem in \autoref{sec-MonSizeVsDepth}. 
Combining this lifting theorem with our width-depth trade-off for resolution
(\autoref{thm:main_trade-off}) we obtain the supercritical size-depth 
trade-offs for monotone (real) circuits.

\begin{theorem}[Monotone circuit trade-offs, \general]
\label{thm:monotonerealtrade-off-general}
For any integers $c=c(n), k=k(n), m=m(n)$ and real number $\delta = \delta(n)\in(0,0.9)$ such that 
$3\leq c \leq k-1 <\frac{n}{2\log n}$ and $m\geq (\frac{50n}{\delta})^{2/\delta}$, the following holds for 
sufficiently large $n$. 
There are $N$-variate functions $f_N$ over 
$N=\bigoh{m^4k^2(2n)^{c+1}}$ variables computable by a monotone circuit with size at 
most 
$\bigoh{m^{k+4}k^2(4n)^k}$, but for which 
any monotone real circuit with size at most $\frac{1}{2}m^{(1-\delta)(k+c)}$ must have depth at least 
$\Omega(n^k/k)$.
\end{theorem}
\begin{proof}
Let $F$ be the $4$-CNF formula obtained from \autoref{thm:main_tradeoff}, our 
supercritical width-depth trade-off, for the parameters $c,k$ and $n$.
Consider the partial monotone function 
$g_N:\{0,1\}^N \rightarrow \{0,1\}$ obtained by applying 
\autoref{fact:searchToFunc} to $F$.
We have that $N=\bigoh{m^4k^2(2n)^{c+1}}$.
Since by \autoref{thm:main_tradeoff}, 
$F$ has a resolution refutation of width $k+3$ and size $O(k^2(4n)^k)$,
we have by \autoref{lem:res-easydirection} 
that $\mKW(g_N)$
can be solved by a rectangle-DAG of size $ \bigoh{m^{k+4}k^2(4n)^k}$,
where we use the fact that a resolution refutation of~$F$ implies
a rectangle-DAG solving $\Search^{X,Y}(\IND_m(F))$ in the same size,
and that by \autoref{fact:searchToFunc}  $\mKW(g_N)$
reduces to $\Search^{X,Y}(\IND_m(F))$.
This implies that there is a monotone circuit of the same size
computing $g_N$. Let $f_N$ be the total function,
which extends $g_N$, computed by this circuit.

Now, if there is a monotone real circuit of size $s$ and depth $d$ computing~$f_N$,
then there is a triangle-DAG of size $s$ and depth $d$ solving $\mKW(f_N)$, 
and hence also $\mKW(g_N)$.
By \autoref{fact:searchToFunc} this implies there is a triangle-DAG 
solving $\Search(F)\circ \IND_m^n$ in the same size and depth.
Finally, combining the triangle-DAG lifting theorem (\autoref{thm:dagLifting}) 
and the width-depth trade-off (\autoref{thm:main_tradeoff}) 
we conclude that if $s \leq \frac{1}{2}m^{(1-\delta)(k+c)}$ 
then $d = \Omega(n^k/(k+c))=\Omega(n^k/k)$.
\end{proof}

We can obtain a similar supercritical trade-off for cutting planes.

\begin{theorem}[Cutting planes trade-offs, \general]
\label{thm:CPtrade-off-general}
For any integers $c=c(n), k=k(n), m=m(n)$ and real number $\delta = \delta(n)\in(0,0.9)$ such that 
$3\leq c< k<\frac{n}{2\log n}$ and $m\geq (\frac{50n}{\delta})^{2/\delta}$, the following holds for 
all $n$. 
There are unsatisfiable $3$-CNF formulas $F_N$ of size 
$N=\bigoh{m^4k^2(2n)^{c+1}}$ 
that can be refuted in resolution in size $O(m^{k+4}k^2(4n)^k)$, but for which 
any semantic cutting planes refutation in size at most $\frac{1}{2}m^{(1-\delta)(k+c)}$ 
must have depth at least $\Omega(n^k/k)$.
\end{theorem}

This theorem can be proven along the same lines as \autoref{thm:monotonerealtrade-off-general},
by applying the lifting theorem (\autoref{thm:dagLifting}) to the 
width-depth trade-off (\autoref{thm:main_tradeoff}) together with 
\autoref{fact:searchToFunc}, and using \autoref{lem:res-easydirection} for the upper bound.
The only caveat is that this would give us a $8$-CNF formula. In order to obtain a
$3$-CNF formula, we need to 
define a $3$-CNF version of $\IND_m(F)$, 
denoted by $\widetilde{\IND}_m(F)$.
Let $F$ be a CNF formula over variables $z = z_1, \ldots, z_n$, %
then the formula $\widetilde{\IND}_m(F)$ is over variables
$x_{i,j}$ and $y_{i,j}$ where $i\in [n]$ and $j\in [m]$, 
the extension variables to write each of the clauses $\bigvee_{j\in [m]} x_{i,j}$, for $i\in[n]$, as a 
$3$-CNF formula, 
along with variables $x_{C,J}$ and $y_{C,J}$ for every $C\in F$ and every 
$J\in[m]^{\widthofarg{C}}$.
The clauses in $\widetilde{\IND}_m(F)$ consist of:
a $3$-CNF encoding of $\bigvee_{j\in [m]} x_{i,j}$ for every $i\in [n]$; 
for every $C = \bigvee_{\ell\in [w]}z_{i_\ell}^{\beta_\ell}$ in $F$ and every 
$J = (j_1, \ldots, j_w) \in [m]^w$,
a $3$-CNF encoding of $(\bigwedge_{\ell\in [w]} x_{i_\ell,j_\ell}) \rightarrow x_{C,J}$,
a $2$-clause $x_{C,J} \rightarrow y_{C,J}$, and a $3$-CNF encoding of
$y_{C,J} \rightarrow  \bigvee_{\ell\in [w]}y_{i_\ell, j_\ell}^{\beta_\ell}$.
Note that if $F$ is a $w$-CNF formula,
then $\widetilde{\IND}_m(F)$ has $\bigoh{w\cdot \clausesizeofarg{F} \cdot m^w + nm}$
variables and clauses.

We observe two basic facts about $\widetilde{\IND}_m(F)$. First, 
every size-$s$ resolution refutation of $\IND_m(F)$
can be made into a size-$\bigoh{s + \clausesizeofarg{\widetilde{\IND}_m(F)}}$ refutation of 
$\widetilde{\IND}_m(F)$. 
This is because $\IND_m(F)$ can be derived from $\widetilde{\IND}_m(F)$ in linear size.
Secondly, for both rectangle- and triangle-DAGs (or any shape-DAG
that is closed under taking intersection with rectangles),
the search problem $\Search^{X,Y}(\IND_m(F))$ reduces to 
$\Search^{\widetilde{X},\widetilde{Y}}(\widetilde{\IND}_m(F))$, where $\widetilde{X}$ corresponds 
to the $x$-variables, and
$\widetilde{Y}$ to the $y$-variables of $\widetilde{\IND}_m(F)$. 
Indeed, we can fix a pair of injective maps $\phi_X:\{0,1\}^{X}\to\{0,1\}^{\widetilde{X}}$ and 
$\phi_Y:\{0,1\}^{Y}\to \{0,1\}^{\widetilde{Y}}$ which extend every assignment on $X\cup Y$ to one on 
$\widetilde{X}\cup\widetilde{Y}$ according to the semantic meaning of the new variables.
Let $\mathcal{O}$ be the set of possible outputs of $\Search^{X,Y}(\IND_m(F))$,
which we view as the set of clauses of $\IND_m(F)$.
Similarly, let $\widetilde{\mathcal{O}}$ be the set of clauses of $\widetilde{\IND}_m(F)$.
We can define an injective map $\phi_{\widetilde{\mathcal{O}}}: \widetilde{\mathcal{O}} \to 
{\mathcal{O}}$ which given a clause in $\widetilde{\IND}_m(F)$ outputs
the clause $\IND_m(F)$ it came from.
Therefore, given an $\mathcal{F}$-DAG, where $\mathcal{F}$ is a shape-DAG 
closed under taking intersections with rectangles, 
$\widetilde{\Gamma}$ solving $\Search^{\widetilde{X},\widetilde{Y}}(\widetilde{\IND}_m(F))$, we 
can 
create an 
 $\mathcal{F}$-DAG $\Gamma$ solving $\Search^{X,Y}(\IND_m(F))$ with the same topology, as 
 follows. For each 
node in~$\widetilde{\Gamma}$---which is a subset of 
$\{0,1\}^{\widetilde{X}}\times\{0,1\}^{\widetilde{Y}}$---we take 
its intersection with $\phi_X(\{0,1\}^{X})\times\phi_Y(\{0,1\}^Y)$ and view it as a subset of 
$\{0,1\}^{X}\times\{0,1\}^Y$ via $\phi_X^{-1}\times\phi_Y^{-1}$, giving the corresponding 
node of 
$\Gamma$. 
It is not hard to see that $\Gamma$ is an $\mathcal{F}$-DAG for $\Search^{X,Y}(\IND_m(F))$.

\begin{proof}[Proof of \autoref{thm:CPtrade-off-general}]
Let $F_N = \widetilde{\IND}_m(F)$, where $F$ is the formula obtained from 
\autoref{thm:main_tradeoff}, our 
supercritical width-depth trade-off, for the parameters $c,k$ and $n$.
Note that $F_N$ is a $3$-CNF formula that has $\bigoh{m^4k^2(2n)^{c+1}}$ 
variables and clauses.
Since by \autoref{thm:main_tradeoff}, 
$F$ has a resolution refutation of width $k+3$ and size $O(k^2(4n)^k)$,
we have by \autoref{lem:res-easydirection} that 
$\IND_m(F)$, and hence also $\widetilde{\IND}_m(F)$,
has a resolution refutation of size $\bigoh{m^{k+4}k^2(4n)^k}$.

Now, if there is a semantic cutting planes refutation of~$F_N$ of size $s$ and depth $d$,
then there is a triangle-DAG solving $\Search^{X,Y}(\IND_m(F))$
of size $s$ and depth $d$, using %
the fact above that $\Search^{X,Y}(\IND_m(F))$ reduces to 
$\Search^{\widetilde{X},\widetilde{Y}}(\widetilde{\IND}_m(F))$.
By \autoref{fact:searchToFunc} this gives a triangle-DAG solving 
${\Search(F)\circ\IND_m^n}$ of the same size and depth.
Finally, combining the triangle-DAG lifting theorem (\autoref{thm:dagLifting}) and the width-depth 
trade-off (\autoref{thm:main_tradeoff}), 
we conclude that if $s \leq \frac{1}{2}m^{(1-\delta)(k+c)}$ 
then $d = \Omega(n^k/(k+c)) = \Omega(n^k/(k+c))$.
\end{proof}

  Now \autoref{thm:monotonerealtrade-off} follows from \autoref{thm:monotonerealtrade-off-general}, 
  and \autoref{thm:CPtrade-off} from \autoref{thm:CPtrade-off-general}, by setting: 
  (1) $k$ to be a sufficiently large constant, 
  $c\coloneqq \floor{0.41k}$, 
  $w\coloneqq k+c$, 
  $\delta\coloneqq \frac{1}{200}$, 
  and $m\coloneqq n^{500}$; %
  and (2) $k\coloneqq \floor{\frac{n}{4\log n}}$,
  $c\coloneqq \floor{\sqrt{k}}$,
  $w\coloneqq k+c$, 
  $\delta\coloneqq \frac{1}{4\sqrt{k}}$, and 
  $m\coloneqq \floor{n^{3/\delta}}=\floor{n^{12\sqrt{k}}}$.

\section{Compressed Cop-Robber Game and the Formula}\label{sec:construction}
In this section, we define the compressed Cop-Robber game introduced by \cite{GLNS23CompressingCFI}.
In \autoref{subsec:graphandconstruction}, we define graph compression and the resulting compressed Tseitin formula in general terms, and then construct a concrete instance generalizing the one in \cite{GLNS23CompressingCFI}.
In \autoref{subsec:VCVR}, we define the compressed game. 
In \autoref{subsec:twofacts}, we show some basic facts about resolution and the compressed game.

\subsection{Graph Compression and Compressed Tseitin Formula}\label{subsec:graphandconstruction}
In the next definition, we assume that the graph $G$ is given together with its adjacency list. That is, for 
each vertex there is an ordered list of its neighbors. 
\begin{definition}[Graph compression]\label{def:compressiongeneral}
    We call an equivalence relation $\vertexeq$ on $\vertices{G}$ \emph{compatible} %
    if $u\vertexeq v$ implies that $u,v$ are non-adjacent and have the same degree. A compatible $\vertexeq$ induces an equivalence relation $\edgeeq$ on $\edges{G}$ as follows. First, we let two edges $e_1$ and $e_2$ be equivalent if there are $v_1, v_2\in\vertices{G}$ such that $e_1=\set{v_1,w_1}$, $e_2=\set{v_2,w_2}$, $v_1\vertexeq v_2$, and the position of $w_1$ in the neighbor list of $v_1$ equals the position of $w_2$ in that of $v_2$. 
    Then we take the transitive closure of this relation on $\edges{G}$ to be $\edgeeq$. 
    We call the triple $(G,\vertexeq,\edgeeq)$ a \emph{graph compression}.
\end{definition}

Given a graph compression $(G,\vertexeq,\edgeeq)$, let $/_{\vertexeq}$ be the map from the 
vertices to their equivalence classes, and
for $W\subseteq \vertices{G}$, let $W/_{\vertexeq}$ denote the image of $W$ and 
$\vclassset{W}\coloneqq\setdescr{v%
}{v\vertexeq w\text{ for some $w\in W$}}$. 
Similarly, let $/_{\edgeeq}$ be the map from the edges to their equivalence classes, 
and for $F\subseteq \edges{G}$, let $\eclass{F}$ denote the image of $F$ 
and $\eclassset{F}\coloneqq\setdescr{e}{e\edgeeq f\text{ for some $f\in F$}}$. 
We write $\set{e}_\edgeeq$ as $e_\edgeeq$ and $\set{v}_\vertexeq$ as~$v_\vertexeq$. 

\begin{definition}[Compressed Tseitin]\label{def:compressedtseitin}
Given a Tseitin formula on $G$ and a graph compression $(G,\vertexeq,\edgeeq)$, the edge 
equivalence $\edgeeq$ induces a variable substitution $x_e\mapsto\xeq{e}$ as follows: for each equivalence class, introduce a single, new variable $\xeq{e}$ and replace the variable of every edge in this equivalence class by $\xeq{e}$. We call the resulting CNF formula the \emph{compressed Tseitin formula}.
\end{definition}

\begin{remark}\label{rmk:formula}
    Observe that the compressed Tseitin formula has width at most $\deg(G)$, the maximal vertex degree, and $\setsize{E(G)/\equiv_E}$-many variables. Also, the parity constraints at equivalent vertices become the same after the substitution, and so the compressed formula has at most $2^{\deg(G)-1}\setsize{V(G)/\equiv_V}$-many clauses.
    
    The formula remains unsatisfiable since it is obtained from an unsatisfiable formula via a variable substitution. 
    Moreover, since this substitution is a projection (i.e., each $x_e$ is substituted with one variable), 
    any resolution refutation of the original formula gives rise to one of the compressed formula with no greater proof width, depth, or size. 
\end{remark}

We will focus on an explicit graph and graph compression. 
For a fixed $k\geq 2$, we define $\basegraph$ to be the \emph{cylinder graph} with $k$ rows and $\midlength+2\earlength$ columns, each column being a cycle, where $L$ and $r$ are parameters to be set later in \autoref{def:parameters}. 
When $k=2$, it is a grid which we still denote by $\basegraph$. 
Every vertex of $\basegraph$ has degree 4 except those on the first and last columns, which have degree 3.  
Denote the rows by $1,\dots ,k$ and the columns by $1, \dots, \midlength+2\earlength$. 
We call vertices on columns $[1,r]$, $[\earlength+1,\earlength+\midlength]$, $[\earlength+	\midlength+1,\midlength+2\earlength]$ the \emph{left, middle, right part} of the graph, respectively. 
We will specify parameters $\midlength$ and $\earlength$ in terms of $k$ later. %

\begin{definition}[Concrete graph compression]\label{def:compression} 
We let $(\basegraph,\vertexeq,\edgeeq)$ be a graph compression on the cylinder $\basegraph$ defined as follows.
\begin{enumerate}
    \item \emph{Compatible vertex equivalence $\vertexeq$.}~ We pick factors $\modulus{1},\dots , \modulus{k}$ of $\midlength$ which are all greater than $2$, called the \emph{moduli of rows}. We define a vertex equivalence relation where, on each row $i$, $(i,a)\vertexeq (i,b)$ if both vertices fall in the middle part (i.e., $a,b\in[\earlength+1,\earlength+\midlength]$) and $a-b=0\mod \modulus{i}$. 
    \item \emph{Edge equivalence $\edgeeq$.}~ The above $\vertexeq$ induces an edge equivalence relation on $\edges{\basegraph}$ as in \autoref{def:compressiongeneral}: we order the edges incident to a vertex by the canonical choice $(\text{left, right, up, down})$, adjusted to a subset of size 3 for vertices on boundary columns.
\end{enumerate}
\end{definition}

\begin{figure}
    \centering
\begin{subfigure}[b]{0.96\textwidth}
	\begin{center}
	\begin{tikzpicture}
		\draw[very thick, color=black!80] (0,0) -- (14.5,0);
		\draw[very thick, color=black!80] (0,-2) -- (14.5,-2);
		
		\draw[very thick, color=blue!40] (0,0) -- (0,-2);
		\draw[very thick, color=betterGreen!75] (0.5,0) -- (0.5,-2);
		\draw[very thick, color=red!40] (1,0) -- (1,-2);
		\draw[very thick, color=blue!40] (1.5,0) -- (1.5,-2);
		\draw[very thick, color=betterGreen!75] (2,0) -- (2,-2);
		\draw[very thick, color=red!40] (2.5,0) -- (2.5,-2);
		\draw[very thick, color=blue!40] (3,0) -- (3,-2);
		\draw[very thick, color=betterGreen!75] (3.5,0) -- (3.5,-2);
		\draw[very thick, color=red!40] (4,0) -- (4,-2);
		\draw[very thick, color=blue!40] (4.5,0) -- (4.5,-2);
		\draw[very thick, color=betterGreen!75] (5,0) -- (5,-2);
		\draw[very thick, color=red!40] (5.5,0) -- (5.5,-2);
		\draw[very thick, color=blue!40] (6,0) -- (6,-2);
		\draw[very thick, color=betterGreen!75] (6.5,0) -- (6.5,-2);
		\draw[very thick, color=red!40] (7,0) -- (7,-2);
		\draw[very thick, color=blue!40] (7.5,0) -- (7.5,-2);
		\draw[very thick, color=betterGreen!75] (8,0) -- (8,-2);
		\draw[very thick, color=red!40] (8.5,0) -- (8.5,-2);
		\draw[very thick, color=blue!40] (9,0) -- (9,-2);
		\draw[very thick, color=betterGreen!75] (9.5,0) -- (9.5,-2);
		\draw[very thick, color=red!40] (10,0) -- (10,-2);
		\draw[very thick, color=blue!40] (10.5,0) -- (10.5,-2);
		\draw[very thick, color=betterGreen!75] (11,0) -- (11,-2);
		\draw[very thick, color=red!40] (11.5,0) -- (11.5,-2);
		\draw[very thick, color=blue!40] (12,0) -- (12,-2);
		\draw[very thick, color=betterGreen!75] (12.5,0) -- (12.5,-2);
		\draw[very thick, color=red!40] (13,0) -- (13,-2);
		\draw[very thick, color=blue!40] (13.5,0) -- (13.5,-2);
		\draw[very thick, color=betterGreen!75] (14,0) -- (14,-2);
		\draw[very thick, color=red!40] (14.5,0) -- (14.5,-2);

		\filldraw[fill=betterYellow!16,very thick] (0,0) circle (5pt);
		\filldraw[fill=betterYellow!16,very thick] (0,-2) circle (5pt);
		\filldraw[fill=betterYellow!16,very thick] (0.5,0) circle (5pt);
		\filldraw[fill=betterYellow!16,very thick] (0.5,-2) circle (5pt);
		\filldraw[fill=betterYellow!16,very thick] (1,0) circle (5pt);
		\filldraw[fill=betterYellow!16,very thick] (1,-2) circle (5pt);
		\filldraw[fill=betterYellow!16,very thick] (1.5,0) circle (5pt);
		\filldraw[fill=betterYellow!16,very thick] (1.5,-2) circle (5pt);
		\filldraw[fill=betterYellow!16,very thick] (2,0) circle (5pt);
		\filldraw[fill=betterYellow!16,very thick] (2,-2) circle (5pt);
		\filldraw[fill=betterYellow!16,very thick] (2.5,0) circle (5pt);
		\filldraw[fill=betterYellow!16,very thick] (2.5,-2) circle (5pt);
		\filldraw[fill=betterYellow!16,very thick] (3,0) circle (5pt);
		\filldraw[fill=betterYellow!16,very thick] (3,-2) circle (5pt);
		\filldraw[fill=betterYellow!16,very thick] (3.5,0) circle (5pt);
		\filldraw[fill=betterYellow!16,very thick] (3.5,-2) circle (5pt);
		\filldraw[fill=betterYellow!16,very thick] (4,0) circle (5pt);
		\filldraw[fill=betterYellow!16,very thick] (4,-2) circle (5pt);
		\filldraw[fill=betterYellow!16,very thick] (4.5,0) circle (5pt);
		\filldraw[fill=betterYellow!16,very thick] (4.5,-2) circle (5pt);
		\filldraw[fill=betterYellow!16,very thick] (5,0) circle (5pt);
		\filldraw[fill=betterYellow!16,very thick] (5,-2) circle (5pt);
		\filldraw[fill=betterYellow!16,very thick] (5.5,0) circle (5pt);
		\filldraw[fill=betterYellow!16,very thick] (5.5,-2) circle (5pt);
		\filldraw[fill=betterYellow!16,very thick] (6,0) circle (5pt);
		\filldraw[fill=betterYellow!16,very thick] (6,-2) circle (5pt);
		\filldraw[fill=betterYellow!16,very thick] (6.5,0) circle (5pt);
		\filldraw[fill=betterYellow!16,very thick] (6.5,-2) circle (5pt);
		\filldraw[fill=betterYellow!16,very thick] (7,0) circle (5pt);
		\filldraw[fill=betterYellow!16,very thick] (7,-2) circle (5pt);
		\filldraw[fill=betterYellow!16,very thick] (7.5,0) circle (5pt);
		\filldraw[fill=betterYellow!16,very thick] (7.5,-2) circle (5pt);
		\filldraw[fill=betterYellow!16,very thick] (8,0) circle (5pt);
		\filldraw[fill=betterYellow!16,very thick] (8,-2) circle (5pt);
		\filldraw[fill=betterYellow!16,very thick] (8.5,0) circle (5pt);
		\filldraw[fill=betterYellow!16,very thick] (8.5,-2) circle (5pt);
		\filldraw[fill=betterYellow!16,very thick] (9,0) circle (5pt);
		\filldraw[fill=betterYellow!16,very thick] (9,-2) circle (5pt);
		\filldraw[fill=betterYellow!16,very thick] (9.5,0) circle (5pt);
		\filldraw[fill=betterYellow!16,very thick] (9.5,-2) circle (5pt);
		\filldraw[fill=betterYellow!16,very thick] (10,0) circle (5pt);
		\filldraw[fill=betterYellow!16,very thick] (10,-2) circle (5pt);
		\filldraw[fill=betterYellow!16,very thick] (10.5,0) circle (5pt);
		\filldraw[fill=betterYellow!16,very thick] (10.5,-2) circle (5pt);
		\filldraw[fill=betterYellow!16,very thick] (11,0) circle (5pt);
		\filldraw[fill=betterYellow!16,very thick] (11,-2) circle (5pt);
		\filldraw[fill=betterYellow!16,very thick] (11.5,0) circle (5pt);
		\filldraw[fill=betterYellow!16,very thick] (11.5,-2) circle (5pt);
		\filldraw[fill=betterYellow!16,very thick] (12,0) circle (5pt);
		\filldraw[fill=betterYellow!16,very thick] (12,-2) circle (5pt);
		\filldraw[fill=betterYellow!16,very thick] (12.5,0) circle (5pt);
		\filldraw[fill=betterYellow!16,very thick] (12.5,-2) circle (5pt);
		\filldraw[fill=betterYellow!16,very thick] (13,0) circle (5pt);
		\filldraw[fill=betterYellow!16,very thick] (13,-2) circle (5pt);
		\filldraw[fill=betterYellow!16,very thick] (13.5,0) circle (5pt);
		\filldraw[fill=betterYellow!16,very thick] (13.5,-2) circle (5pt);
		\filldraw[fill=betterYellow!16,very thick] (14,0) circle (5pt);
		\filldraw[fill=betterYellow!16,very thick] (14,-2) circle (5pt);
		\filldraw[fill=betterYellow!16,very thick] (14.5,0) circle (5pt);
		\filldraw[fill=betterYellow!16,very thick] (14.5,-2) circle (5pt);

		\node[] at (0,0.5) {$1$};
		\node[] at (2,0.5) {$5$};
		\node[] at (4.45,0.5) {$10$};
		\node[] at (6.95,0.5) {$15$};
		\node[] at (9.45,0.5) {$20$};
		\node[] at (12,0.5) {$25$};
		\node[] at (14.5,0.5) {$30$};
	\end{tikzpicture}	
	\end{center}
   \caption{Before compression.}
   \label{fig:} 
\end{subfigure}

\begin{subfigure}[b]{1\textwidth}
   \begin{center}
   \begin{tikzpicture}
		\draw [black!80, very thick] plot [smooth, tension=1] coordinates { (0,0) (1,1) (9,1) (10,0)};
		
		\draw [black!80, very thick] plot [smooth, tension=1] coordinates { (0,-3) (1,-4) (9,-4) (10,-3)};
	
		\draw[very thick, color=black!80] (0,0) -- (10,0);
		\draw[very thick, color=black!80] (0,-3) -- (10,-3);
		
		\draw[very thick, color=blue!40] (0,0) -- (0,-3);
		\draw[very thick, color=blue!40] (0,0) -- (2.14,-3);
		\draw[very thick, color=blue!40] (0,0) -- (4.28,-3);
		\draw[very thick, color=blue!40] (0,0) -- (4.28,-3);
		\draw[very thick, color=blue!40] (0,0) -- (6.42,-3);
		\draw[very thick, color=blue!40] (0,0) -- (8.56,-3);
		
		\draw[very thick, color=betterGreen!75] (2,0) -- (0.71,-3);
		\draw[very thick, color=betterGreen!75] (2,0) -- (2.86,-3);
		\draw[very thick, color=betterGreen!75] (2,0) -- (5,-3);
		\draw[very thick, color=betterGreen!75] (2,0) -- (7.14,-3);
		\draw[very thick, color=betterGreen!75] (2,0) -- (9.28,-3);
		
		\draw[very thick, color=red!40] (4,0) -- (1.42,-3);
		\draw[very thick, color=red!40] (4,0) -- (3.57,-3);
		\draw[very thick, color=red!40] (4,0) -- (5.71,-3);
		\draw[very thick, color=red!40] (4,0) -- (7.85,-3);
		\draw[very thick, color=red!40] (4,0) -- (10,-3);
		
		\draw[very thick, color=blue!40] (6,0) -- (0,-3);
		\draw[very thick, color=blue!40] (6,0) -- (2.14,-3);
		\draw[very thick, color=blue!40] (6,0) -- (4.28,-3);
		\draw[very thick, color=blue!40] (6,0) -- (4.28,-3);
		\draw[very thick, color=blue!40] (6,0) -- (6.42,-3);
		\draw[very thick, color=blue!40] (6,0) -- (8.56,-3);
		
		\draw[very thick, color=betterGreen!75] (8,0) -- (0.71,-3);
		\draw[very thick, color=betterGreen!75] (8,0) -- (2.86,-3);
		\draw[very thick, color=betterGreen!75] (8,0) -- (5,-3);
		\draw[very thick, color=betterGreen!75] (8,0) -- (7.14,-3);
		\draw[very thick, color=betterGreen!75] (8,0) -- (9.28,-3);
		
		\draw[very thick, color=red!40] (10,0) -- (1.42,-3);
		\draw[very thick, color=red!40] (10,0) -- (3.57,-3);
		\draw[very thick, color=red!40] (10,0) -- (5.71,-3);
		\draw[very thick, color=red!40] (10,0) -- (7.85,-3);
		\draw[very thick, color=red!40] (10,0) -- (10,-3);
		
		\filldraw[fill=betterYellow!16,very thick] (0,0) circle (5pt);
		\filldraw[fill=betterYellow!16,very thick] (2,0) circle (5pt);
		\filldraw[fill=betterYellow!16,very thick] (4,0) circle (5pt);
		\filldraw[fill=betterYellow!16,very thick] (6,0) circle (5pt);
		\filldraw[fill=betterYellow!16,very thick] (8,0) circle (5pt);
		\filldraw[fill=betterYellow!16,very thick] (10,0) circle (5pt);
		
		\filldraw[fill=betterYellow!16,very thick] (0,-3) circle (5pt);
		\filldraw[fill=betterYellow!16,very thick] (0.71,-3) circle (5pt);
		\filldraw[fill=betterYellow!16,very thick] (1.42,-3) circle (5pt);
		\filldraw[fill=betterYellow!16,very thick] (2.14,-3) circle (5pt);
		\filldraw[fill=betterYellow!16,very thick] (2.86,-3) circle (5pt);
		\filldraw[fill=betterYellow!16,very thick] (3.57,-3) circle (5pt);
		\filldraw[fill=betterYellow!16,very thick] (4.28,-3) circle (5pt);
		\filldraw[fill=betterYellow!16,very thick] (5,-3) circle (5pt);
		\filldraw[fill=betterYellow!16,very thick] (5.71,-3) circle (5pt);
		\filldraw[fill=betterYellow!16,very thick] (6.42,-3) circle (5pt);
		\filldraw[fill=betterYellow!16,very thick] (7.14,-3) circle (5pt);
		\filldraw[fill=betterYellow!16,very thick] (7.85,-3) circle (5pt);
		\filldraw[fill=betterYellow!16,very thick] (8.56,-3) circle (5pt);
		\filldraw[fill=betterYellow!16,very thick] (9.28,-3) circle (5pt);
		\filldraw[fill=betterYellow!16,very thick] (10,-3) circle (5pt);
		
		\node[] at (0,0.5) {$1$};
		\node[] at (2,0.5) {$2$};
		\node[] at (4,0.5) {$3$};
		\node[] at (6,0.5) {$4$};
		\node[] at (8,0.5) {$5$};
		\node[] at (10,0.5) {$6$};
		
		\node[] at (0,-3.5) {$1$};
		\node[] at (0.71,-3.5) {$2$};
		\node[] at (1.42,-3.5) {$3$};
		\node[] at (2.14,-3.5) {$4$};
		\node[] at (2.86,-3.5) {$5$};
		\node[] at (3.57,-3.5) {$6$};
		\node[] at (4.28,-3.5) {$7$};
		\node[] at (5,-3.5) {$8$};
		\node[] at (5.71,-3.5) {$9$};
		\node[] at (6.42,-3.5) {$10$};
		\node[] at (7.14,-3.5) {$11$};
		\node[] at (7.85,-3.5) {$12$};
		\node[] at (8.56,-3.5) {$13$};
		\node[] at (9.28,-3.5) {$14$};
		\node[] at (10,-3.5) {$15$};	
	\end{tikzpicture}
	\end{center}
   \caption{After compression.}
   \label{fig:Ng2}
\end{subfigure}
    \caption{The compression of two rows in the middle part is depicted. The parameters are chosen as $\midlength=P_1 \cdot P_2 \cdot P_3 = 2 \cdot 3 \cdot 5$, $c=1$, $m_1=P_1\cdot P_2 = 6$, and $m_2=P_2\cdot P_3 = 15$. Equivalent vertical edges are drawn in the same color. 
    }
    \label{fig:compr}
\end{figure}
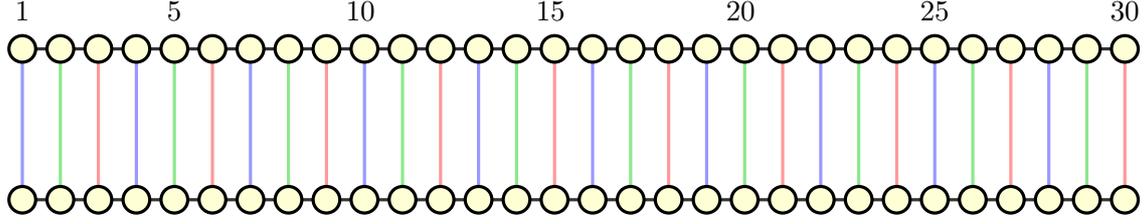
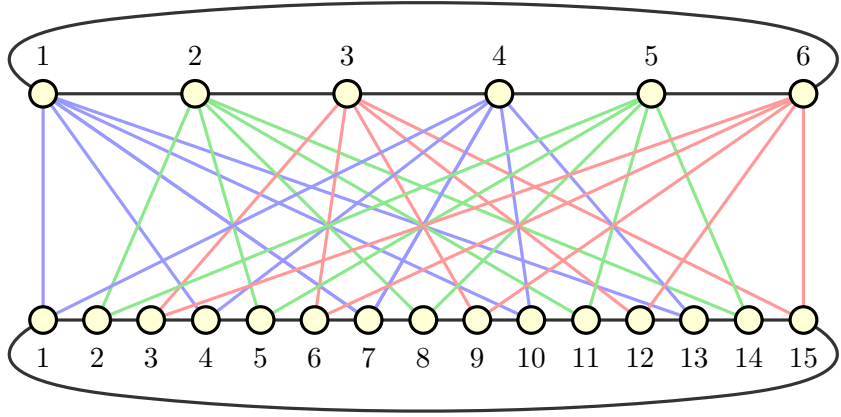

\begin{definition}[Concrete compressed Tseitin]\label{def:taubasegraph}
We use $\Ts(\basegraph)$ to denote the Tseitin formula where the only odd charge is at vertex $(1,1)$. Then the edge equivalence $\edgeeq$ induces a variable identification on $\Ts(\basegraph)$ and hence a new CNF formula, which we denote by $\tau(\basegraph)$. 
\end{definition}

We make the following three observations about the compression.

\begin{observation} \label{item:obs1} 
In the edge equivalence $\edgeeq$ on $E(\basegraph)$, a horizontal edge can be equivalent only to horizontal edges in the same row, and a vertical edge only to vertical edges between the same two adjacent rows.
\end{observation}

\begin{observation} \label{item:obs2} 
We can visualize the compression as follows. 
Fixing a row $i$, %
we use $a$ to represent vertex $(i,a)$. 
Then each vertex on the left and right parts of $\basegraph$ is a singleton vertex class, 
and there are $\modulus{i}$ vertex classes in the middle part. 
In other words, the set of vertices in row $i$ is partitioned into $2r+\modulus{i}$ many subsets according to $\vertexeq$: the singletons $\set{1},\dots,\set{\earlength},\set{L+r+1},\dots ,\set{L+2r}$, and $\vclassset{(r+1)}$, \dots, $\vclassset{(r+\modulus{i})}$, where $\vclassset{x}=\setdescr{y\in[\earlength+1,\earlength+\midlength]}{y=x\mod \modulus{i}}$.  
The horizontal edge classes are partitioned into $2(\earlength-1)+\modulus{i}$ subsets: a singleton for each edge within the left or within the right part, and $H_i(1),\dots ,H_i(\modulus{i})$, where $H_i(a)=\setdescr{\set{x,x+1}}{ x\in[\earlength,\earlength+\midlength],\ x=a \mod \modulus{i}}$. 
Note that the two edges $\set{\earlength,\earlength+1}$ and $ \set{\earlength+\midlength,\earlength+\midlength+1}$ that ``cross parts'' are equivalent and fall in $H_i(r)$ since $\modulus{i}|\midlength$. 
As for the vertical edges, they are all between adjacent rows. 
Between rows $i$ and $i+1$ (counted mod $k$), the vertical edges within the left or the right part give $2\earlength$ singleton edge classes. 
In the middle part, using the Chinese remainder theorem, it is not hard to see that there are $g_i\coloneqq \gcd(\modulus{i},\modulus{i+1})$-many vertical edge classes $T_i(1),\dots ,T_i(g_i)$, where
\begin{equation}
    T_i(a)=\Setdescr{\set{(i,x),(i+1,x)}}{\ x\in[\earlength+1,\earlength+\midlength]\text{ and }x-a=0\hspace{-6pt} \mod g_i}\eqperiod
\end{equation}
See \autoref{fig:compr} for an illustration.
\end{observation}

\begin{observation} \label{item:obs3} 
The number of vertex class is $\setsize{V(\basegraph)/\equiv_V}=2k\earlength+\sum_{i=1}^k\modulus{i}$, and the number of edge class is $\setsize{E(\basegraph)/\equiv_E}=(8\earlength-2)k + \sum_{i=1}^k(\modulus{i} + \gcd(m_i,m_{i+1}))$.
\end{observation}

In this and the next sections, given $n$, $k$, and $c$, we set the parameters $m_i$, $L$, and $r$ in $\basegraph$ as follows.

\begin{definition}[Parameters in construction] \label{def:parameters}
For fixed $k\geq 2$ and $c\in[1,k-1]$, let $n>4k+2$ be an integer such that there are $k$ pairwise coprime numbers $P_1,\dots ,P_k$ in $[n,2n]$.
(The prime number theorem implies that having $n>2k\ln k+C_{abs}$ suffices, where $C_{abs}$ is 
an absolute constant.)
For $i\in[k]$, we choose the parameters in \autoref{def:compression} as follows.
\begin{align}
\modulus{i}&\coloneqq 2(k+c) \cdot P_i\cdots P_{i+c} \eqcomma \label{eq:mi}\\
\midlength&\coloneqq  2(k+c) \cdot P_1\cdots P_{k} \eqcomma \label{eq:L}\\
r&\coloneqq k+c+1 \eqperiod \label{eq:r}
\end{align}
\end{definition}

The condition $P_1,\dots ,P_k \in [n,2n]$ in \autoref{def:parameters} is only for the estimates $\modulus{i}\leq k^2(2n)^{c+1}$ and $\midlength\geq k^2n^k$ to hold. 
The more structural properties we need of the parameters are summarized as follows.

\begin{property}[Properties of parameters]\label{prop:parameters} For any subset $I\subseteq[k]$ of rows, denote ${g_I\coloneqq \gcd(\modulus{i}\mid i\in I)}$. Then the following holds for the parameters in \autoref{def:parameters}.
\begin{enumerate}[label=(P{\arabic*})]
\item $g_{[k]}\geq 2(k+c)$;\label{P1}
\item For all $i\in[k]$ and all $a,b\geq 1$ such that $a+b\leq c+1$, $\modulus{i}\mid\lcm\left(g_{[i-a,i]},\ g_{[i,i+b]}\right)$. 
Here, $[i-a,i]$ and $[i,i+b]$ are cyclic intervals modulo $k$;\label{P2}
\item For all $I\subseteq[k]$ with $\setsize{I}\geq k-c$, it holds that $\lcm(\modulus{i}\mid i\in I)=\midlength$;\label{P3}
\item $r>k+c$.\label{P4}
\end{enumerate}
\end{property}

Let us see that the parameters indeed satisfy the above. 
It is clear for \ref{P1} and \ref{P4}. 
For \ref{P2}, by \eqref{eq:mi} we have that %
{$A:=g_{[i-a,i]}/(2(k+c))$} is divisible by $P_i,\dots ,P_{i-a+c}$ and %
{$B:=g_{[i,i+b]}/(2(k+c))$} is divisible by $P_{i+b},\dots ,P_{i+c}$. 
Since $i+b\leq i-a+c+1$, we get that  %
{$\lcm\left(A,B\right)$} is divisible by $P_i,\dots ,P_{i+c}$. 
For \ref{P3}, note that every $P_i$ appears (as a factor) in $c+1$ rows by \eqref{eq:mi}, so if $\setsize{I}\geq k-c$, then $P_i$ is in $\modulus{j}$ for some $j\in I$ and hence $\lcm(\modulus{i}\mid i\in I)=P_1\cdots P_k=\midlength$.
\smallskip

Finally, by \autoref{rmk:formula} and \autoref{item:obs3}, 
the compressed Tseitin formula $\tau(\basegraph)$ has width~$4$, and  
$N$ variables and $\bigtheta{N}$ clauses, for $2k^2n^{c+1} \leq N \leq 40k^2(2n)^{c+1}$.

\subsection{The Compressed Game}\label{subsec:VCVR}
We are ready to define the compressed vertex Cop-Robber game.

\begin{definition}[Compressed Cop-Robber game]\label{def:cgame}
Given a graph compression $(G,\vertexeq,\edgeeq)$, the \emph{compressed (vertex) $k$-Cop-Robber game} proceeds on $G$, where the Cops and the Robber stay on vertices and are always visible to each other. Initially, the Robber is at a vertex, %
and all $k$ Cops are lifted from the graph. %
In a round, the following happens in turn: 
\begin{enumerate}[label=(G{\arabic*})]
\item \label{G1} If there is no lifted Cop, choose and lift one. Then, a lifted cop signals a vertex $v$ to the Robber; 
\item \label{G2} The Robber does a \emph{compressible move} from his current vertex $w_1$ to some $w_2$, which means he provides an edge set $M\subseteq E$ such that 
\begin{enumerate}
\item\label{G2.a} ($M$ is closed under $\edgeeq$.) Whenever $e\edgeeq e'$, $e\in M\Leftrightarrow 
e'\in M$; 
\item\label{G2.b} (No vertex class is occupied.) If we denote the set of vertices occupied by Cops 
(the \emph{Cop position}) after \ref{G1} by $\coppos\subseteq V$, then $\coppos\intersection 
V(M)=\emptyset$; 
\item\label{G2.c} (Parity flip.) For all $u\in V(M)$, $\deg_M(u)\text{ is odd if and only if }u\vertexeq 
w_1\text{ or } u\vertexeq w_2$.
(Note that this implies $w_1\neq w_2$.)
\end{enumerate}
\item A lifted Cop lands at the signaled vertex $v$.\label{G3}
\end{enumerate}
The game ends when a Cop is at a vertex in the Robber's vertex class.
\end{definition}

\begin{remark}\label{rmk:game}
By definition of a compressible move $M$, if $M\subseteq E$ is closed under $\edgeeq$, then so is $V(M)$ under $\vertexeq$, so condition \ref{G2.b} is the same as requiring that $\coppos_\vertexeq\intersection V(M)_\vertexeq=\emptyset$. 
In the trivial case where all vertices form singleton equivalence classes, and hence so do all edges, a compressible move from $w_1$ to $w_2$ is a path connecting $w_1$ and $w_2$ plus some edge disjoint cycles that can be removed.
\end{remark}

\subsection{From Narrow Proofs to Cop Strategies} 
\label{subsec:twofacts}
Before diving into the compressed game, let us show two facts about the refutation width of Tseitin formulas and about simulating refutations by a Cop strategy. The arguments are somewhat straightforward.

\begin{lemma}[Small-width refutation]\label{lem:k+3}
The formula $\tau(\basegraph)$ has a resolution refutation with width $k+3$ and size $O((L+r)2^kk)$.
\end{lemma}
\begin{proof}  
We view a clause as its minimum falsifying assignment and construct a proof as a top-down DAG, where we query the value of a variable at each node until arriving at an axiom.  
We first query the values of all right-going edges of column $1$. Denote this edge set by $E_1$. 
Since $E_1$ separates the graph (i.e., its removal disconnects the graph into two components), 
at each of the $2^k$ branches, the restricted formula contains a subset of clauses that form a Tseitin contradiction on either the induced subgraph on column $1$ or the sub-cylinder after column $1$. 
At any of the $2^{k-1}$ nodes where the contradiction is on the 1st column, 
we continue to query the up, down edges of vertex $(1,1)$, 
to arrive at either an axiom at this vertex (clause width $3$) or a node in the DAG where we forget the right edge of $(1,1)$ (clause width $k+1$). 
Then we continue to do so for vertices $(1,2)$, $(1,3),\ldots$ until we cyclically get back from row $k$ to row $1$, where we end up at an axiom at vertex $(1,1)$. Note this part of the proof has width $k+1$ and size $O(2^kk)$.

At any of the other $2^{k-1}$ nodes, the Tseitin instance is on the induced subgraph from the second to the last column. 
We query the up, down, right edges of $(2,1)$ to arrive at either an axiom about this vertex (clause width 4), or, by forgetting its left edge, a node where the assignment has domain $E_{1,1}:=E_1\setminus\set{\text{left edge of $(2,1)$}}\union\set{\text{up, right, down edges of $(1,2)$}}$ (clause width $k+2$). 
Then we query the right edge of $(2,2)$ to arrive at another node (clause width $k+3$), and then query the down edge of $(2,2)$ to arrive at either an axiom about vertex $(2,2)$ or, by forgetting the left and up edges, a node where the assignment is on $H_{1,1}\setminus\set{\text{up, left edges of $(2,2)$}}\union\set{\text{right, down edges of $(2,2)$}}$ (clause width $k+2$). 
We continue in this fashion for vertices $(2,3),(2,4),\dots$ while staying in width $k+3$. At last, we forget 3 edges (up, left, down) of vertex $(2,k)$ to arrive at a node whose assignment is on $E_2$, the set of all right edges from column 2. 
In this process, we maintained that the edges mentioned by a non-axiom clause separate the graph, so at any node whose assignment is on $E_2$, the restricted formula is a Tseitin contradiction on the sub-cylinder after column 2. This part of proof has width $k+3$ and size $O(2^{k}k)$.

The rest of the proof is clear: we repeat the above paragraph to $E_3, E_4,\ldots, E_{L+2r}$ until the 
last column, and on the last column the query processwhere the proof is symmetric to the case one 
we did on the first column $1$.
\end{proof}

We will only consider Tseitin formulas that are \emph{nice with respect to a graph compression} $(G,\vertexeq,\edgeeq)$, which means that the parities at $v,v'$ are the same if $v \vertexeq v'$, and that there is a total assignment $\alpha_\Ts: E(G)/_{\edgeeq}\to\F_2$ that satisfies all axioms except for at one vertex class. 
Note that the instance $\tau(\basegraph)$ is nice with respect to $(\basegraph,\vertexeq,\edgeeq)$ (\autoref{def:compression} and \autoref{def:taubasegraph}), with the total assignment being all-zero and the unique vertex class with odd charge being $(1,1)$, which is a singleton class.

\begin{lemma}[Cops simulate refutation]\label{lem:sim}
The following holds for any graph compression $(G,\vertexeq,\edgeeq)$ and nice $\Ts(G)$ with associated assignment $\alpha_{\Ts}$. If there is a width-$w$ and depth-$d$ resolution refutation of the compressed $\Ts(G)$, then for the compressed $(w+1)$-Cop-Robber game where the Robber starts at a vertex in the unique falsified class of $\alpha_\Ts$, the Cops can win in $d+1$ rounds.
\end{lemma}
\begin{proof}
For ease of notation, in this proof, we write $[e]$ for the $\edgeeq$-class of $e\in\edges{G}$, and $[v]$ for the $\vertexeq$-class of $v\in V(G)$. 
Given a refutation $\refpi$ of the compressed $\Ts(G)$, the Cops travel down the proof DAG $\refpi$ from the empty clause $\emptycl$.
At each clause $\cld$, a new game round begins. Assume that the  Robber is at $v_\cld$ and that the Cop position is $\coppos_{\cld}\subseteq V(G)$. The Cops keep the following invariants.
\begin{enumerate}
	\item $\setsize{\coppos_\cld}\leq w$;
	\item Each variable $x_{[e]}$ in $\cld$ is associated to a vertex $v\in\coppos_\cld$ that is incident to some edge in class $[e]$;
	\item  They have a total assignment $\alpha_\cld$ on $\setdescr{x_{[e]}}{e\in\edges{G}}$ that falsifies only the parity at $v_{\cld}$, and moreover, $\alpha_\cld(\cld)=0$.
\end{enumerate}
Initially, all Cops are lifted and $\alpha_\emptycl=\alpha_\Ts$, so the invariants hold. 
Assume that the process arrives at clause $\cld$ keeping the invariants and the two precedent clauses are $\cld_1$ and $\cld_2$, 
and denote the resolved variable at $D$ by $x_{[e_D]}$. 
Since $\setsize{\coppos_D}\leq w$ and there are $w+1$ Cops in total, there is at least one lifted Cop. 
A lifted Cop then signals a vertex $v$ that is adjacent to some edge in class $[e_D]$. 
The Robber does a compressible move $M$ from $v_\cld$ to some $v'_{\cld}$, and a lifted Cop lands at $v$. 
We choose the total assignment $\alpha'$ to be one such that, for every edge $e$, it flips the value $\alpha_{\cld}([e])$ if and only if $e\in M$. 
This is well-defined since $M$ is a compressible move and by \ref{G2.a}. 
The simulation now proceeds to the precedent clause where $\alpha'$ falsifies the literal over $x_{[e_\cld]}$, say clause $\cld_1$. 

Let $\alpha_{\cld_1}\coloneqq\alpha'$, and let us see that invariant \textbf{(3)} holds. 
First, $\alpha'$ falsifies only the parity constraint at $[v'_\cld]$ where $v'_\cld$ is the Robber's current vertex. 
This is because $\alpha_\cld$ only falsifies the parity at $[v_\cld]$ and $M$ flips only the parities at $[v_\cld]$ and $[v'_\cld]$ by \ref{G2.c}. 
Second, $\alpha'(\cld_1)=0$. This is because $\alpha'$ falsifies the literal over $x_{[e_\cld]}$ in $\cld_1$, and $\alpha=\alpha'$ on the rest of the literals in $\cld_1$ since their underlying variables all have some associated vertex in $\coppos_\cld$ by the inductive hypothesis on \textbf{(2)}, so no edge in these classes can be used in the compressible move by \ref{G2.b}. 

Now we associate the variable $x_{[e_D]}$ to the newly landed Cop, and lift a Cop that is not associated to any variable in $\vars{\cld_1}$, if there is one. Then there must be least one Cop in the helicopter since at most $\setsize{\cld_1}\leq w$ many Cops are associated to some variable in $\vars{\cld_1}$. So, if we let $\coppos_{\cld_1}$ denote the remaining Cops' positions, then $\setsize{\coppos_{\cld_1}}\leq w$, and thus the invariants \textbf{(1)} and \textbf{(2)} hold. This completes the induction, where a resolution step gives a game round. 

Finally, when we reach an axiom clause, the clause is by definition expanded from the parity constraint at a vertex class. By \textbf{(3)} the total assignment we keep falsifies this axiom and thus the parity at this vertex class, so by the same invariant, the Robber's position $u$ is in this class. The variables mentioned in this axiom are precisely the edge classes of all incident edges to $u$, so by \textbf{(2)}, either a Cop is occupying some $u'\vertexeq u$ and they win, or every neighbor of the Robber has a Cop in its class, in which case a lifted Cop can be sent to $u$ in the next round while the Robber has no compressible move, and the Cops win.
\end{proof} 
\section{Lower Bound Proof for
  the
  Cop-Robber
Game}
\label{sec:lbd}
 
In this section, we prove the lower bound for the compressed game (\autoref{thm:main}), stated formally below. 
We use the construction $(\basegraph,\vertexeq,\edgeeq)$ from the previous section,
including the definition of the parameters $\midlength$ and $\earlength$ from 
\autoref{def:parameters}.

\begin{theorem}[Cop-Robber, formal]\label{thm:round}
    For any $k\geq 2$ and $c\leq k-1$, we have that $k+1$ Cops can win the compressed game on 
    $(\basegraph,\vertexeq,\edgeeq)$, but 
    as long as there are at most $k+c$ Cops, the Robber can survive 
    $(\midlength-2\earlength)/(8(k+c))$ rounds.
\end{theorem}

From this theorem, together with the relation between the Cops-Robber game and resolution width and depth (\autoref{lem:sim}) and the upper bound for resolution (\autoref{lem:k+3}), we immediately get the supercritical width-depth trade-off stated in \autoref{thm:main_tradeoff}.

In \autoref{subsec:concepts}, we prepare the concepts to be used in the proof. 
In \autoref{subsec:keylemma}, we describe the idea of the Robber's strategy and prove a key lemma. 
We provide the Robber's strategy to survive many rounds against $k+c$ Cops in 
\autoref{subsec:proof}, proving \autoref{thm:round}. 
Henceforth, we fix $k\geq 2$ and $c\in[1, k-1]$, and consider the graph compression $(\basegraph,\vertexeq,\edgeeq)$ %
with parameters as in \autoref{def:parameters}. 
We write $V\coloneqq \vertices{\basegraph}$, and $E\coloneqq \edges{\basegraph}$.

\subsection{Preparations}\label{subsec:concepts}
We begin with the following special kind of compressible moves of the Robber. 

\begin{definition}[$I$-periodic path]\label{def:I-move} 
Given a nonempty row set $I\subseteq [k]$ and columns $a$ and $b$ in $[1,2\earlength+\midlength]$ where $a<b$, 
let $g_I\coloneqq \gcd(\modulus{i}\mid i\in I)$. 
An \emph{$I$-periodic path} between $a$ and $b$ is a path in the induced subgraph 
$\basegraph|_{I\times [a,b]}$ that can be obtained as follows. 
Take a simple path in $\basegraph|_{I\times[a,a+g_I]}$ from column $a$ to column $a+g_I$ 
such 
that the path starts and ends in the same row, and 
its restrictions to column $a$ and to column 
$a+g_I$, when viewed as subsets of the cycle on $[k]$, 
share no edges. 
Then extend this path 
periodically on $\basegraph|_{I\times [a+g_I,a+2g_I]}$ until hitting column $b$ for the first time. 
\end{definition}

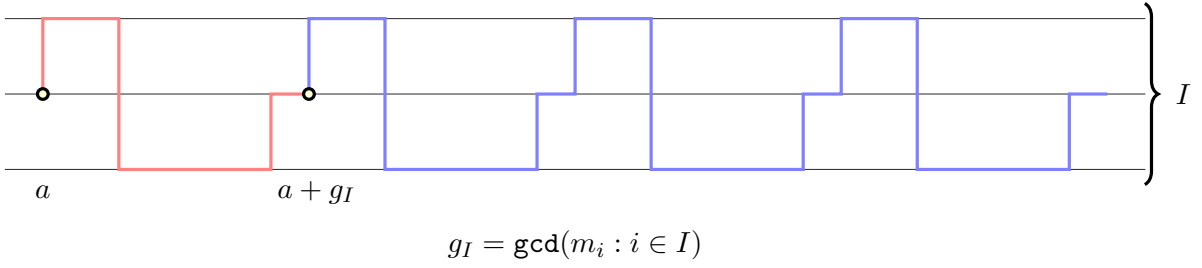
\begin{figure}
	\begin{center}
		\begin{tikzpicture}
		\draw[color=black!80] (-0.5,0) -- (14.5,0);
		\draw[color=black!80] (-0.5,1) -- (14.5,1);
		\draw[color=black!80] (-0.5,-1) -- (14.5,-1);
		
		\draw[very thick, color=red!50] (0,0) -- (0,1) -- (1,1) -- (1,-1) -- (3,-1) -- (3,0) -- (3.5,0); 
		
		\draw[very thick, color= blue!50] (3.5,0) -- (3.5,1) -- (4.5,1) -- (4.5,-1) -- (6.5,-1) -- (6.5,0) -- (7,0) -- (7,0) -- (7,1) -- (8,1) -- (8,-1) -- (10,-1) -- (10,0) -- (10.5,0) -- (10.5,0) -- (10.5,0) -- (10.5,1) -- (11.5,1) -- (11.5,-1) -- (13.5,-1) -- (13.5,0) -- (14,0); 
		
		\filldraw[fill=betterYellow!16,very thick] (0,0) circle (2pt);
		\filldraw[fill= betterYellow!16,very thick] (3.5,0) circle (2pt);
		
		\node[] at (0,-1.3) {$a$};
		\node[] at (3.6,-1.3) {$a+g_I$};
		
		\node[] at (7,-2) {$g_I = \mathtt{gcd}(m_i: i \in I)$};
		
		\node[] at (15,0) {$I$};
		
		\draw [decorate,decoration={brace,amplitude=5pt,mirror,raise=4ex}, very thick]
  (13.8,-1.2) -- (13.8,1.2) node[midway,yshift=-3em]{};
	\end{tikzpicture}
	\end{center}	
	\caption{A special move $P$ starting in column $a$ is shown as the segmented path, where $I$ is the underlying row set.}\label{fig:specialmove}
\end{figure}

Using $I$-periodic paths, we obtain the following Robber moves, generalizing those in \cite{GLNS23CompressingCFI} which correspond to the case where $\setsize{I}=2$.

\begin{remark}[Periodic paths give special compressible moves]\label{rmk:move} 
Suppose $W$ is the set of vertices occupied by Cops in the compressed game (\autoref{def:cgame}). 
If $P$ is an $I$-periodic path ($I\neq\emptyset$) from vertex $v_1$ on column 1 to vertex $v_2$ on column $2\earlength+\midlength$ respectively and it avoids $W$, then $P$ is a compressible move. 
We will call such path a {\it special move} of the Robber (cf.\ \autoref{fig:specialmove}), if he is at $v_1$. 

To see why it is a compressible move, note that $P$ is a complete $g_I$-periodic (horizontally) repetition of some edges in $I\times[1,1+g_I]$. 
Using the fact that $g_I$ divides every row modulus in $I$, 
by~\autoref{item:obs2} we have that
the edges in $P$ are closed under $\edgeeq$, 
thus Condition~\ref{G2.a} holds. 
Condition~\ref{G2.b} holds by assumption. 
Condition~\ref{G2.c} holds since $P$ is a path connecting $v_1$ and $v_2$ which are singleton 
classes in~$\vertexeq$.

In fact, more is true for the above special move $P$. 
For any two columns $a$ and $b$ in the left and right part of $\basegraph$, respectively, 
if $v_a$ is the \emph{first} intersection of $P$ with column $a$ and 
$v_b$ is the \emph{last} intersection of $P$ with column $b$, 
then the subpath of $P$ from $v_a$ to $v_b$ is also a compressible move, 
which we call the \emph{$(a,b)$-truncation of $P$}. 
To see why, note that Condition~\ref{G2.b} is satisfied since we took a subset of $P$. 
The parts of $P$ before $v_a$ and after $v_b$ consist of singleton-class edges, 
so after removing them, the subpath (as an edge set) is still closed under $\edgeeq$, thus Condition~\ref{G2.a} is satisfied. 
Finally, the subpath connects $v_a$ and $v_b$ which are singleton-classes in $\vertexeq$, 
so the parity Condition~\ref{G2.c} is satisfied. 
\end{remark}

Next, we set up some notation for vertex sets and vertex separators. 
For ${W}\subseteq V$ and a row set $I\subseteq[k]$, we let ${W}_I\coloneqq W\intersection (I\times[1,2\earlength+\midlength])$, which is the restriction of ${W}$ on $I$. 
We let ${W}_{\equiv_I}\coloneqq\setdescr{v\in V}{(\exists v'\in{W}_I)\ \text{s.t. } v,v'\text{ are on the same row and their distance is 0 mod $g_I$}}$, which is the vertex set $g_I$-periodically generated by ${W}_I$. 
We call a column \emph{${W}$-free} if it contains no vertex in ${W}$. 
We call a vertex set $I\times [b,b+g_I]\subseteq I\times [1,2\earlength+\midlength]$ a \emph{good period of ${W}_{\equiv_I}$} if the column $b$ is ${W}_{\equiv_I}$-free. 

\begin{definition}[Vertex separators]\label{def:vs-Isep}
We say a vertex set $W\subseteq V$ is a \emph{vertex separator for $I\times[a,b]$} if in the induced 
graph $\basegraph|_{I\times[a,b]}$
there is no path from column $a$ to column $b$ avoiding $W\intersection(I\times[a,b])$. 

We say a vertex set ${W}$ is \emph{$I$-separating} if in the induced graph 
$\basegraph|_{I\times[1,2\earlength+\midlength]}$ there is no $I$-periodic path from column $1$ to 
column $2\earlength+\midlength$ that avoids $W_{\equiv_I}$. 
We say $W$ is \emph{$a$-separating}, where $1\leq a\leq k$, if it is $I$-separating for all $I$ such that $1\leq\setsize{I}\leq a$. 
\end{definition}

We now record a simple fact that we will use later on, and then give some basic properties of separating sets (\autoref{obs:Isep}) and vertex separators (\autoref{prop:diam}).

\begin{fact}\label{fact:separates}
For any vertex set $R\subseteq I\times[a,b]$, where 
$a\leq b$, it holds that $R$ separates columns $a$ and $b$ in the graph 
                  $\basegraph|_{I\times[a,b]}$ 
if and only if $R$ separates columns $1$ and $2\earlength+\midlength$ 
in $\basegraph|_{I\times[1,2\earlength+\midlength]}$. In particular, $R$ is a vertex separator for $I\times[a,b]$ if and only if $R$ is a vertex separator for $I\times[1,2\earlength+\midlength]$.
\end{fact}

\begin{proposition}\label{obs:Isep}
For any $W\subseteq V$ and $\emptyset\neq I\subseteq[k]$, if $\setsize{{W}_I}<g_I$, then the following three statements hold.

\begin{enumerate}
    \item\label{item:Isep1} There is at least one ${W}_{\equiv_I}$-free column, and 
    ${W}_{\equiv_I}$-free columns appear $g_I$-periodically.
    \item\label{item:Isep2} ${W}$ is $I$-separating if and only if $W_{\equiv_I}$ is a vertex separator 
    for each good period of $W_{\equiv_I}$.
    \item\label{item:Isep3} ${W}$ is $I$-separating if and only if ${W}_{\equiv_I}$ is a vertex separator for $I\times[1,2\earlength+\midlength]$.
\end{enumerate}
\end{proposition}

\begin{proof}
To see Item~\ref{item:Isep1}, note that since $\Setsize{{W}_{\equiv_I}\intersection(\text{columns }[1,g_I])}=\setsize{{W}_I}<g_I$, there is a column in $[1,g_I]$ that is ${W}_{\equiv_I}$-free. 
The $g_I$-periodicity of ${W}_{\equiv_I}$-free columns is immediate.
 
To prove Item~\ref{item:Isep2}, assume that
$W$ is not $I$-separating.
Then any $I$-periodic path from column 1 to column $2\earlength+\midlength$ avoiding ${W}_{\equiv_I}$ certainly implies that ${W}_{\equiv_I}$ is not a vertex separator for any $\basegraph|_{I\times [b,b+g_I]}$, as witnessed by a subpath between columns $[b,b+g_I]$. (Since any edge connects vertices in the same or adjacent columns, we can take e.g. the subpath from the last time of visiting column $b$ to the next first time of visiting column $b+g_I$.)

For the other direction, suppose $W_{\equiv_I}$ is not a vertex separator on a good period of $W_{\equiv_I}$. 
By the $g_I$-periodicity of $W_{\equiv_I}$, this is the case for every good period of $W_{\equiv_I}$. 
By Item~\ref{item:Isep1} a good period exists, say it is $I\times[b,b+g_I]$, then ${W}_{\equiv_I}$ is 
not a vertex separator for $I\times[b,b+g_I]$, so there is a path $P=(P(1),\ldots,P(t))$ in 
$\basegraph|_{I\times [b,b+g_I]}$ from column $b$ to $b+g_I$ avoiding ${W}_{\equiv_I}$. By 
truncating $P$ if necessary, we can assume that $P(1)$ is the last vertex of $P$ on column $b$ and $P(t)$ 
is the first vertex of $P$ on column $b+g_I$. Then we extend $P$ on column $b+g_I$ to the same row as 
$P(1)$, which is possible since $I$ is a segment in the row set and column $b+g_I$ is 
${W}_{\equiv_I}$-free. 
Finally, this extended path can be further extended periodically in both directions to columns 1 and 
$2\earlength+\midlength$, guaranteeing that it remains a simple path. The final path witnesses that 
${W}$ is not $I$-separating.

We proceed to Item~\ref{item:Isep3}.
If ${W}_{\equiv_I}$ is a vertex separator for 
$I\times[1,2\earlength+\midlength]$, then there is no path from column~$1$ to column~$2\earlength+\midlength$, so in particular there is no such path that is $I$-periodic and in the induced subgraph, so $W$ is 
$I$-separating. 
For the other direction, assume ${W}$ is $I$-separating. 
By Item~\ref{item:Isep2}, ${W}_{\equiv_I}$ is a vertex separator for every good period of itself. 
In particular, for any good period $I\times [b,b+g_I]$ of ${W}_{\equiv_I}$ (which exists by Item~\ref{item:Isep1}), since 
${W}_{\equiv_I}\intersection I\times [b,b+g_I]\subseteq I\times [b,b+g_I]$
is a vertex separator for $I\times [b,b+g_I]$, 
we get by \autoref{fact:separates} that 
${W}_{\equiv_I}\intersection I\times [b,b+g_I]$ is a vertex separator for $ 
I\times[1,2\earlength+\midlength]$, and hence so is $W_{\equiv_I}$.	
\end{proof}

In the following, we use $\diam(W)$ for the difference of the maximum and minimum indices of columns on which $W$ is nonempty. 
The term \emph{minimal} will always be used with respect to the set-inclusion relation. 

\begin{proposition}\label{prop:diam}
If a set $S$ is a minimal vertex separator for $\basegraph$, 
then $\diam(S)\leq \setsize{S}-1$.
\end{proposition}
\begin{proof}
Suppose not, then there are columns $a<b$ on which $S$ is nonempty and $b-a\geq \setsize{S}$. 
Thus there is a column $c\in(a,b)$ that is $S$-free. The set
$S$ must be a vertex separator for either $\basegraph|_{[k]\times[1,c]}$ or 
$\basegraph|_{[k]\times[c,2\earlength+\midlength]}$, say the former. But then $S_{[k]\times[1,c]}$ 
is also a vertex separator for $\basegraph$, contradicting the minimality of $S$.
\end{proof}

Finally, we define the concept of virtual cordons, which is key to our improved analysis over \cite{GLNS23CompressingCFI}. 
To simplify the notation, given a vertex set $W\subseteq V$ and row $i\in [k]$, we use the abbreviation $W_i$ for $W_{\set{i}}=W\intersection\big(\{i\}\times[1,2\earlength+\midlength]\big)$. 
We use $\unirow(W)\subseteq[k]$ to denote the set of rows where $W$ has at most one vertex, and call them the {\it unique rows of $W$}.

\begin{definition}[Virtual cordons\protect\footnotemark]\label{def:virtual}
\footnotetext{We used the term \emph{semi-separators} for this definition in an earlier manuscript, which was subsequently adopted in the follow-up work~\cite{BerkholzLV2024}.
We think \emph{virtual cordons} is more descriptive.} 
Given a vertex set $W$, a vertex separator $S$ for $\basegraph$ is called a \emph{virtual cordon of $W$} if 
$\setsize{S_i}\leq \setsize{W_i}$ on every row $i\in[k]$, and 
$S_j\subseteq (W_j)_\vertexeq$ for every $j\in \unirow(W)$. 
\end{definition} 

\begin{figure}
  \vspace{-1in}
      \begin{center}
          \begin{tikzpicture}[scale=0.70]
      \def\rows{7}
      \def\cols{7}
  
      \foreach \x in {0,...,\cols} {
          \foreach \y in {0,...,\rows} {
              \coordinate (n\x\y) at (\x,\y) {};
          }
      }

  \node at (2,0) (v20) {}; 
  
          \filldraw[color=white, fill= betterYellow!20, rounded corners]  (1.7,0.2) -- (1.7,-0.2) --(6.3,-0.2) 
          --(6.3,0.2) -- (4,1) 
          -- 
          cycle;
          \filldraw[color=white, fill= betterYellow!20, rounded corners] (4,1) -- (1.7,1.8) -- (1.7,2.2) -- (4,3) -- 
          (6.3,2.2) -- (6.3,1.8) -- 
          cycle;
          \filldraw[color=white, fill= betterYellow!20, rounded corners] (4,3) -- (2,4)  -- (0.7,4.8) -- (0.7,5.2) 
          -- 
          (3,6) -- 
          (6.3,5.1) -- (6.3,3.9) -- cycle;
          \filldraw[color=white, fill= betterYellow!20, rounded corners] (3,6) -- (0.7,6.8) -- (0.7,7.2) -- 
          (5.3,7.2) -- (5.3,6.8) 
          --cycle;

      \draw[style=help lines,step=1cm,color=black!90] (-6,0) grid (13,7);
      \foreach \x in {-6,...,13} {
          \draw[style=help lines,color=black!50,looseness=1.5, dashed] (\x,0) to [out=-99,in=99] node {} 
          (\x,\rows);
      }
  
  \draw [cyan, thick] plot [smooth] coordinates {(2.5,-0.5) (n30) (n41) (n52) (n43) (n34) 
  (n25) (n36) (n27) (1.5,7.5) };
  \draw [red, thick] plot [smooth] coordinates {(5.5,-0.5) (n50) (n41) (n42) (n43) (n54) 
  (n45) (n36) (n47) (4.5,7.5) };

          \filldraw[fill= green!16,very thick] (2,7) circle (6pt);
          \filldraw[fill= green!16,very thick] (4,7) circle (6pt);
          \filldraw[fill= green!16,very thick] (3,6) circle (6pt);
          \filldraw[fill= green!16,very thick] (2,5) circle (6pt);
          \filldraw[fill= green!16,very thick] (4,5) circle (6pt);
          \filldraw[fill= green!16,very thick] (7,5) circle (6pt);
          \filldraw[fill= green!16,very thick] (5,4) circle (6pt);
          \filldraw[fill= green!16,very thick] (3,4) circle (6pt);
          \filldraw[fill= green!16,very thick] (4,3) circle (6pt);
          \filldraw[fill= green!16,very thick] (4,2) circle (6pt);
          \filldraw[fill= green!16,very thick] (5,2) circle (6pt);
          \filldraw[fill= green!16,very thick] (4,1) circle (6pt);
          \filldraw[fill= green!16,very thick] (5,0) circle (6pt);
          \filldraw[fill= green!16,very thick] (3,0) circle (6pt);
      \end{tikzpicture}
      \end{center}    
  \vspace{-1in}
  \caption{
  The green circles are vertices in $W_\equiv$ (i.e., they are in the equivalent classes of the Cops).
  The red and blue curves illustrate two minimal vertex separators contained in $W_\equiv$. 
  The yellow region represents all virtual cordons associated with $W$.}\label{fig:VirtualCordons}
\end{figure}
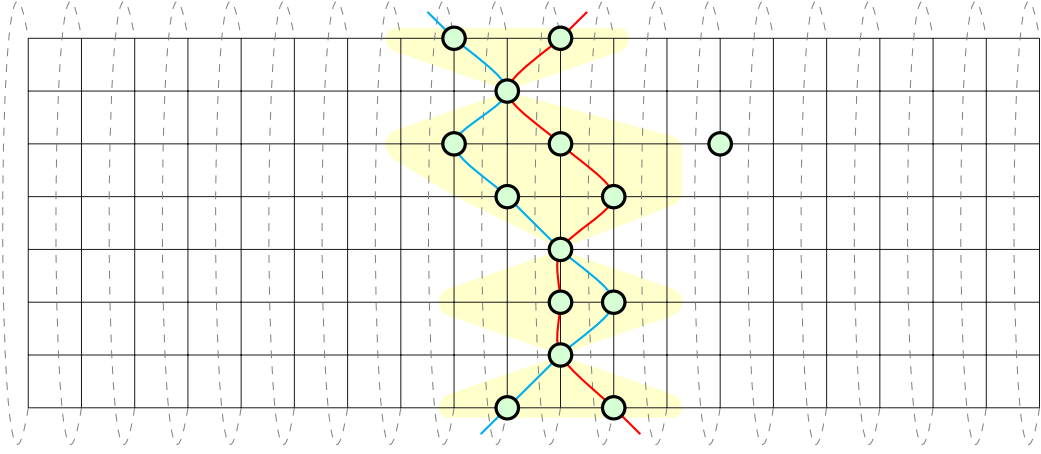

\begin{definition}[Critical set]\label{def:critical}
We say a vertex set $W$ is \emph{critical} (to the Robber) if
\begin{enumerate}
\item\label{item:critical1} $W$ is $(c+1)$-separating, and
\item\label{item:critical2} there exists a virtual cordon of $W$.
\end{enumerate}
\end{definition}

Intuitively, the set of virtual cordons of $W$ can be viewed as a kind of `closure' of the set of vertex separators in %
$W_\equiv$ (see \autoref{fig:VirtualCordons}). 
In particular, if $W$ is the Cop position, then a virtual cordon is a vertex separator which the Cops can potentially occupy while maintaining their positions on the unique rows. 
We will later design a Robber strategy against this closure, where critical Cop positions play an important role in the analysis (see \autoref{subsec:keylemma}).  

Regarding the definition of critical sets, we note that both conditions in are monotone: if $W\subseteq W'$ and $W$ is critical, then so is $W'$. 
We also remark that the $c+1$ in this definition can be replaced with any larger number while the proofs in later sections can be adapted accordingly.
We pick $c+1$ for concreteness. 

We end this subsection with the following property of virtual cordons.

\begin{proposition}\label{prop:coincide}
Let $W\subseteq V$. Suppose that $\setsize{W}\leq k+c$ and that $W$ is $1$-separating. 
Then any two minimal virtual cordons of $W$, should they exist, coincide on $\unirow(W)$.
\end{proposition}
\begin{proof}
Let  $S_1$ and $S_2$ be two minimal virtual cordons of $W$.
Since $W$ is $1$-separating, on each row it contains at least one vertex, meaning that 
\begin{equation}\label{eq:1sep}
\setsize{\text{non-unique rows of $W$}}\leq (k+c)-k=c \eqperiod
\end{equation} 
We list $\unirow(W)$ as $i_1,\ldots,i_t$, increasing in $[k]$. It holds $t\geq k-c$. 
For $j=1,..,t$, let $(i_j,a_j)\in S_1$, and $(i_j,b_j)\in S_2$ denote the unique vertices on row $j$ in 
$S_1$ and $S_2$, respectively.  

If $c=k-1$, then $m_1=\ldots=m_k=L$, meaning there is no compression at all. In this case, the 
proposition trivially holds. 

So we can assume that $c\leq k-2$, and hence $t\geq 2$.
Since minimal virtual cordons of $W$ are minimal vertex separators for 
$\basegraph$ as well, by their 
size condition and \autoref{prop:diam}, we have that $\abs{a_{2}-a_{1}},\abs{b_2-b_1}\leq 
\setsize{W}-1< k+c$. So, denoting by $\Delta_1\coloneqq (a_2-a_1)-(b_2-b_1)$, then 
$|\Delta_1|<2(k+c)$.
Now, since $i_1$ and $i_2$ are unique rows, $a_1-b_1$ and $a_2-b_2$ are multiples of 
$\modulus{i_1}$ and $\modulus{i_2}$, respectively, and thus 
their difference $\Delta_1$ is a multiple of 
$\gcd(\modulus{i_1},\modulus{i_2})$.
Property~{\ref{P1}} implies that $\gcd(\modulus{i_1},\modulus{i_2})\geq 2(k+c)$. 
Since $|\Delta_1|<2(k+c)$ and $\Delta_1$ is a multiple of a number greater than $2(k+c)$, 
it must be that 
$\Delta_1=0$.

The same argument shows that $\Delta_j\coloneqq (a_{j+1}-a_{j})-(b_{j+1}-b_j)$ is zero for 
$j=1,\ldots,t$ (where $t+1$ is taken as $1$). So $a_1-b_1=\ldots=a_t-b_t$ are all equal to some 
number $\Delta$ which is in turn a multiple of each  $\modulus{i_j}$ for $j\in[t]$, 
and hence of 
$M=\lcm(\modulus{i_1},\ldots,\modulus{i_t})$.
Now $t\geq k-c$ so by Property~{\ref{P3}}, %
$M=L$ and so $L\mid \Delta$. 
Since $\Delta=a_1-b_1<2L$, it can only be $0$ or $L$. 
If $\Delta=L$ then one of $a_1,b_1$, 
say $a_1$, is 
a singleton class, %
a singleton class, %
implying that $a_1{\not\equiv}_{V} b_1$, 
which contradicts the condition on virtual cordons on unique rows. 
So $\Delta=0$ %
and thus $a_j=b_j$ for all $j=1,\ldots,t$, proving the proposition.
So $\Delta=0$ %
and thus $a_j=b_j$ for all $j=1,\ldots,t$, proving the proposition.
\end{proof}

\subsection{The Key Lemma}\label{subsec:keylemma}

In the Robber strategy, we will always let the Robber stay on either the left or the right part of the graph. 
If the Cop position $W$ is not critical, the Robber can either use special moves to switch sides (if $W$ is not $(c+1)$-separating), 
or move within the same part to avoid the Cops (if $W$ has no virtual cordons and hence no vertex separators). 
In either case, the Cops cannot catch him. 
The worrying case is when $W$ is about to become critical in a round and the Robber cannot switch sides. 

\autoref{lem:key} below says that this worry is unnecessary. 
Namely, in a round when the Cops go from having no virtual cordon to having one, the Robber has special moves to both sides of the graph.

\begin{lemma}\label{lem:key}
Suppose $W\subseteq V$ is critical and $\setsize{W}\leq k+c$. If a subset $W^-\subseteq W$ has size $\setsize{W}-1$ and is $(c+1)$-separating, then $W^-$ must be critical as well. 
\end{lemma}
\begin{proof}
Note the assumption implies that $W$ is $1$-separating, so as in \eqref{eq:1sep} we have
\begin{equation}\label{eq:unisize}
\setsize{\text{non-unique rows of $W$}}\leq c \eqperiod
\end{equation}
Since $W= W^- \union \set{v}$, for some vertex $v$, we consider two cases:
\begin{enumerate}[label=(\Alph*)]
	\item\label{case:A} $v$ is on a unique row $i$ of $W^-$; or 
	\item\label{case:B} $v$ is on a non-unique row $i$ of $W^-$.  
\end{enumerate}
In both cases, our goal is to construct a virtual cordon of $W^-$. 

\spcnoindent \emph{Case \ref{case:A}.}~ We may assume $i$ is not a unique row of $W$, since 
otherwise $W_{\vertexeq}=W^-_{\vertexeq}$ and the conclusion follows. 
Suppose $v=(i,x)$ is the unique element of $W^-$ on row $i$. 

Take a \emph{minimal} virtual cordon $S$ of $W$. Let $i_1$ and $i_2$ be the two nearest unique 
rows of $W$, above and below $i$ cyclically, which exist by \eqref{eq:unisize} and since $c\leq k-1$. 
Note they are also the two nearest unique rows of $W^-$ above and below $i$ by the case assumption, and these two rows can coincide in the case $c=k-1$. 
Define two cyclically down-going intervals $I_1\coloneqq [i_1,i]$ and $I_2\coloneqq [i,i_2]$. 
Then by \eqref{eq:unisize} again we have %
$(i-i_1)+(i_2-i)\leq c+1$ (hinting that we will eventually apply Property~{\ref{P2}}).
The idea is to construct a virtual cordon $S'$ of $W^-$ starting from $S$, which we will achieve by 
maintaining $S$ on rows outside of $I_1\union I_2$, changing $S$ on $I_1\union I_2$, and then 
``gluing'' them together.

We first define some notation. 
Let $(i_1,a)$ and $(i_2,b)$ be the two unique vertices of $S$ in the corresponding rows. Since $S$ 
is also a minimal separator for $\basegraph$, by \autoref{prop:diam},
\begin{equation}\label{eq:a-b}
\abs{a-b}<k+c \eqperiod
\end{equation}
Consider the graph restricted on $I_1$. Since $\setsize{I_1}\leq c+1$ we have, by the 
$(c+1)$-separating 
assumption, that $W^-$ is $I_1$-separating.  
By Items~\ref{item:Isep1} and \ref{item:Isep2} in \autoref{obs:Isep}, there is a good period of 
$W^-$ 
on $I_1$ containing $(i_1,a)$, in which $(W^-)_{\equiv_{I_1}}$ is a vertex separator, and similarly for 
$I_2$.  
We fix such good periods on $I_1$ and $I_2$ below.

Pick an $S_1$ that is a \emph{minimal} subset of $(W^-)_{\equiv_{I_1}}$ such that $S_1$ is a vertex 
separator in the good period. 
Note that $S_1\intersection  (\text{row $i$})\neq\emptyset$, and since there is a unique element of 
$(W^-)_{\equiv_{I_1}}$ on row $i$ in that good period, say $(i,y)$, 
this element must be the only element of $S_1$ on $i$.
Similarly, pick a minimal subset $S_2$ of $(W^-)_{\equiv_{I_2}}$ so that $S_2$ is a vertex separator 
within the good period there, and let $(i,z)$ be the unique element of $S_2$ on row $i$. 

\begin{claim}\label{claim:xyz}
It holds that $y=z$ and, moreover, that $(i,y)\in (W^-)_\vertexeq$.
\end{claim}
\begin{proof}%
By \autoref{prop:diam} applied to $S_1$ and the induced subgraph on $I_1\times(\text{good 
period})$, we have that
\begin{equation}\label{eq:a-y}
\abs{a-y}<\setsize{S_1}\leq \setsize{W_{I_1}} \eqcomma
\end{equation}
where the last inequality holds because $W_{\equiv_{I_1}}$ within a good period has size $\setsize{W_{I_1}}$.
Similarly, for $I_2$, we have that
\begin{equation}\label{eq:b-z}
\abs{b-z}<\setsize{S_2}\leq \setsize{W_{I_2}} \eqperiod
\end{equation}
Now we pick any $p$ such that $(i,p)\in(W^-)_\vertexeq$. The rest of the proof of the claim is 
similar to that of \autoref{prop:coincide}. 
First we show $y=z$. On the one hand, by definition of $\equiv_{I_1}$, $p-y$ is a multiple of 
$g_{I_1}$. Similarly, $p-z$ is a multiple of $g_{I_2}$. 
So $y-z$ is a multiple of $\gcd(g_{I_1},g_{I_2})=g_{I}$, where $I\coloneqq I_1\union I_2$. 
By Property~{\color{blue}\ref{P1}} of the parameters, we have $g_I \geq 2(k+c)$. 
On the other hand, by \eqref{eq:a-b}, \eqref{eq:a-y}, \eqref{eq:b-z} we have 
\begin{equation}
\abs{y-z}\leq \abs{y-a}+\abs{a-b}+\abs{b-z}\leq (\setsize{W_{I_1}}-1)+(k+c-1)+(\setsize{W_{I_2}}-1)<k+c+\setsize{W}<2(k+c)\eqperiod
\end{equation}
Together this means $y-z=0$, so $y=z$. Next, note that $y-x=z-x$ is a multiple of $g_{I_1}$ and 
$g_{I_2}$ and thus of $g\coloneqq \lcm(g_{I_1},g_{I_2})$, so by Property~{\color{blue}\ref{P2}} $g$ 
is a multiple of $\modulus{i}$. Therefore, $\modulus{i}|y-x$, meaning that 
$(i,y)\in(W^-)_{\vertexeq}$. The claim is proved.
\end{proof}

Let us now also consider $S_3$ which is $S$ restricted to the rows $I_3\coloneqq %
[i_2,i_1]$ (a cyclic interval mod~$k$). 
By  \autoref{fact:separates}, $S_1$, $S_2$ and $S_3$ are 
vertex separators for $\basegraph|_{I_1}$, $\basegraph|_{I_2}$ and $\basegraph|_{I_3}$, 
respectively. 
Note that $S_1,S_2$ and $S_3$ pairwise intersect at a single vertex.
Our goal is to show that their union $S'$ is a vertex separator for $\basegraph$.

\begin{claim}[Gluing vertex separators]\label{claim:glue}
Let $J_1=[j_1,j_2]$ and $J_2=[j_2,j_3]$ be cyclical intervals (modulo~$k$) of rows, where 
$j_1$ and $j_3$ could be the same. 
If $T_1$ is a vertex separator for $\basegraph|_{J_1}$ and
$T_2$ is a vertex separator for $\basegraph|_{J_2}$
such that $j_1,j_2$ and $j_3$ 
are among the unique rows of both $T_1$ and $T_2$, and $T_1$ and $T_2$ share their unique 
vertex 
on row $j$ for all $j\in J_1\intersection J_2$,
then $T_1\union T_2$ is a vertex separator for $\basegraph|_{J_1\union J_2}$. 
\end{claim}

Once the claim is proved, we can use it to glue $S_1$ and $S_2$ and then to glue $S_1\union S_2$ 
and $S_3$ to obtain a virtual cordon of $W^-$, which will complete the proof for Case~\ref{case:A}.

\begin{proof}[Proof of \autoref{claim:glue}] 
We prove it for the case $j_1=j_3$; the case $j_1\neq j_3$ is similar (and simpler). 
Pick any path $P=(P(1),P(2),\ldots ,P(m))$ that starts from the first column in $\basegraph|_{J_1\union J_2}$ and avoids $T_1\union T_2$. 
We show that $P$ cannot reach the last column. 
Suppose $(j_1,x_1)$ and $(j_2,x_2)$ are the unique vertices in both $T_1$ and $T_2$ on the 
corresponding rows. 
Let $t_1,\ldots,t_K\in [m]$ be the times, in increasing order, where $P$ hits rows $j_1$ or $j_2$. 
There is at least one such $t_i$ since if $P$ never hits $j_1$ nor $j_2$ then it stays in either 
$\basegraph|_{J_1}$,
or $\basegraph|_{J_2}$, depending of the starting row,
and cannot reach the last column since either $T_1$ or $T_2$ is a vertex separator there.
We prove by induction on $\ell$ that if $P(t_\ell)$ is on row $j_1$ then it is to the left of $(j_1,x_1)$, 
and similarly 
if it is on row $j_2$ it is to the left of $(j_2,x_2)$. 

Base case ($j=1$): From $P(1)$ to $P(t_1)$ the path stays within $\basegraph|_{J_1}$ or $\basegraph|_{J_2}$, without loss of generality we assume the former and $P(t_1)$ is on row $j_1$. 
If $P(t_1)$ is to the right of $x_1$, then we can extend it on row $j_1$ straight to the last column. 
This extended path is in $\basegraph|_{J_1}$, avoids $T_1$ (since $T_1$ on row $j_1$ is the singleton $\set{(j_1,x_2)}$), 
and reaches from left to right, contradicting the assumption that $T_1$ is a vertex separator for $\basegraph|_{J_1}$.

From $t_{\ell}$ to $t_{\ell+1}$: the inductive hypothesis is that $P(t_\ell)$ is to the left of $(j_1,x_1)$ or $(j_2,x_2)$ depending on the row it is in. Assume without loss of generality it is $j_1$. 
Then again the path from $P(t_\ell)$ to $P(t_{\ell+1})$ falls within one of $\basegraph|_{J_1}$ and $\basegraph|_{J_2}$. 
If $P(t_{\ell+1})$ is to the right of the unique vertex on the same row, then consider a new path from $(j_1,t)$ straight right to $P_{t_\ell}$ then following $P$ to $P(t_{\ell+1})$ and then straight right to column $2\earlength+\midlength$. 
This path falls in one of $\basegraph|_{J_1}$ and $\basegraph|_{J_2}$ (since $P$ from $P(t_\ell)$ to 
$P(t_{\ell+1})$ does) and it avoids $T_1$ and $T_2$, a contradiction. This completes the induction 
step.

Finally, for $P(t_K)$ to $P(m)$ the argument is symmetric to that for $P(1)$ to $P(t_1)$. This shows 
that $P(m)$ cannot be in the last column and the claim follows.  
\end{proof}

\spcnoindent \emph{Case~\ref{case:B}.}~ This case is similar, and in fact simpler. 
Recall that in this case, there is a non-unique row $i$ of $W^{-}$ on which the vertex $v\in 
W\backslash W^-$ lies. 
The only reason that the virtual cordon $S$ of $W$ is not one of $W^-$ is that 
$\setsize{S_{i}}>\setsize{(W^-)_{i}}$, and we change $S$ as follows. 

As in the previous case, the two nearest unique rows $i_1$ and $i_2$ of $W^-$ above and below 
$i$ 
(cyclically) exist. We denote $[i_1,i_2]$ by $I$. 
From the case assumption, $W$ contains at least $3$ vertices in row $i$, so the number of 
non-unique rows in $W$ is at most $c-1$. 
Thus, $I$ contains at most $c+1$ rows. 
Then %
since $W^-$ is $I$-separating, we know there is a minimal vertex separator $S_1$ for 
$\basegraph|_{I}$ such that on each row $j\in I$, $\setsize{(S_1)_{j}}\leq \setsize{W^-_j}$, and that 
$S_1$ has the same vertex as $S$ on row $i_1$. 
By the minimality of $S$ (as a vertex separator for $\basegraph$) and $S_1$ (as a vertex separator 
for $\basegraph|_{I}$) and \autoref{prop:diam}, 
both $\diam(S)$ an $\diam(S_1)$ are at most $k+c$, 
so the vertex of $S_1$ on row $i_2$ and that of $S$ have horizontal distance at most $2(k+c)$, 
which is smaller than $g_{I}$ by \ref{P1}. 
But they both belong to $(W_{i_2})_{\equiv_{I}}$ where $\Setsize{W_{i_2}}=1$, so they coincide. 
Now we can apply \autoref{claim:glue} to glue together $S_1$ and $S_{(I)^c}$ to form a virtual 
cordon of $W^-$.\smallskip

In both cases, we have constructed a virtual cordon of $W^-$, so \autoref{lem:key} is proved.
\end{proof}

\subsection{Proof of Cop-Robber Lower Bound}\label{subsec:proof}
We now prove \autoref{thm:round},
that is, we give a Robber strategy to survive at least $\midlength/4(c+k)$ rounds in the compressed 
game defined with respect to $(\basegraph,\vertexeq,\edgeeq)$, where we choose the moduli and 
parameters in the construction according to \autoref{def:parameters}. Since $\basegraph$ is 
connected, it does not matter where the Robber initially is; we assume without loss of generality that 
he starts at vertex ${(1,1)}$, a singleton class.  

At a high level, the idea is for the Robber to monitor all virtual cordons corresponding to the Cop 
position and consistently avoid them. To show this is possible, we will use \autoref{lem:key} and the 
properties of separators shown in \autoref{subsec:concepts}. 

For the sake of clarity, let us denote the evolution of the Cop position over rounds by  
\begin{equation}
    (\emptyset=\coppos_1^-\to\coppos_1^+)\to\ldots\to(\coppos_t^-\to\coppos_t^+)\to(\coppos_{t+1}^-\to \coppos_{t+1}^+)\to\cdots
\end{equation}
That is, $\coppos_t^-$ is the Cop position in round $t$ after Step~\ref{G1}, 
and $\coppos_t^+$ after Step~\ref{G3}. 
Note that $\coppos_t^-,\coppos_{t+1}^-\subseteq \coppos_t^+$, $\setsize{\coppos_t^+\setminus \coppos_t^-}= 1$, and when the Robber moves, he knows both $\coppos_t^-$ and $\coppos_t^+$.\\
 
\noindent\textbf{Invariants of the Robber strategy.} In each game round, we will keep the following invariants.
\begin{enumerate}[label=(I{\arabic*})]
\item The Robber is on a Cop-free column on either the left or the right part of $\basegraph$;\label{I1}
\item If $\coppos_t^-$ is critical, then any minimal virtual cordon $S$ of $\coppos_t^-$ has \emph{horizontal distance} at least $\midlength/2-4(k+c)\cdot t$ to the part (left or right) of the graph where the Robber is in.\label{I2}
\end{enumerate}
Here, the horizontal distance between $A,B\subseteq V$, denoted by $d_h(A,B)$, is the minimum distance between columns on which $A$ is non-empty and columns on which $B$ is non-empty.

\begin{proof}[Proof of \autoref{thm:round}]
To see that $k+1$ Cops can win the game, the Cops just play as if the game is uncompressed: they first occupy the middle column of $\basegraph$, then keep moving towards the next full column in the Robber's direction. 
It is easy to check that $k+1$ Cops can do this in $k$ rounds while always maintaining that their positions form a vertex separator for $[k]\times[1,2\earlength+\midlength]$.

We are left to prove the round lower bounds against $k+c$ Cops. 
It suffices to show that the Robber can maintain Invariants~\ref{I1} and~\ref{I2} in the first 
$(\midlength-2\earlength)/(8(k+c))$ many rounds (since \ref{I1} implies that the Robber is not 
caught),  
and we prove this by induction on $t<(\midlength-2\earlength)/(8(k+c))$.
The base case $t=1$ is trivial since $\coppos_1^-=\emptyset$, and so \ref{I2} vacuously holds, and 
the Robber is at $(1,1)$ so \ref{I1} holds. 
For the inductive step, suppose the invariants hold for $\coppos_t^-$. 
At the beginning of round $t$, assume without loss of generality the Robber is at vertex $v_t$ in column $a_t$ on the left part of the graph. 
Since both the left and right parts of the graph have $\earlength>k+c\geq \setsize{\coppos_t^+}$ 
columns by \ref{P4}, each part contains a $\coppos_t^+$-free column, say column $a_l$ on the left 
and $a_r$ on the right.  

\spcnoindent \emph{Case (1). $\coppos_{t}^+$ is not critical.}~ Since $\coppos_{t+1}^-\subseteq \coppos_t^+$ and the property of being critical is monotone, 
\ref{I2} holds vacuously for $\coppos_{t+1}^-$, so we only need to maintain \ref{I1}.

The Robber will move to column $a_l$ during Step~\ref{G2} of round $t$. 
Note that if $a_l=a_t$, then the Robber just moves up one row (mod $k$). 
So we consider the case where 
$a_l\neq a_t$.
First assume that $\coppos_{t}^+$ violates Condition~\ref{item:critical1} of critical sets. 
Then we have an $I$-periodic path for some row set $I$, say going from left to right, from column 1 to column $L+2r$. 
If $a_t<a_l$, then we can truncate this path to be between the first time of meeting column $a_t$ and the first time of meeting column $a_l$; the case for $a_t>a_l$ is similar. 
Otherwise, $\coppos_{t}^+$ must violate Condition~\ref{item:critical2} of critical sets. In particular, $(\coppos_t^+)_\vertexeq$ (whose restriction to the left part is the same as that of $\coppos_t^+$) is not a vertex separator for $\basegraph|_{[k]\times[a_t,a_l]}$. Hence, there is a path between the two columns avoiding $(\coppos_t^+)_\vertexeq$.

Since both $a_t$ and $a_l$ are $\coppos_t^-$-free columns, the Robber can move vertically to the 
appropriate row and take the path to column $a_l$, which is a compressible move and \ref{I1} holds 
for $\coppos_{t+1}^-$.

\spcnoindent \emph{Case (2). $\coppos_t^+$ is critical but $\coppos_t^-$ is not.}~ 
By \autoref{lem:key}, $\coppos_t^-$ is not $(c+1)$-separating. 
So there is a compressible move (actually a path) $P''$ from $v_t$ to a vertex in column $a_r$, for 
example by using an $(a_t,a_r)$-truncation of some $I$-periodic path as in \autoref{rmk:move}, plus 
a suitable modification on the starting vertical subpath within column $a_t$. There is another 
compressible move $P'$, again a path, from $v_t$ to $a_l$ by truncating the same $I$-periodic 
path as in \emph{Case (1)}. 
To decide which of the two moves to use, we first estimate the horizontal span of all minimal virtual cordons of $\coppos_t^+$. 
For any two minimal virtual cordons $S$ and $S'$, since $\Setsize{\unirow(\coppos_t^+)}>0$, we 
have for all $v\in S$ and $v'\in S'$ that
\begin{equation}\label{eq:Hspan}
d_h(v,v') \stackrel{\text{Prop. }\ref{prop:coincide}}{\leq}\diam(S)+\diam(S')\stackrel{\text{Prop. 
}\ref{prop:diam}}{<}2(k+c)-1 \eqperiod
\end{equation}
In particular, the column indices of the vertices over all minimal virtual cordons of $\coppos_t^+$ can be contained in an interval $H=[h_1,h_2]$ of length $2(k+c)-1$, 
so $\Abs{h_1-a_l}$ or $\Abs{h_2-a_r}$ is at least $\midlength/2-(k+c)$. 
The Robber then takes compressible move $P'$ or $P''$ so that he stays horizontally at least $\midlength/2-(k+c)> \midlength/2-4(k+c)\cdot t$ far away from any minimal virtual cordon of $\coppos_t^+$. 
Both invariants hold for $\coppos_{t+1}^-$ since $\coppos_{t+1}^-\subseteq \coppos_t^+$.\smallskip

\spcnoindent \emph{Case (3). $\coppos_t^-$ is critical.}~ In this case, $\coppos_t^+$ is also 
critical, and we assume $a_t\leq a_l$ (the case for $a_t>a_l$ is similar). 
Note that $\coppos_t^+$ cannot be a vertex separator for $[k]\times [a_t,a_l]$, since otherwise the 
horizontal distance between the Robber and a minimal virtual cordon will be smaller than 
$\earlength$, contradicting the inductive hypothesis on \ref{I2} as 
$t<(\midlength-2\earlength)/(8(k+c))$.  
So the Robber can move from $v_t$ to column~$a_l$ via a path within the left part and avoiding 
$\coppos_t^-$, which is a compressible move against $(\coppos_t^-)_\vertexeq$, 
and Invariant~\ref{I1} will hold for $\coppos_{t+1}^-$. 

As for \ref{I2}, note that $\unirow(\coppos_t^-)\intersection 
\unirow(\coppos_t^+)=\unirow(\coppos_t^+)$ and that this set must be nonempty. 
Denote the set of minimal virtual cordons of $\coppos_t^-$ by $\alpha^-$, and that of $\coppos_t^+$ by $\alpha^+$.
Fix any $S\in\alpha^-$. We have that $S$ has only one vertex, which we denote by $w_i$, on each 
row $i\in \unirow(\coppos_t^+)$ by the definition of virtual cordons.
By monotonicity, $S$ is also a virtual cordon of $\coppos_t^+$, 
so there is a minimal $S_1\in\alpha^+$ contained in $S$ and, moreover, such that $(i,w_i)\in S_1$.
Now for any two vertices $w$ and $w'$ in some members of $\alpha^-$ and $\alpha^+$, 
respectively, 
similar to \eqref{eq:Hspan} we have that both $d_h(w,w_i)$ and $d_h(w',w_i)$ are smaller than 
$2(k+c)-1$. 
Therefore, since the Robber stays on the same part of the graph, the horizontal distance between 
this part and anything in $\alpha^+$ decreases by at most 
$\maxofexpr[w,w']{\abs{d_h(w,w')}}<4(k+c)$ compared to $\alpha^-$, so Invariant~\ref{I2} holds. 
\smallskip 

This completes the proof of \autoref{thm:round}.
\end{proof}
\section{Lifting Theorem for Treelike Resolution}
\label{sec-WidthVsTreesize}

The aim of this section is to prove the lifting theorem for 
treelike resolution, restated below for convenience.

\treelikeReslifting*

We reformulate this lifting theorem for the equivalent model of \emph{decision trees},
in which the proof becomes more intuitive.
Let us start by introducing some notation.
Let $\relation \subseteq  \set{0,1}^n \times \Output$ be a total search problem. 
A \emph{partial assignment} to the input variables of $\relation$ is a function 
$\rho : [n] \rightarrow \set{0,1,\star}^n$ mapping the variables to $0$, $1$ or
leaving them unassigned, which corresponds to mapping them to $\star$.
We define the \emph{fixed indices of $\rho$} to be $\fix(\rho) = \set{i \in [n] : \rho(i) \neq \star}$
and the \emph{width} of $\rho$ to be $\setsize{\fix(\rho)}$.
Given two partial assignment $\rho, \rho': [n] \rightarrow \set{0,1,\star}^n$,
we say $\rho'$ \emph{extends} $\rho$ if $\rho(i) \in \set{\rho'(i), \star}$ for all $i\in[n]$. 
Similarly, we say $x \in  \set{0,1}^n$ extends $\rho: [n] \rightarrow \set{0,1,\star}^n$, 
 if $\rho(i) \in \set{x_i, \star}$ for all $i\in[n]$ .

We use the following definition of decision tree.

\begin{definition}[Decision DAG and decision tree]\label{def:decisiondag}
    A decision DAG
    solving $\relation \subseteq  \set{0,1}^n \times \Output$ 
    is a rooted DAG where each node $v$ is 
    labelled with a partial assignment $\rho_v$ such that the following hold: 
    \begin{enumerate}
        \item \emph{Root.}~ The root $r$ is labelled with the constant-$\star$ function, 
        that is, $\rho_r(i) = \star$ for all $i\in [n]$.\label{it:root}
        \item \emph{Non-leaf.}~ If $v$ is a non-leaf node then it has two children, the $0$-child and the 
        $1$-child, and it is labelled with some index $\hati$ such that 
        $\rho_v(\hati) = \star$ and, for $\bit\in\set{0,1}$, the partial assignment of its $\bit$-child $v'$
        satisfies $\rho_{v'}(\hati) \in \set{\bit,\star}$ and
        $\rho_{v'}(i) \in \set{\rho_v(i), \star}$ for $i \neq \hati$.
        \label{it:internal}
        \item \emph{Leaf.}~ If $v$ is a leaf then it is labelled with an $o \in \Output$ such that 
       $(x,o) \in \relation$
        for every $x\in \set{0,1}^n$ that extends $\rho_v$.
        \label{it:leaf}
    \end{enumerate}
    The \emph{size} of a decision DAG is the number of nodes it has, the 
    \emph{depth} is the length of the longest root-to-leaf path in the DAG,
    and the \emph{width} is the maximum over $v$ of the width of any $\rho_v$.

A decision tree is a decision DAG where the underlying DAG is a tree.
\end{definition}

Recall that
$\Search(F) \subseteq \set{0,1}^n \times [m]$ is the search problem defined as
$(x, i) \in \Search(F) \Longleftrightarrow C_i(x) = 0$.
The following folklore lemma relates resolution refutations of $F$ to decision DAGs solving 
$\Search(F)$.
\begin{lemma}[Folklore]
	\label{lem:res-decisionDAG}
	Let $F$ be an unsatisfiable CNF formula. There is a width-$w$
	resolution refutation~$\pi$ of $F$ with underlying DAG $G_{\pi}$ if and only if there is a 
     width-$w$
	decision DAG solving $\Search(F)$ with the same underlying DAG $G_{\pi}$.
\end{lemma}

We can now state our lifting theorem for decision trees, from which
\autoref{thm:treelikeReslifting} follows easily.

\begin{theorem}
\label{thm:treelikelifting}
    Let $\relation \subseteq \set{0,1}^n \times \Output$ be a search problem and let $\xorsize\geq 2$. 
    If there is a width-$w$, size-$s$ decision tree for $\relation\circ \xorText_{\xorsize}^n$, then 
    there is a width-$\left({\frac{w}{\xorsize-1}}\right)$, depth-$\log s$ decision tree for $\relation$.
\end{theorem}

\begin{proof}
Given a treelike resolution refutation for $\XOR{\formf}{\xorsize}$, we get by \autoref{lem:res-decisionDAG} a decision tree solving $\Search(\XOR{\formf}{\xorsize})$, which we can turn into one solving $\Search(\formf) \circ \xorText_{\xorsize}^n$ by changing the labels on leaves. 
Namely, if a leaf is labelled a clause in the CNF expansion of $\XOR{C'}{\xorsize}$ for some clause $C'\in \formf$, we change the label to $C'$. 
Then we apply \autoref{thm:treelikelifting} to this decision tree and use \autoref{lem:res-decisionDAG}.
\end{proof}

Before diving into the proof of \autoref{thm:treelikelifting}, we make a simple observation.

\begin{claim}\label{claim:leafytox}
Fix $o\in \Output$ and 
let $\rho$ be a partial assignment to $\set{0,1}^{nm}$ such that 
for every $y\in \set{0,1}^{nm}$ which extend $\rho$,
$(y,o) \in \relation \circ \xorText_{\xorsize}^n$.
Then for any $x \in \set{0,1}^n$ such that 
$x_i =  \rho(i,1) \oplus \ldots \oplus \rho(i,\xorsize)$ for all $i \in \fix(\rho)$, it holds that $(x,o) \in \relation$.
\end{claim}
\begin{proof}
    Let $x \in \set{0,1}^n$ be such that 
   $x_i =  \rho(i,1) \oplus \ldots \oplus \rho(i,\xorsize)$ for all $i \in \fix(\rho)$.
	Let $y$ be the extension of $\rho$ defined as follows: for each $j \in [n] \setminus \fix(\rho)$ set 
	$y_{j,1},\ldots, y_{j,m}$ so that $y_{j,1} \oplus \ldots \oplus y_{j,m} = x_j$. Then 
	$\xorText_{\xorsize}^n(y)= x$. As $(y,o) \in \relation \circ \xorText_{\xorsize}^n$, it follows by definition that $(x,o) 
	= (\xorText_{\xorsize}^n(y),o) \in {\cal S}$.
\end{proof}

We are now ready to prove \autoref{thm:treelikelifting}. 
The intuition is that we solve $\relation$ by simulating a decision tree $T$ for $\relation\circ\xorText_{\xorsize}^n$ top-down. 
At a node $v$ of $T$ that queries $y_{i,j}$, we move directly to the child with a smaller subtree size, unless the number of $y$-variables over $x_i$ assigned by $\rho_v$ reaches the threshold $m-1$; in that case, we query $x_i$.

\begin{proof}[Proof of \autoref{thm:treelikelifting}]
Given a width-$w$, size-$s$ decision tree~$T$ for $\relation \circ \xorText_{\xorsize}^n$, we construct a decision tree~$\widetilde{T}$ for $\relation$
of width at most ${\frac{w}{\xorsize-1}}$ and depth at most
$\log s$. 
Recall that for every node~$v$ of~$T$
we have a partial assignment 
$\rho_v: [n] \times [\xorsize] \rightarrow \{0,1,\star\}^{n\xorsize}$,
where for simplicity
we view the partial assignments to the input variables of $\relation \circ \xorText_{\xorsize}^n$ 
as maps from the domain $[n] \times [\xorsize]$ instead of $[n\xorsize]$.
We also view the query labels in $T$ as pairs $(\hati,\hatj) \in [n] \times [\xorsize]$.
We define~$\widetilde{T}$ by describing, for any
sequence of query answers, a root-to-leaf path
in~$\widetilde{T}$ with the labels for each node in the path. %

To define this root-to-leaf path, 
we start at the root of~$T$ and walk down to a leaf following some rules we describe
below and occasionally querying input variables of $\relation$ and 
creating new nodes in~$\widetilde{T}$.
At every step, we are at some node $v$ in $T$ with a corresponding partial assignment $\rho_v$, 
and at some node $\widetilde{v}$ in~$\widetilde{T}$ without any labels yet.
In the beginning we are at the root of both~$T$ and~$\widetilde{T}$.
For each $v$ in $T$ that we traverse, 
we inductively define a partial assignment $\sigma_v: [n] \rightarrow \set{0,1,\star}^n$
that will guide us in choosing the path in $T$ to follow and in defining the path
in~$\widetilde{T}$.
For the root~$r$ of~${T}$, we let~$\sigma_{r}$ be the constant-$\star$ function.
Suppose we are at some non-leaf node $v$ 
in $T$ with query label $(\hati,\hatj)$,
and at some node $\widetilde{v}$ in $\widetilde{T}$. We distinguish three cases.
\begin{enumerate}
    \item \label{case1:halving}\textit{(Halving case) If $\setsize{\set{j \mid  
    \rho_v(\hati,j) = \star}}\geq 2$}:~
    In $T$, we move to the child of $v$ that is the root of the smallest subtree, breaking ties 
    arbitrarily. We do nothing in $\widetilde{T}$.

    \item \label{case2:forced}\textit{(Forced case) If $\setsize{\set{j \mid 
    \rho_v(\hati,j) = \star}} \leq 1$ 
    and  
    $\sigma_v(\hati) \neq \star$}:~
    In $T$, we move to the $\bit$-child of $v$, where $\bit\in\set{0,1}$ is such that
    $\bit \oplus \bigoplus_{j\neq \hatj} \rho_v({\hati,j}) = \sigma_v(\hati)$.
    We do nothing in $\widetilde{T}$.

    \item \label{case3:query}\textit{(Query case) Otherwise, $\setsize{\set{j \mid \rho_v(\hati,j) = 
    \star}} \leq 1$  and $\sigma_v(\hati) = \star$}:~
    In $\widetilde{T}$, we label $\widetilde{v}$ with the partial assignment
    $\rho_{\widetilde{v}} = \sigma_{v}$ and
    with the query label~$\hati$. We then query the $\hati$th variable. Let $\widetilde{\bit}$ be the 
    result of this query. 
    We create and move to 
    the $\widetilde{\bit}$-child of $\widetilde{v}$.
    In~$T$, we move to the $\bit$-child of $v$, where $\bit\in\set{0,1}$ is such that
    $\bit\oplus \bigoplus_{j\neq \hatj} \rho_v({\hati,j}) =  \widetilde{\bit}$. 
\end{enumerate}
Let $v'$ be the child of $v$ we chose to move to.
If $v'$ is chosen in case~\ref{case1:halving} or~\ref{case2:forced}, then
\begin{equation}\label{def:case12sigma}
\sigma_{v'}(i) =
\begin{cases}
\sigma_{v}(i), &\text{ if } \setsize{\set{j \mid \rho_{v'}(\hati,j) = \star}} \leq 1;\\
\star, &\text{ otherwise.}
\end{cases}
\end{equation}
and if $v'$ is chosen in case~\ref{case3:query}, then
\begin{equation}\label{def:case3sigma}
\sigma_{v'}(i) =
\begin{cases}
\widetilde{\bit}, &\text{ if } i = \hati \text{ and } \setsize{\set{j \mid \rho_{v'}(\hati,j) = 
\star}}\leq 1;\\
\sigma_{v}(i), &\text{ if } i \neq \hati \text{ and } \setsize{\set{j \mid \rho_{v'}(\hati,j) = \star}}\leq 1;\\
\star, &\text{ otherwise.}
\end{cases}
\end{equation}

When we reach a leaf~$v$ of~$T$ labelled with $o\in\Output$, 
we label the node $\widetilde{v}$ we are at in $\widetilde{T}$ with the answer $o$
and with the partial assignment $\rho_{\widetilde{v}} = \sigma_{v}$,
and the path ends (i.e., $\widetilde{v}$ is a leaf).
Note that this process defines a root-to-leaf path in~$\widetilde{T}$, with 
labels on all nodes in the path, and thus completes the description of~$\widetilde{T}$.
We now need to argue that~$\widetilde{T}$ solves $\relation$,
has width at most ${\frac{w}{\xorsize-1}}$ and depth at most
$\log s$.

To see why~$\widetilde{T}$ is a decision tree solving $\relation$, 
we need to show it satisfies items~\ref{it:root}-\ref{it:leaf} in \autoref{def:decisiondag}. 
For the first two items we use the following observation.

\begin{claim}\label{claim:sigmamonotone}
If at any given point we are at a node $v$ in $T$ and at a node
$\widetilde{v}$ in~$\widetilde{T}$, then $\sigma_{v}$ extends $\rho_{\widetilde{v}}$.
\end{claim}
\begin{proof}
Let $u$ be the last node in $T$ such that we are at $u$ in $T$ and 
at $\widetilde{v}$ in~$\widetilde{T}$ at the same time.
By definition of $\rho_{\widetilde{v}}$ we have that $\rho_{\widetilde{v}} = \sigma_{u}$.
This claim follows since all nodes between $v$ and $u$
are chosen in case~\ref{case1:halving} or~\ref{case2:forced},
which implies $\sigma_v$ extends $\sigma_{u}$.
\end{proof}

Let~$r$ be the root of~$T$ and $\widetilde{r}$ be the root of~$\widetilde{T}$.
Item~\ref{it:root} is satisfied since, by \autoref{claim:sigmamonotone}, $\sigma_r$ extends 
$\rho_{\widetilde{r}}$, and $\sigma_r$ is the constant-$\star$.
Now, let $v'$ be the child of $v$ 
and $\widetilde{v}'$ be the $\widetilde{b}$-child of $\widetilde{v}$
chosen in case~\ref{case3:query} after querying the $\hati$th variable.
First note that $\rho_{\widetilde{v}}(\hati) = \sigma_{v}(\hati) = \star$ by definition of
$\rho_{\widetilde{v}}$ and since we are in case~\ref{case3:query}.
Observe, moreover, that
$\sigma_{v'}(\hati) \in \set{\widetilde{b}, \star}$ and, for $i\neq \hati$, 
that $\sigma_{v'}(i) \in  \set{\sigma_{v}(i), \star}  = \set{\rho_{\widetilde{v}}(i), \star}$. 
These two observations, together with
the fact that, by \autoref{claim:sigmamonotone}, 
$\rho_{\widetilde{v}'}(i) \in \set{\sigma_{v'}(i), \star}$ for all $i\in [n]$,
implies $\widetilde{T}$ satisfies item~\ref{it:internal}.

To see that $\widetilde{T}$ satisfies item~\ref{it:leaf}  in 
\autoref{def:decisiondag},
we need the following claim.

\begin{claim}\label{claim:allfixed}
For all nodes $v$ we traverse in $T$, it holds that:
\begin{align}\label{inv:allfixed}
\sigma_v(i) = \rho_v(i,1) \oplus \ldots \oplus \rho_v(i,\xorsize) & & \forall i\in \fix(\rho_v) \eqperiod
\end{align}
\end{claim}
\begin{proof}
We prove by induction on the distance from the root that the claim holds.
It is not hard to see that the claim holds 
for the root $r$ of $T$, since $\rho_r$ is the constant-$\star$ function.
Now, suppose we are at node~$v$ in~$T$ where~\refeq{inv:allfixed} holds, 
and let $v'$ be the child of $v$ we choose to move to.
Let $i \in \fix(\rho_{v'})$ (if there is no such $i$ then the claim holds trivially),
which implies that $\setsize{\set{j \mid \rho_{v'}(i,j) = \star}} = 0$.
Let $(\hati,\hatj)$ be the query label of $v$.

If $i\neq \hati$, by definition of a decision tree we have that
$\setsize{\set{j \mid \rho_{v'}(i,j) = \star}} \geq \setsize{\set{j \mid \rho_v(i,j) = \star}}$,
which implies that 
$\setsize{\set{j \mid \rho_v(i,j) = \star}} = 0$ and thus $i\in \fix(\rho_v)$.
We can therefore conclude that
\begin{align}
\sigma_{v'}(i) &= \sigma_v(i) & \text{ by definition of $\sigma_{v'}$ since $i \in \fix(\rho_{v'})$}\\
&=\rho_v(i,1) \oplus \ldots \oplus \rho_v(i,\xorsize) & \text{ by the induction hypothesis}\\
&=\rho_{v'}(i,1) \oplus \ldots \oplus \rho_{v'}(i,\xorsize) 
& \text{since $i \in \fix(\rho_{v'})$ and $\rho_{v'}(i,j) \in \set{\rho_{v}(i,j), \star}$.} 
\end{align}

Similarly, if $i= \hati$, by definition of a decision tree we have that
$\setsize{\set{j \mid \rho_{v'}(i,j) = \star}} \geq \setsize{\set{j \mid \rho_v(i,j) = \star}} - 1$,
which implies that 
$\setsize{\set{j \mid \rho_v(i,j) = \star}} \leq 1$.
This means we are either
in case~\ref{case2:forced} or~\ref{case3:query}.
Let $\bit$ be such that $v'$ is the $\bit$-child of $v$.
If we are in case~\ref{case2:forced}, let $\widetilde{b} = \sigma_{v}(i)$.
We can conclude that
\begin{align}
\sigma_{v'}(i) &= \widetilde{b} & \text{ by definition of $\sigma_{v'}$ since $i \in \fix(\rho_{v'})$}\\
&= \bit \oplus \bigoplus_{j\neq \hatj} \rho_v({i,j})  & \text{ by choice of $v'$}\\
&=\rho_{v'}(i,1) \oplus \ldots \oplus \rho_{v'}(i,\xorsize) \eqcomma
\end{align}
where for the last equality we use that since $\rho_{v'}(i,j) \neq \star$ for all $j\in[\xorsize]$,
it must be the case that $\rho_{v'}({i,\hatj}) = \bit$ and 
for all $j\neq\hatj$ that $\rho_{v'}({i,j}) = \rho_{v}({i,j})$.
\end{proof}

We now argue $\widetilde{T}$ satisfies item~\ref{it:leaf} of \autoref{def:decisiondag}.
Let $\widetilde{v}$ be a leaf of $\widetilde{T}$ labelled $o\in\Output$,
and let $v$ be the leaf in~$T$ that we reach when we are still at $\widetilde{v}$.
By \autoref{claim:allfixed}, we have that 
$\sigma_{v}(i) = \rho_{v}(i,1) \oplus \ldots \oplus \rho_{v}(i,\xorsize)  $ for all $i\in \fix(\rho_{v})$.
Therefore, since $\rho_{\widetilde{v}}(i) = \sigma_{v}(i)$, \autoref{claim:leafytox} implies
that every $x\in \set{0,1}^n$ that extends $\rho_{\widetilde{v}}$
is such that $(x,o) \in \relation$, and thus
item~\ref{it:leaf} is satisfied.

To see that $\widetilde{T}$ has width
at most $\frac{w}{\xorsize-1}$, note that
for all nodes $v$ we traverse in $T$, and
for all $i\in[n]$ it holds that  if $\sigma_v(i) \neq \star$,
then $\setsize{\set{j \mid \rho_v(i,j) = \star}} \leq 1$.
Indeed, for the root $r$ of $T$ this holds since $\sigma_{r}$ 
is the constant-$\star$ function and for 
all other nodes $v'$ it holds by the definition of $\sigma_{v'}$.
This implies that for all nodes $v$ we traverse in $T$ the partial assignment
$\sigma_{v}$ has width at most $\frac{w}{\xorsize-1}$.
The bound on the width of $\widetilde{T}$ follows since
for all nodes $\widetilde{v}$ we traverse in $\widetilde{T}$,
$\rho_{\widetilde{v}} = \sigma_{v}$ for some
$v$ we traverse in $T$.

It remains to argue that $\widetilde{T}$ has depth at most $\log s$.
To this end,
first note that before a variable is queried in $\widetilde{T}$
for the first time, there are $m-1$ nodes $v$ in $T$
that fall into case~\ref{case1:halving}.
Moreover, if $\hati$ is queried in $\widetilde{T}$
when we are at a node $v$ %
and also %
at a subsequent node $u$ in~$T$, 
then there must be at least one node $w$ between $v,u$
such that $w$ falls into case~\ref{case1:halving}.
Indeed, first note that 
 there must be some $w'$ in between $v$ and $u$ (including possibly $v$),
such that $\setsize{\set{ j \mid \rho_{w'}(\hati,j) = \star}} \geq 2$, otherwise
$\sigma_u(\hati) \neq \star$ or $\setsize{\set{ j \mid \rho_{u}(\hati,j) = \star}} \geq 2$,
which contradicts the fact that $u$ falls into case~\ref{case3:query}.
This implies that there must be a $w$ in between $w'$ and $u$ (including possibly $w'$),
that falls into~\ref{case1:halving} so that $\setsize{\set{ j \mid \rho_{u}(\hati,j) = \star}} \leq 1$.
This implies that for every query---i.e., every time
the depth of $\widetilde{T}$ increases by $1$---the 
size of the subtree of $T$ rooted at the current node $v$ has decreased 
by at least $1/2$ and thus
the depth of $\widetilde{T}$ is at most $\log s$.
\end{proof}

We conclude this section by noting that, if we increase the threshold 
$\setsize{\set{j \mid \rho_{v'}(\hati,j) = \star}} \leq 1$ to $\setsize{\set{j : \rho_{v'}(\hati,j) = \star}} \leq c$ for %
some $c<\xorsize$, a similar proof gives us the following statement. 

\begin{theorem}%
    Let $\relation \subseteq \set{0,1}^n \times \Output$ be a search problem and let $\xorsize> c \geq 
    1$ and $d\geq 1$. If there is a width-$w$, size-$2^{cd}$ decision tree for $\relation\circ 
    \xorText_{\xorsize}^n$, then there is a width-$\left({\frac{w}{\xorsize-c}}\right)$, depth-$d$ 
    decision tree for $\relation$.
\end{theorem}

\section{Lifting Theorem for Resolution}
\label{sec-SizeVsDepth}

In this section, we prove a lifting theorem from resolution width to resolution size: 

\sizelifting*

Let~$\formf$ be a CNF formula over variables $z_1, \ldots, z_{\indexsize}$.
Recall that we define $\indexing{m}{\formf_n}$ to be the CNF formula 
obtained by substituting in $\formf$ every occurrence of $z_i$ by 
\begin{equation}
(x_{i,1} \rightarrow y_{i,1}) \land \ldots \land (x_{i,\indexsize} \rightarrow y_{i,\indexsize}) \eqcomma
\end{equation}
expanding out to CNF, and including a $3$-CNF formula encoding 
$x_{i,1} \lor \ldots \lor x_{i,\indexsize}$ for every $i\in[n]$.
We will call the set of variables $\{y_{i,1},\ldots,y_{i,\indexlength}\}$ the \textit{$y$-block over 
$z_i$}, the \textit{$i$th $y$-block}, or simply a \textit{$y$-block}. 
We say that clause $\clc$ mentions a $y$-block if a literal over a variable in that $y$-block appears 
in $\clc$.

To prove \autoref{thm:size_lifting}, we define a random restriction $\restr$ of the lifted formula (\autoref{def:restr-size}). We will ensure that a refutation $\proofstd$ of the lifted formula becomes a refutation of the original formula $\restrict{\proofstd}{\restr}$, and that the latter is as narrow as claimed in the theorem. Note that %
restriction certainly does not increase the depth of the refutation. 

\begin{definition}[Random restriction $\restr$ on variables in $\indexingformulawithoutlength$]\label{def:restr-size}
    For every $i \in [n]$,  pick a uniform random $j_i\in [\indexlength]$. Then   
    set $x_{i,j_i} = 1$ and set $x_{i,j'} = 0$ for all $j'\neq j_i$. Moreover, set the extension 
    variables encoding the clause $x_{i,1} \lor \ldots \lor x_{i,\indexsize}$ as a $3$-CNF
    in the way that all clauses in this $3$-CNF are satisfied.
    Finally, set $y_{i,j}$ to $\{0,1\}$ uniformly at random for each $j\neq j_i$, leaving $y_{i,j_i}$ 
    unassigned.
\end{definition}

Note that $\restr$ always fixes all $x$-variables and extension variables to $\{0,1\}$, and the 
restricted formula $\restrict{\indexing{m}{\formf}}{\restr}$ is exactly $\formf$ after variable renaming 
$z_i\leftarrow y_{i,j_i}$ for all $i$. 
So after applying $\restr$, a refutation $\proofstd$ of $\indexingformulawithoutlength$ becomes a refutation of $\formf$. 
It remains to prove that for some $\restr$ in \autoref{def:restr-size}, $\restrict{\proofstd}{\restr}$ is of width at most $\widthbndsize$.
We start with a claim bounding the probability that a clauses has width exactly $\sumindext$ after restriction.

\begin{claim}\label{claim:blocks-after-restr}
If a clause $\clc$ over $\var(\IND(F))$ mentions $\numberofblocks$ many $y$-blocks, then for any $\sumindext\leq\numberofblocks$, 
   \begin{align}\label{eq:blocks-after-restr}
        \Pr_{\restr}[\lowercasewidthofarg{\restrict{\clc}{\restr}} = \sumindext] \leq 
        {\numberofblocks \choose \sumindext}\left(\frac{1}{\indexlength} \right)^{\sumindext}\left(\frac{1 - \frac{1}{\indexlength}}{2}\right)^{\numberofblocks -\sumindext}.
    \end{align}
\end{claim}
\begin{proof}
Without loss of generality, assume there are $n$ many $z$-variables and the $y$-blocks mentioned in $C$ are over $z_1,\ldots,z_\numberofblocks$. 
The event $\lowercasewidthofarg{\restrict{\clc}{\restr}}=\sumindext$ implies that there is an $I\subseteq[\numberofblocks]$, $|I|=t$, such that: 
\begin{itemize}
    \item For every $y$-block in $I$, the $y$-variable unassigned by $\restr$ appears in $\clc$, and all other $y$-literals in $C$ in that block are  set to $0$. We call such a $y$-block \emph{contributing}; 
    \item For every $y$-block in $[\numberofblocks]\backslash I$, all literals in $C$ over $y$-variables in that block are set to 0. We call such a $y$-block \emph{non-contributing.}
\end{itemize}
For each fixed $I\subseteq[r]$ of size $t$, we show: 
\begin{equation}\label{eq:fixedI}
   \Pr_{\restr}
    \begin{bmatrix}
  \text{all $y$-blocks in $I$ are contributing, and}\\\text{ all $y$-blocks in $[\numberofblocks]\backslash I$ are non-contributing}
  \end{bmatrix}
  \leq \left(\frac{1}{\indexlength} \right)^{\sumindext}\left(\frac{1 - \frac{1}{\indexlength}}{2}\right)^{\numberofblocks -\sumindext}. 
\end{equation}
Then \eqref{eq:blocks-after-restr} follows by a union bound over $I$. 

Recall that $\restr$ consists of independent components $\restr_1,\ldots,\restr_n$ where $\restr_i$ is the part of $\restr$ on variables over $z_i$. 
We can write $\clc=\clc_1\lor\ldots\lor\clc_\numberofblocks\lor \clc'$ where $\clc_i=\clc_{i,x}\lor\clc_{i,y}$ contains the literals in $C$ over variables over $z_i$. 
By independence of the $\rho_i$'s, \eqref{eq:fixedI} would follow if we show that $\forall i\in[\numberofblocks]$, 
\begin{align}
\Pr_{\restr_i}[\text{$y$-block $i$ is contributing}]&\leq {1\over \indexlength}\label{eq:contributing},\\
\Pr_{\restr_i}[\text{$y$-block $i$ is non-contributing}]&\leq \frac{1-{1\over\indexlength}}{2}.\label{eq:surviving}
\end{align}
Assume $C_{i,y}$ mentions $a$ many $y$-variables, then $a\geq 1$. 
To see \eqref{eq:contributing}, the probability that one of these $y$-variables is left intact by $\rho_i$ while others are all set to $0$ is $a\cdot \frac{1}{\indexlength}\cdot (\frac{1}{2})^{a-1}\leq \frac{1}{\indexlength}$. 
To see \eqref{eq:surviving}, the probability that all of them are set to 0 is $\frac{\indexlength-a}{\indexlength}\cdot(\frac{1}{2})^{a}\leq \frac{1-1/\indexlength}{2}$. 
The claim follows as described. 
\end{proof}
Next, we bound the probability that a clause has width at least $\widthlem$ after restriction.
\begin{lemma}\label{lem:size-clausewidth}
    For any clause $\clc$ 
    over $\var(\IND(F))$, %
    \begin{align}
        \Pr_\restr[\lowercasewidthofarg{\restrict{\clc}{\restr}} \geq \widthlem]\leq \left(\frac{2}{\indexlength + 1}\right)^{\widthlem}.
    \end{align}
\end{lemma}
\begin{proof} 
Suppose $C$ mentions $r$ many $y$-blocks ($r\geq w$). By \autoref{claim:blocks-after-restr}, 
    \begin{align}
        \Pr_\restr[\lowercasewidthofarg{\restrict{\clc}{\restr}} \geq \widthlem]
        &\leq \sum_{i = w}^r \binom{r}{i} \left(\frac{1}{m}\right)^i \left(\frac{1-1/m}{2}\right)^{r-i}\\
        &= \left(\frac{2}{m+1}\right)^{w} \sum_{i = w}^{r} \binom{r}{i} 
        \left(\frac{m+1}{m}\right)^i \left(1-\frac{1}{m}\right)^{r-i}
        \left(\frac{1}{m+1}\right)^{i-w} 
        \left(\frac{1}{2}\right)^{r-i+w}\\
        & \leq \left(\frac{2}{m+1}\right)^{w} \left(\frac{1}{2}\right)^{r} \sum_{i = w}^{r} \binom{r}{i} 
        \left(1+\frac{1}{m}\right)^{i} \left(1-\frac{1}{m}\right)^{r-i}\\
        & \leq \left(\frac{2}{m+1}\right)^{w},
    \end{align}
    where we use $\frac{1}{m+1} \leq \frac{1}{2}$ for the second to last inequality, and the Binomial 
    theorem for the last one.
\end{proof}
With \autoref{lem:size-clausewidth} at hand, we can finish the proof of \autoref{thm:size_lifting}.
\begin{proof}[Proof of \autoref{thm:size_lifting}]
    Applying a union bound over all clauses in the refutation $\proofstd$ of $\indexingformulawithoutlength$ yields 
    \begin{align}
       \Pr[\lowercasewidthofarg{\restrict{\proofstd}{\restr}}>\widthbndsize] & \leq
        \sum_{\clc \in \proofstd} \Pr[\lowercasewidthofarg{\restrict{\clc}{\restr}} \geq \widthbndsize + 1] 
       \\
       & \leq \sizestd \left(\frac{2}{\indexlength +1}\right)^{\widthbndsize +1}<1
   \end{align}
  where we used \autoref{lem:size-clausewidth} in the second step. 
\end{proof}

\section{Lifting Theorem
  for Cutting Planes and
  Circuits
  via Triangle DAGs
}
\label{sec-MonSizeVsDepth}

In this section, we prove our triangle-DAG lifting theorem, restated below,
by which we obtain our results for cutting planes and monotone circuits
as described in \autoref{sec:prelim}.

\dagLifting*

Note that, for large enough $m$, this lifting theorem is nearly tight. 
Indeed, if $w$ is the smallest width  of any resolution 
proof of $F$ then there always exists a rectangle-DAG solving $\Search(F)\circ \IND^n_m$ of size at 
most $2(mn)^w$, which is at most $m^{1.01w}$ if $m\geq n^{300}$. On the other hand, by setting 
$\delta=\frac{1}{120}$, the lifting theorem implies that there is no triangle-DAG 
(and hence no rectangle-DAG) 
of size  $m^{0.99w}$. 

Our proof of \autoref{thm:dagLifting} will follow the high-level strategy of~\cite{GGKS20MonotoneCircuit}, which consists of two steps: 
\begin{enumerate}
    \item For each node in the triangle-DAG, construct a set of clauses of width at most $w$. 
    The clauses at the root include $\perp$, and the clauses at leaves are weakenings of the initial 
    clauses of~$F$. 
    \item Argue that the clauses that were constructed for each non-leaf node in the DAG can be 
    derived from the clauses of its children in low width and depth.
\end{enumerate}
The novel aspect of our proof lies in the first step. In previous DAG-like lifting 
theorems~\cite{GGKS20MonotoneCircuit, lmmpz2022liftingWithSunflowers}, each triangle is 
partitioned into a set of {\it structured rectangles} (from which a low-width clause can be extracted) 
and an error set. We forgo this approach in favor of covering the triangle by a set of 
\emph{strips}---%
a %
large set of rows of the triangle which can be broken up into {\it pre-structured} rectangles
(from which we can also extract low-width clauses). This 
allows us to more aggressively identify structured rectangles, decreasing the size of the set of 
errors. 
With this improvement, we can complete the proof using two further elements that are largely 
established in the literature: a %
\fullImg similar to that in \cite{lmmpz2022liftingWithSunflowers} (with a simpler and self-contained 
proof), and a careful width analysis in the second step.
In addition, %
our approach naturally gives a unified proof of both the rectangle- and triangle-DAG lifting theorems.

\subsection{Set Up: \fullImg and \triLemma}

The proof of \autoref{thm:dagLifting} relies on two technical lemmas: the  \fullImg and
the \triLemma. Before stating them, we introduce some notation.
For sets $A \subseteq B\neq\emptyset$, the \emph{density} of $A$ in $B$ is $|A|/|B|$. For a random variable $\pmb x$ over a finite set $C$, its \textit{min-entropy} is 
\begin{equation}
    H_{\infty}(\pmb x):=\log \frac{1}{\max_{c\in C}\Pr\left[\pmb x=c\right]}.
\end{equation}
For a set $C$, denote by $\pmb C$  a random variable distributed uniformly over $C$, and note that $H_\infty(\pmb C) = \log |C|$.
For $x \in [m]^n$, $y \in \set{0,1}^{mn} = \left(\set{0,1}^{m}\right)^n $ and $I \subseteq [n]$, 
we write $x_I \in [m]^{\setsize{I}}$ to denote the projection of~$x$,
and $y_I \in \set{0,1}^{m\setsize{I}}$ to denote the projection of $y$, to the coordinates $I$.
Similarly, for $A \subseteq [m]^n$ or $A \subseteq \set{0,1}^{mn}$,
we denote by $A_I$ the projection of $A$ to $I$, and
$\pmb{A}_I$ denotes the marginal distribution of $\pmb A$ 
on the coordinates $I$. 

For most of %
this section %
the structures of $[m]^n$ and $\{0,1\}^{mn}$ 
are immaterial 
and we can simply view $[m]^n\times\{0,1\}^{mn}$ %
as a product of two sets. Under this perspective, we refer to an $x\in[m]^n$ as a {\it row} and a 
$y\in\{0,1\}^{mn}$ as a {\it column}. 
For any set $S \subseteq [m]^n \times \{0,1\}^{mn}$ we will denote by $S^X$ and $S^Y$  
the projection of $S$ to the rows and columns, respectively, 
that is, $S^X \coloneqq \set{x\in[m]^n \mid (\set{x} \times \{0,1\}^{mn}) \intersection S \neq 
\emptyset}$,
and 
$S^Y \coloneqq \set{y\in\{0,1\}^{mn} \mid ([m]^n \times \set{y}) \intersection S \neq \emptyset}$. 
In particular,
for any rectangle $R \subseteq [m]^n \times \{0,1\}^{mn}$, we have that $R= R^X \times R^Y$. 
Given a row $x\in [m]^n$, we denote the set of columns in $S$ along the row $x$ by 
\begin{equation}
	\RowInSet{x}{S}~:=~\{y \in \{0,1\}^{mn} \mid \{x\} \times \{y\} \in S \}
\end{equation}
and note that $\{x\} \times \RowInSet{x}{S} = (\{x\} \times \{0,1\}^{mn}) \cap S$.

For each triangle $T \subseteq [m]^n \times \{0,1\}^{mn}$ we fix some arbitrary 
choice of $a_T$ and $b_T$ for which $T= \{(x,y) \in [m]^n \times \{0,1\}^{mn}\mid a_T(x) 
<b_T(y)\}$. 
For convenience, we arrange all rows in ascending order of $a_T$ from top to bottom and all 
columns in descending order of $b_T$ from left to right.
We call this the \emph{ordering for the triangle~$T$}.  

Throughout the proof we consider block-wise partial assignments
$\alpha : [n] \rightarrow [m] \cup \{*\}$, which we refer to as \emph{pointers}.
Let $\fix(\alpha)~:=~\{i \in [n] \mid \alpha^{-1}(i) \neq *\}$.
We say $x\in [n]^m$ \emph{is consistent with} $\alpha$ if $x_i = \alpha(i)$ for all $i\in \fix(\alpha)$.
 For $X \subseteq [n]^m$, we define 
 \begin{equation}
	X_\alpha := \{x\in X\mid \text{$x$ is consistent with $\alpha$}\}.
 \end{equation}
The proof of the lifting theorem relies on sets that have
high min-entropy when restricted to any subset of its coordinates. In the following definition, $\delta > 0$ is a parameter to be specified later.

\begin{definition}\label{def:predense}
    Say that a set $X\subseteq [m]^n$ is \emph{$\alpha$-dense} if $X\neq \emptyset$ and for every 
    nonempty $I\subseteq [n] \setminus \fix(\alpha)$, $H_\infty((X_\alpha)_I) \geq \delta |I| \log m$. 
    We say $X\subseteq[m]^n$ is \emph{$\alpha$-predense} if $X$ contains an $\alpha$-dense 
    subset. 
\end{definition}

Observe that being $\alpha$-dense for a fixed $\alpha$ is generally not a monotone property. For 
instance, the set $[m^\delta]^n$ is $\alpha$-dense, where $\alpha$ is the 
constant-$\star$ 
function, 
while its union with $\{1\}\times[m]^{n-1}$ is 
not, provided that $1/\log m<\delta<1-1/n$. 
By contrast, being predense is by definition a monotone property. 

The basic objects in our analysis are rectangles of the following type. 

\begin{definition}\label{def:prestructured}
	A rectangle $R:=X \times Y$ is \emph{$\alpha$-pre-structured} if 
	\begin{enumerate}[label=\roman*)]
		\item \label{item:predense,1} 
            $X$ is $\alpha$-predense; and
		\item\label{item:predense,2} 
            $Y$ satisfies
            $H_\infty(\Y)\geq mn-m^{\delta}/(16n)$. 
	\end{enumerate}
We say that $R$ is \emph{$\alpha$-structured} if condition \ref{item:predense,1} is replaced with $R$ being $\alpha$-dense. 
\end{definition}

Note that condition \ref{item:predense,2} is equivalent to saying that the density of $Y$ in 
$\{0,1\}^{mn}$ is at least $2^{-m^{\delta/(16n)}}$. 

Similar to previous DAG lifting 
theorems~\cite{GLMWZ16Rectangles,GGKS20MonotoneCircuit,lmmpz2022liftingWithSunflowers}, 
ours relies on the following key lemma which states that the image of a structured rectangle is the 
full Boolean cube. 
Here, and throughout this section, we omit the subscript $m$ in $\IND_m$,
and for $X \subseteq [m]^n$ and $Y \subseteq \set{0,1}^{mn}$,
we denote by $\IND(X,Y) \subseteq \set{0,1}^n$ the image of $X \times Y$ under the map $\IND$.

A self-contained proof of the lemma below is given in \autoref{sec-Appendix-monotone}.

\begin{fullImageLemma}\label{lem:fullImage}

For any positive integers $m$ and $n$, any rectangle 
$X\times Y\subseteq [m]^n\times \{0,1\}^{mn}$, 
and any parameter $\delta\in(0,1)$ such that $m^\delta \geq (2\ln 2)n$, it holds that if
\begin{enumerate}[label=(\arabic*)]
    \item \label{item:condX} $ H_{\infty}(\textbf{X}_I)\geq \delta|I| \log m-1$ for any nonempty 
    $I\subseteq[n]$, and
    \item \label{item:condY} $H_{\infty}(\Y) \geq {mn} - {m^{\delta}/(8n)}$,
\end{enumerate}
then there exists $x^*\in X$ such that $\IND(\{x^*\},Y)=\{0,1\}^n$.
\end{fullImageLemma}

For a partial assignment $\rho \in \{0,1,\star\}^n$, where $I = \fix(\rho)$, 
define the width-$|I|$ clause $C_\rho$ %
$:=\bigvee_{i\in I} z_i^{1-\rho(i)}$ where $z^1:=z$ and $z^0:=\stdnot{z}$. 
In other words, the falsifying assignments of $C_\rho$ in $\{0,1\}^n$,
denoted by $C_\rho^{-1}(0)$, comprise precisely the sub-cube specified by $\rho$.  
We can interpret the \fullImg as saying that
$\IND(\{x^*\} \times Y) = C_\rho^{-1}(0)$,
for $\rho$
being the constant-$\star$ assignment.
The following corollary 
is a simple application of the lemma to (almost) pre-structured rectangles.

\begin{corollary}
	\label{cor:prestrucImage}
        Let $m\geq (16n)^{{1}/{\delta}}$, let $\alpha$ be a pointer and let 
       $R=R^X \times R^Y$ be a 
        rectangle that satisfies the following two properties (which are a slightly weaker version of 
        \autoref{def:prestructured}):
\begin{enumerate}
\item\label{item:almostpre,1}
$(R^X)_\alpha$ contains a subset $D$ where $H_{\infty}(\pmb{D}_I) \geq \delta|I|\log m-1$ for all nonempty 
$I\subseteq [n]\backslash \fix(\alpha)$; and
\item\label{item:almostpre,2}
$H_{\infty}(R^Y) \geq mn - m^{\delta}/(16n) - 1$.
\end{enumerate}
      If for some $\rho\in\{0,1\}^{\fix(\alpha)}$ it holds that $\IND(x_{\fix(\alpha)}, y_{\fix(\alpha)}) = \rho$ for all 
      $(x,y) \in R$, then  there exists an $x^*\in R^X$ such that $\IND(\{x^*\},Y)=C^{-1}_{\rho}(0)$.
\end{corollary}

	\begin{proof}
        If $\fix(\alpha)=[n]$, the conclusion follows from the fact that $R\neq\emptyset$ and for any $(x,y)\in R$,  $\IND(x,y)=C^{-1}_\rho(0)$, which is a singleton set. 
        Hence, we assume $\setsize{\fix(\alpha)}<n$.
       
       	Let $D\subseteq (R^X)_\alpha$ be a subset witnessing property \ref{item:almostpre,1}. 
		Let $J:=[n] \setminus \fix(\alpha)$. 
		We consider the rectangle $R':= D_J \times R^Y_J \subseteq [m]^{|J|}\times\{0,1\}^{m|J|}$, 
		where $D_J$ 
		and $R^Y_J$ are the projections of $D$ and of~$R^Y$ to the coordinates $J$. 
		We want to show that the \fullImg is applicable to rectangle $R'$. 
		For this purpose, 
        observe that for any $I\subseteq J$, it follows from the assumption on $D$ that $D_J$ satisfies condition 
        \ref{item:condX} of the lemma. 
		Secondly, note that
		\begin{equation}
        \frac{|R^Y_J|}{2^{m|J|}}\geq \frac{|R^Y|}{2^{mn}} \geq 2^{-m^{\delta}/(16n) -1} \geq 
       	 2^{-m^{\delta}/(8|J|)},
		\end{equation}
       	 where the first inequality follows from the definition of projection (i.e., projection does not 
       	 decrease density), the second 
       	 inequality follows from \autoref{item:almostpre,2}, 
       	 and for the third inequality we use that $m^{\delta} \geq 16n$.
		 Together, this verifies the conditions in the {\fullImg}, whose application to $R'$ gives us an 
		 element $x^{**}\in D_J$ such that $\IND(\{x^{**}\} \times R^Y_J) = \{0,1\}^{|J|}$. 
		 The concatenation of $\alpha$ and $x^{**}$ provides the desired element $x^*\in D \subseteq R^X$. 
	\end{proof}

In order to extract a resolution proof from a triangle-DAG, we cover a 
triangle by a set 
of pre-structured rectangles (from which \autoref{cor:prestrucImage} allows us to extract low-width clauses), 
along with a small number of ``error'' rows and columns. Unlike previous approaches which partition the 
triangle into rectangles, we will cover the triangle with (potentially overlapping) \emph{strips}---sets of 
pre-structured rectangles which all share the same rows, along with a set of ``secured'' rows on 
which 
can apply %
\autoref{cor:prestrucImage}. This overlapping covering, as opposed to partitioning, allows us to reduce 
the number of rows and columns which are not within any pre-structured rectangle and hence reduce the 
error sets. 

We now formally define this notion of a strip. For this, let $w \leq n$ be a parameter which  
corresponds to the width of the resolution proof to be extracted. For a triangle $T$, recall that  $T^X 
\subseteq [m]^n$ and $T^Y \subseteq \{0,1\}^{mn}$ are the row and column projections of $T$. 

\begin{definition}[Strips]
\label{def:strip}
	For a triangle $T \subseteq T^X \times T^Y$ a \emph{strip} $S$ of $T$ is a subset of rows $S 
	\subseteq T^X$ that is $\alpha$-predense for some pointer $\alpha \in 
	\left([m]\union\set{\star}\right)^n$ with $|\fix(\alpha)| \leq w$. 
	Associated with $S$ are the following: 
	\begin{enumerate}[label=\roman*)]
		\item A collection of $\alpha$-pre-structured rectangles ${\cal R}^S= 
		\{R_\beta\}_{\beta}$ indexed by a set of
		$\beta \in \set{0,1,\star}^{n}$ with 
		$\fix(\beta) = \fix(\alpha)$, where each $R_\beta = S \times 
		Y_\beta$ is such that $\IND(\alpha_{\fix(\alpha)},y_{\fix(\alpha)}) = \beta_{\fix(\alpha)}$ for all $y 
		\in Y_\beta$. 
		Furthermore, within each 
		$R_\beta$ there is an ``inner'' sub-rectangle $R^\inn_\beta \subseteq R_\beta \cap T$ which is 
		$\alpha$-structured and fully 
		contained within~$T$. \label{item:strip1}
		\item A subset of rows $\hatS \subseteq S$ which we call the rows \emph{secured} by $S$.  
		\label{item:strip2}
	\end{enumerate}
\end{definition}

A depiction of a strip is given in \autoref{fig:strips}.
The purpose of the secured rows is described by the following lemma, which states that for any triangle $T$, 
we can construct a set of strips such that the associated
pre-structured rectangles cover all of~$T$ except a small set of error rows---rows that are not 
secured by any strip constructed---and error columns. 
We note that the definition of strips depends on parameters $n,m,w$ and, due to the definition of 
$\alpha$-pre-structured and $\alpha$-structured, also on $\delta$.

\begin{triangleLemma}
\label{lem:triangle}
	For any positive integers $m,n$ and $w\leq n$, and parameter $\delta \in (0,1)$ and any
 triangle $T \subseteq T^X \times T^Y \subseteq [m]^n \times \{0,1\}^{mn}$ there is a set of strips 
 $\mathsf{Strips}(T)$ of $T$  and ``error'' sets $X^T_{\err} \subseteq [m]^n, Y^T_{\err} \subseteq 
 \{0,1\}^{mn}$ such that for any $x \in T^X$ one of the following cases holds:
	\begin{itemize}
		\item \emph{ Security.}~ If %
		$x$ is secured by a strip $S \in \mathsf{Strips}(T)$, then %
		$\{x\} \times \RowInSet{x}{T}$ is covered by the rectangles in~$\mathcal{R}^S$ together with 
		the error 
		columns, that is, 
        \begin{equation}
			\{x\} \times \RowInSet{x}{T}\subseteq \bigcup_{R \in \mathcal{R}^S} R \cup \left(\{x\}\times 
		Y^T_{\err}\right).
		\end{equation}
		\item \emph{ Error.}~ If %
		$x$ is not secured by any strip in $\mathsf{Strips}(T)$, then $x \in X^T_{\err}$. 
		\item \emph{ Maximality.}~ If there exists a rectangle $R\subseteq T^X \times (T^Y \setminus Y^T_{\err})$ that is $\alpha$-pre-structured for some pointer $\alpha$ with $|\fix(\alpha)| \leq w$ 
		and $\IND(R)\subseteq C_{\beta}^{-1}(0)$ for some $\beta\in\{0,1,\star\}^n$ with 
		$\fix(\beta)=\fix(\alpha)$, 
		then there exists a strip $S\in \mathsf{Strips}(T)$ with 
        associated pointer $\alpha$ such that $\mathcal{R}^S$%
		contains a rectangle indexed by $\beta$.
	\end{itemize}	
	Furthermore, $|X^T_{\err}| \leq m^{n-(1-\delta)w}$ and $|Y^T_{\err}| \leq 2^{mn-m^{\delta}/(16n) + (w+1)\log(2mn)}$.
\end{triangleLemma}

We defer the proof of the lemma together with the construction of strips to \autoref{sec:triLem} in 
favor of first completing the proof of the lifting theorem.

\subsection{Proof of Lifting Theorem %
}

Now we prove \autoref{thm:dagLifting} using the \triLemma and \autoref{cor:prestrucImage}. 
		
	\begin{proof}[Proof of \autoref{thm:dagLifting}]
		Let $\Pi$ be any triangle DAG of size $m^{(1-\delta)w}/2$ solving $S \circ \IND^n_m$.
We can assume $w\leq n$, otherwise the theorem trivially holds.
        We first remove the error rows and columns from $\Pi$ as follows.
 
\paragraph{Error Removal.} Sort the triangles of $\Pi$ in any topological order $T_1,\ldots, T_s$ 
from the leaves to the root. That is, if $T$ is a child of $T'$ then $T$ comes before $T'$ in the 
order. We 
process $\Pi$ by the following procedure.
	
\spcnoindent Initialize $X_{\err}^0 = Y_{\err}^0~:=~ \emptyset$. For $i=1,\ldots, s$ do the following 
in order:
		\begin{enumerate}
			\item Remove from $T_i$ the error rows and columns accumulated at $i-1$, that is,
			\begin{equation}
				T_i ~\leftarrow~ T_i \setminus \big((X_{\err}^{i-1} \times \{0,1\}^{mn}) \cup ([m]^n \times Y_{\err}^{i-1} )\big).
			\end{equation} 

                \item Let $X_{\err}^{T_i}$ and $Y_{\err}^{T_i}$ be the $X$- and $Y$-error sets, respectively, 
                obtained by applying the \triLemma to $T_i$.

                \item Define $X_{\err}^i:= X_{\err}^{i-1} \cup X_{\err}^{T_i}$ and $Y_{\err}^i:= Y_{\err}^{i-1} \cup Y_{\err}^{T_i}$.

		\end{enumerate}

  \noindent Note that in this procedure, the children nodes will each contribute some error rows/columns to the parents, and every node remains a triangle, as we have only removed whole rows/whole columns from it. 
        Henceforth, $\Pi$ will refer to the resulting triangle-DAG after this procedure. 
    \smallskip
    
        We extract from $\Pi$ a resolution refutation of $F$ by showing that the following two items 
        hold.
		\begin{itemize}
			\item \emph{Clauses.}~ We can associate with every triangle $T$ in $\Pi$ a set ${\cal C}(T)$ 
			of clauses---each of width at most~$w$---such if $T$ is a leaf of $\Pi$ then ${\cal C}(T)$ is 
			a weakening of an initial clause of  $F$, and if $T$ is the root then the empty clause $\bot$ is 
			contained in ${\cal C}(T)$.
			\item \emph{Inferences.}~ If triangle $T$ has children $T_1$ and $T_2$ in $\Pi$ then each 
			clause in ${\cal C}(T)$ has a width-$w$ and depth-$w$ derivation from the clauses ${\cal 
			C}(T_1) \cup {\cal C}(T_2)$.
		\end{itemize}
        \noindent We now prove these items.

\paragraph{Clauses.} For each triangle $T$ in $\Pi$, apply the \triLemma to obtain a set of strips 
$\mathsf{Strips}(T)$ of $T$. We define ${\cal C}(T)$ as follows: for each strip $S \in 
\mathsf{Strips}(T)$ and each pre-structured rectangle $R_\beta \in \mathcal{R}^S$, we include the 
clause $C_\beta$; that is,
\begin{equation}
	{\cal C}(T):=\bigcup_{S \in \mathsf{Strips}(T)}
            \left\{ C_\beta\mid R_\beta\in\mathcal{R}^S\right\}.
\end{equation}
To see that $C_\beta$ is a clause of width at most $w$, let $\alpha$ with $|\fix(\alpha)|\leq w$ be 
the pointer associated with the strip $S$. Then \autoref{cor:prestrucImage} guarantees that 
$\IND(R_\beta) = 
C^{-1}_\beta(0)$ where the width of $C_\beta$ is $|\beta^{-1}(0) \cup 
\beta^{-1}(1)|=|\fix(\alpha)| \leq w$.

We now verify that these sets of clauses satisfy the desired root and leaf properties.

		\begin{itemize}
			\item \emph{Root.}~	Let $R=R^X \times R^Y$ be the triangle at the root of $\Pi$ (which is a 
			rectangle, though we won’t need this).  By the  \triLemma and a union 
			bound over the triangles in $\Pi$, the density $X$-error 
    accumulated at the root is at most 
	\begin{equation}
		m^{-(1-\delta)w} \cdot |\Pi'| \leq m^{-(1-\delta)w} \cdot m^{(1-\delta)w}/2 = 1/2.
	\end{equation} 
		Hence $R^X$ has density at least $1/2$. This implies that for any $\emptyset \neq I \subseteq 
		[n]$, 
		\begin{equation}
			H_\infty \left(R^X_I \right) \geq |I|\log m-1 \geq \delta |I| \log m,
		\end{equation}
  and so we have that  $R^X$ is $\star^n$-predense.

  Similarly, the density of the $Y$-errors accumulated at the root is at most 
  \begin{equation}
	2^{-m^{\delta}/(16n) + (w+1)\log(2mn)}\cdot|\Pi| < 
  2^{-m^{\delta}/(8n)} < 1/2,
  \end{equation}
  where in the final inequality we use that $m\geq (50n/\delta)^{\frac{2}{\delta}}$. 
Therefore, $H_\infty(R^Y \setminus Y^R_{\err})\geq mn - 1 \geq mn-m^{\delta}/(16n)$. 
		We therefore conclude that $ R^X \times (R^Y \setminus 
		Y^R_{\err})$ is a $\star^n$-(pre-)structured rectangle.

By the maximality condition of \triLemma applied to $ R^X \times (R^Y \setminus 
		Y^R_{\err})$, we have that
there exists a strip $S\in \mathsf{Strips}(R)$ with associated pointer $\star^n$
and such that the collection $\mathcal{R}^S$ is non-empty. 
By the item~\ref{item:strip1} of \autoref{def:strip}
$\mathcal{R}^S = \set{R_\beta}_\beta$ has
to be a singleton set since only the empty string can be a subscript $\beta$.
Let $R_\beta$ be the unique rectangle in $\mathcal{R}^S$. 
Since $R_\beta$ is $\star^n$-pre-structured, by \autoref{cor:prestrucImage},
it holds that
$\IND(R_\beta) = \{0,1\}^n = \bot^{-1}(0)$.
Therefore $\bot \in \mathcal{C}(R)$.

		 	\item \emph{Leaves.}~ Consider any leaf triangle $T$ of $\Pi$. By definition, there is an axiom 
		 	clause $C_i\in F$ such that $\IND(T) \subseteq C_i^{-1}(0)$. Therefore, for any clause 
		 	$C_\beta \in {\cal C}(T)$, we have $C_{\beta}^{-1}(0) = \IND(R_\beta) = \IND(R_\beta^\inn) 
		 	\subseteq \IND(T) 
		 	\subseteq C_i^{-1}(0)$, meaning 
		 	$C_\beta$ is a weakening of $C_i$. 
		\end{itemize}
				
\paragraph{Inferences.} Let $T$ be any non-leaf triangle in $\Pi$ with children $T_1$ and $T_2$. 
Consider any clause $C \in {\cal C}(T)$ generated by some pre-structured rectangle 
$R_\beta = S 
\times Y^T_\beta$ in a strip $S$ defined from some pointer $\alpha$. We will show that $C$ has a 
resolution derivation of width $w$ and depth $w$ from either ${\cal C}(T_1)$ or ${\cal C}(T_2)$.
		
		Consider the ``inner'' structured sub-rectangle $R^{\inn}_\beta \subseteq R_\beta \cap T$. 
		Since $T$ 
		is covered by its children $T_1$ and $T_2$, 
		\begin{equation}
			R^{\inn}_\beta \subseteq T \subseteq T_1 \cup T_2.
		\end{equation}
		We claim that at least one of $T_1$ or $T_2$ contains a sub-rectangle $Q = Q^X \times 
		Q^Y\subseteq R^{\inn}_\beta$ with $X$- and $Y$-density at least half that of 
		$R^{\inn}_\beta$. 
		To see this, order the rows/columns according to the ordering of $T_1$, then the center $p$ of 
		$R^{\inn}_\beta$ divides $R^{\inn}_\beta$ into four quadrants. If $p\in T_1$ then, as $T_1$ is a 
		triangle, the 
		top-left quadrant $Q$ of $R^{\inn}_\beta$ is contained entirely within $T_1$; see 
		\autoref{fig:triangleCover}. Otherwise, if $p\notin T_1$, then as $T_1$ is a triangle, the 
		bottom-right quadrant $Q$ is disjoint from $T_1$ and so it must be contained within $T_2$. In 
		either case, $H_\infty(Q^Y)\geq H_\infty(Y^T_\beta)-1$ and $H_\infty(Q^X_J) \geq 
		H_\infty(S_J)-1$ for any $\emptyset\neq J \subseteq [n] \setminus I$. In particular, $Q$ 
		satisfies the premises of \autoref{cor:prestrucImage}. 
Suppose without loss of generality that $Q\subseteq T_1$. 
		
	\begin{figure}[h]
		\centering
		\begin{tikzpicture}[scale=1.2]
		\draw[very thick] (-2.4,2)  -- (2.4,2) -- (2.4,-2) -- (-2.4,-2) -- cycle;

		\draw[thick, color = black, fill =green!9] (-1.6,-2) -- (-1.6,-1.6) -- (-1.2,-1.6) -- (-1.2,-0.8) -- (-0.4,-0.8) -- (-0.4,-0.4) -- (0.4,-0.4) -- (0.4,0.8) -- (0.8,0.8)--(0.8,1.2) -- (1.6,1.2) -- (1.6, 1.6) -- (2,1.6) -- (2,2) -- (-2.4,2) -- (-2.4,-2) -- cycle;
		\draw[thick, color=black, fill=betterYellow!16] (1,-1)  -- (1,1) -- (-1.4,1) -- (-1.4,-1) -- cycle;
		\draw[color=black!60, fill=red!12] (-0.2,0)  -- (-0.2,1) -- (-1.4,1) -- (-1.4,0) -- cycle;
		\draw[color = black!60] (-0.2,-1) -- (-0.2,1);
		\draw[color = black!60] (-1,0) -- (1,0);
		\draw[thick, color=black] (1,-1)  -- (1,1) -- (-1.4,1) -- (-1.4,-1) -- cycle;
		\draw[thick, color = black] (-1.6,-2) -- (-1.6,-1.6) -- (-1.2,-1.6) -- (-1.2,-0.8) -- (-0.4,-0.8) -- (-0.4,-0.4) -- (0.4,-0.4) -- (0.4,0.8) -- (0.8,0.8)--(0.8,1.2) -- (1.6,1.2) -- (1.6, 1.6) -- (2,1.6) -- (2,2) -- (-2.4,2) -- (-2.4,-2) -- cycle;
		
		\node[text width=2cm] at (-0.1 ,0.5) {$Q$};
		
		\node[text width=2cm] at (1.9 ,0.4) {$R^{\inn}_\beta$};
		
		\node[text width=2cm] at (-1.4 ,1.7) {$T_1$};
		
		 \end{tikzpicture}
		 \caption{The structured rectangle $R^\inn_\beta$ for triangle $T$, whose quadrant $Q$ is 
		 contained entirely within child $T_1$.} \label{fig:triangleCover}	
	\end{figure}
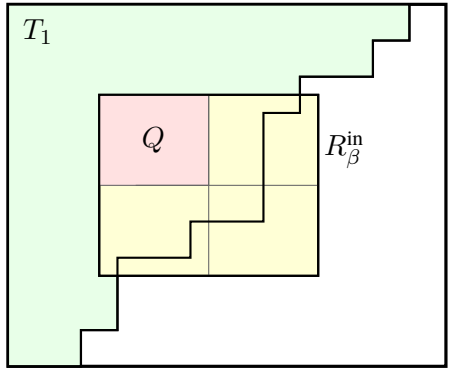
		
		Applying \autoref{cor:prestrucImage} to $Q$, we get a row $x^* \in Q^X\subseteq T_1^X$ such 
		that $\IND(\{x^*\} \times Q^Y) = C_\beta^{-1}(0)$. As we have removed $X_{\err}^{T_1}$ and 
		$Y_{\err}^{T_1}$ from $T$ in the Error Removal step, $x^*\notin X^{T_1}_{\err}$ and 
		$Q^Y\subseteq T^Y$ is disjoint from $Y^{T_1}_{\err}$. Thus, $x^*$ is secured by a strip $S'$ 
		of $T_1$ defined by some pointer $\alpha'$. By the \triLemma,
        \begin{equation}
			\{x^*\} \times Q_Y \subseteq \{x^*\} \times (\RowInSet{x^*}{T_1} \setminus Y^{T_1}_{\err}) 
        \subseteq %
        \bigcup_{R_\xi \in {\cal R}^{S'}} R_\xi,
		\end{equation}
         where ${\cal R}^{S'}$ is the set of pre-structured rectangles in strip $S'$. %
		By \autoref{cor:prestrucImage}, each $\alpha'$-pre-structured rectangle $R_\xi\in\mathcal{R}^{S'}$ satisfies 
		$\IND(R_\xi) = C^{-1}_\xi(0)$, %
        and so
		\begin{equation}
			C^{-1}(0) = \IND(R_\beta) = \IND(\{x^*\} \times Q_Y) \subseteq \bigcup_{R_\xi \in {\cal R}^{S'}} 
		\IND(R_\xi) = 
		\bigcup_{R_\xi \in {\cal R}^{S'}} C^{-1}_\xi(0).
	    \end{equation}
		In particular, this means that $C$ is logically implied by the clauses $\{C_\xi\}$, which depend only on variables in $\fix(\alpha')\subseteq [n]$. Letting $C'$ be the clause obtained from $C$ by discarding the literals over $\fix(\alpha) \setminus \fix(\alpha')$, then $C'$ is also logically implied by the clauses $\{C_\xi\}$. Hence we can derive $C$ in resolution by first deriving $C'$ and then weakening it to $C$. 
		As the variables of $C'\cup\{C_\xi\}$ are in $\fix(\alpha')$ which has size $\leq w$, and the width of $C$ is at most $w$, this derivation takes width and depth at most $w$.  Here we have used the fact that any derivation over $k$ variables takes width and depth at most $k$. 
	\end{proof}

\subsection{Proof of \triLemma}
\label{sec:triLem}

The rest of this section is dedicated to the proof of the \triLemma. 
That is, our goal is to describe, 
for any given parameter $w$, 
how to associate with any triangle $T \subseteq T^X \times T^Y$ a set of strips 
$\mathsf{Strips}(T)$ and error sets which satisfy the \emph{security}, \emph{error} and 
\emph{maximality} properties of the lemma. 

Let parameters $w$ and $\delta$ be given. 
For every pointer $\alpha$ with $\setsize{\fix(\alpha)}\leq w$ such that the rows 
in $T^X$ that are consistent with $\alpha$ form an $\alpha$-predense set, i.e., $(T^X)_\alpha$ 
is 
$\alpha$-predense, we construct a strip $S \coloneqq (T^X)_\alpha$, to be included in 
$\mathsf{Strips}(T)$, by
associating $S$ with the following structures.
\begin{itemize}
	\item {\it Secured Rows.}~ Let $x^S \in S$ be the highest row (according to the ordering of $T$) such that the 
	elements in $S$ above or equal to $x^S$ form an $\alpha$-predense set. Let the secured 
	rows $\hatS 
	\subseteq S$ be those below or equal to $x^S$.
	\item {\it Pre-Structured Rectangles.}~ Generate the set of pre-structured rectangles ${\cal R}^S$ as follows: for 
	every $\beta \in \{0,1,\star\}^{n}$ with $\fix(\beta) = \fix(\alpha)$ consider the set of columns 
	\begin{equation}\label{eq:Ybeta}
		Y_\beta~:=~\{y \in \{0,1\}^{mn} \mid
	\IND 
	\left(\alpha_{\fix(\alpha)}, 
	y_{\fix(\alpha)} \right) = \beta_{\fix(\alpha)}\} \eqperiod 
	\end{equation} 
If $H_\infty(Y_\beta) \geq mn-m^{\delta}/(16n)$ 
 then we include the rectangle $R_\beta \coloneqq S \times Y_\beta$ in ${\cal R}^S$. Otherwise, we 
 include the columns 
 $Y_\beta$ in a set $Y_\err^S$. 

 \item \emph{Inner Rectangle.} It remains to show that we can find some sub-rectangle 
 $R^\inn_\beta \subseteq R_\beta 
 \cap T$ which is $\alpha$-structured and contained entirely within $T$. Since $S$ is 
 $\alpha$-predense 
 there is some $\alpha$-dense subset of 
 rows $S' \subseteq S$. Note that by definition $S'$ is only above (and including) $x^S$, and so the 
 rectangle $R^\inn_\beta \coloneqq S' 
 \times Y_\beta$ is only above (and including) $\{x^S \} \times Y_\beta \subseteq T$. Hence, as $T$ 
 is a 
 triangle, 
 $R^\inn_\beta 
 \subseteq T$. Finally, note 
 that as $R_\beta$ was not categorized as ``error'', $R^\inn_\beta$ is $\alpha$-structured.

\end{itemize}
Observe that with this construction each strip in $\mathsf{Strips}(T)$ is uniquely determined by a pointer~$\alpha$. Finally, define the associated error sets $X_{\err}^T \subseteq [m]^n$ 
and $Y_{\err}^T \subseteq \{0,1\}^{mn}$ as follows:
	\begin{itemize}
		\item \emph{$X$-Error.}~ Let $X_{\err}^T$ be the set of rows in $T^X$ which are not secured 
		by \emph{any} strip in $\mathsf{Strips}(T)$.
		\item \emph{$Y$-Error.}~ Let $Y_{\err}^T$ be collected over all strips $S \in 
		\mathsf{Strips}(T)$, that is, $Y_{\err}^T \coloneqq \bigcup_{S \in 
		\mathsf{Strips}(T)}Y_{\err}^S\,$.
	\end{itemize}
A depiction of a strip is in \autoref{fig:strips}.

	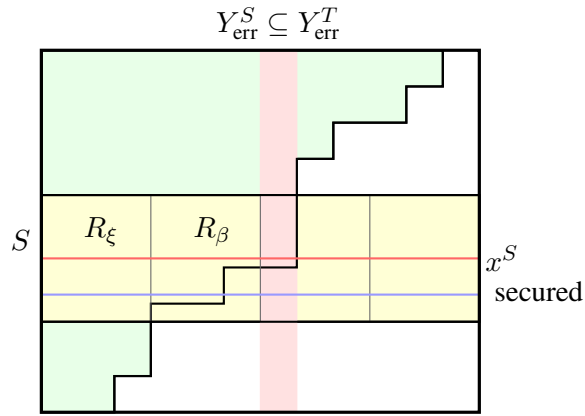
\begin{figure}[h]
		\hspace{120pt}
		\begin{tikzpicture}[scale=1.2]

		\draw[thick, color = black, fill =green!9] (-1.6,-2) -- (-1.6,-1.6) -- (-1.2,-1.6) -- (-1.2,-0.8) -- (-0.4,-0.8) -- (-0.4,-0.4) -- (0.4,-0.4) -- (0.4,0.8) -- (0.8,0.8)--(0.8,1.2) -- (1.6,1.2) -- (1.6, 1.6) -- (2,1.6) -- (2,2) -- (-2.4,2) -- (-2.4,-2) -- cycle;
		\draw[thick, color=black, fill=betterYellow!16] (2.4,-1)  -- (2.4,0.4) -- (-2.4,0.4) -- (-2.4,-1) -- cycle;
		
		\draw[color=red!12, fill=red!12] (0.4,-2)  -- (0.4,2) -- (0,2) -- (0,-2) -- cycle;
		
		\draw[thick, color = black] (-1.6,-2) -- (-1.6,-1.6) -- (-1.2,-1.6) -- (-1.2,-0.8) -- (-0.4,-0.8) -- (-0.4,-0.4) -- (0.4,-0.4) -- (0.4,0.8) -- (0.8,0.8)--(0.8,1.2) -- (1.6,1.2) -- (1.6, 1.6) -- (2,1.6) -- (2,2) -- (-2.4,2) -- (-2.4,-2) -- cycle;
		
		\draw[color = black!60] (1.2,-1) -- (1.2,0.4);
		
		\draw[color = black!60] (-0,-1) -- (-0,0.4);
		
		\draw[color = black!60] (-1.2,-1) -- (-1.2,0.4);
		
		\draw[thick, color=red!60] (-2.4,-0.3) -- (2.4,-0.3);
		
		\draw[thick, color=blue!40] (-2.4,-0.7) -- (2.4,-0.7);
		
		\draw[very thick] (-2.4,2)  -- (2.4,2) -- (2.4,-2) -- (-2.4,-2) -- cycle;
		
		\draw[thick, color=black] (2.4,-1)  -- (2.4,0.4) -- (-2.4,0.4) -- (-2.4,-1) -- cycle;
		
		\node[text width=2cm] at (-1.9 ,-0.1) {$S$};
		
		\node[text width=2cm] at (3.3 ,-0.3) {$x^S$};
		
		\node[text width=2cm] at (0.1 ,0) {$R_\beta$};

  \node[text width=2cm] at (-1.1 ,0) {$R_\xi$};
		
		\node[text width=2cm] at (0.35 ,2.3) {$Y_{\err}^S \subseteq Y_\err^T$};
		
		\node[text width=2cm] at (3.4 ,-0.66) {secured};
		
		 \end{tikzpicture}
		 \caption{A strip $S$ of a triangle, including two pre-structured rectangles $R_\xi,R_\beta$, a 
		 set of error columns $Y^S_\err$, and an example of a secured row.} \label{fig:strips}	
	\end{figure}
 
We first argue that the {error}, {security} and maximality properties of the \triLemma hold. 
The {error} property 
 holds by construction. 
To see why the security property holds, note that
given 
a strip $S \in 
\mathsf{Strips}(T)$, the row $x^S$ in $T$ is covered by the $\alpha$-pre-structured rectangles 
and the 
error columns in~$S$. That is,
\begin{equation}
	\{x^S\} \times \RowInSet{x^S}{T} \subseteq \bigcup_{R \in {\cal R}^S} R ~\cup 
 \left(\{x^S\} \times Y_{\err}^T \right) \eqperiod
\end{equation}
Fix any secured row $x \in \hatS \subseteq S$. Then $x$ is below or equal to $x^S$ and, therefore, 
since $T$ is a triangle, 
$\RowInSet{x}{T} \subseteq \RowInSet{x^S}{T}$. Hence, 
\begin{equation}
	\{x^S\} \times \RowInSet{x}{T} \subseteq \bigcup_{R \in 
{\cal R}^S} R ~\cup 
 \left(\{x^S\} \times Y_{\err}^T \right)\eqperiod
\end{equation}
Now, for the maximality property, %
assume there is a rectangle $R\subseteq T^X \times (T^Y \setminus Y^T_{\err})$ that is 
$\alpha$-pre-structured for some pointer $\alpha$ with $|\fix(\alpha)| \leq w$ and 
$\IND(R)\subseteq C_{\beta}^{-1}(0)$ for some $\beta\in\{0,1,\star\}^n$ with 
$\fix(\beta)=\fix(\alpha)$. 
Note that $(R^X)_\alpha$ is $\alpha$-predense and hence so is its superset $(T^X)_{\alpha}$, 
thus by our construction there is a strip $S\in \mathsf{Strips}(T)$ associated with $\alpha$. 
Since $\IND(R)\subseteq C_{\beta}^{-1}(0)$ it follows that $R^Y$ is a subset of 
$Y_{\beta}=\{y\in\set{0,1}^{mn}\mid \IND(\alpha_{\fix(\alpha)}, y_{\fix(\alpha)}) = 
\beta_{\fix(\alpha)}\}$ defined in \eqref{eq:Ybeta}.  
As $R^Y \intersection Y^T_{\err} = \emptyset$ and  $R^Y\subseteq Y_{\beta}$  (and 
$R^Y\neq\emptyset$ 
 since $R$ is $\alpha$-pre-structured), 
it must be the case that $ Y_{\beta} \not\subseteq Y^T_{\err}$ and thus,
by the construction of \emph{Pre-Structured Rectangles}, it must be the case that
$H_{\infty}(Y_{\beta}) \geq mn-m^{\delta/(16n)}$.
 Therefore, the rectangle $R_{\beta}:=(T^X)_{\alpha}\times Y_\beta$ is  
 $\alpha$-pre-structured 
 and thus, by our construction, is in $\mathcal{R}^S$.

 Finally, we bound the size of the error sets, using the following two claims. 

\newcommand{\alphaprime}{\alpha'}

 \begin{claim}
		For any triangle $T$, the density of $X^T_{\err}$ in $[m]^n$ is less than $m^{-(1-\delta)w}$.
\end{claim}
	\begin{proof}
		Suppose for contradiction that $X^T_{\err}$ has density at least $m^{-(1-\delta)w}$. For simplicity, we denote 
		$\hatx \coloneq X^T_{\err}$. Let $I \subseteq [n]$ be a maximal set of blocks where 
		$\hatx$ 
		is not dense---meaning that $H_\infty({\bf \hatx}_I) < \delta |I|\log m$---and fix any pointer $\alpha$ with 
		$\fix(\alpha)=I$ that witnesses $\Pr[{\bf \hatx}_I = \alpha_I] \geq m^{-\delta|I|}$. If no such $I$ 
		exists, we 
		let $I \coloneqq \emptyset$ and $\alpha=\star^n$. 
        We record the following two basic properties:
		\begin{enumerate}[label=(\arabic*)]
			\item\label{item:|I|} $|I| \leq w$, 
			\item\label{item:Xerror-dense} $\hatx_{\alpha} \coloneqq \{x \in X_{\err}^T \mid 
			x_I=\alpha_I\}$ is $\alpha$-dense.
		\end{enumerate}
		To see \autoref{item:|I|}, observe that by the definition of $\alpha$, 
		\begin{equation}
			|\hatx| \leq 
		\frac{|\hatx_{\alpha}|}{m^{-\delta|I|}} 
		\leq \frac{|\{x\in[m]^n\mid x_I=\alpha_I\}|}{m^{-\delta|I|}} = m^{n-(1-\delta)|I|} \eqperiod
		\end{equation}
		From this and our assumption that $|\hatx|$ has density at least $m^{-(1-\delta)w}$, it follows that $|I| \leq w$. 
		To prove \autoref{item:Xerror-dense}, we show that if $\hatx_{\alpha}$ is not $\alpha$-dense 
		then this contradicts 
		the maximality of $I$. Indeed, if $\hatx_{\alpha}$ is not $\alpha$-dense then there exists a nonempty subset 
		$J \subseteq [n] \setminus I$ and a witness $\alphaprime\in ([m] \union \set{\star})^n$ with 
		$\fix(\alphaprime) = {{J}}$ such that 
		$\Pr_{x\sim 
		\hatx_{\alpha}}[x_J = \alphaprime_J] \geq m^{-\delta |J|}$. 
        Let $\alpha \circ \alphaprime$ be the pointer 
		with $\fix(\alpha\circ\alphaprime)=I\cup J$ such that $(\alpha\circ\alphaprime)_I=\alpha_I$ and 
		$(\alpha\circ\alphaprime)_J=\alpha_J$.
        Then 
		 \begin{align}
		 	\Pr_{x\sim \hatx}[x_{I \cup J} = (\alpha \circ \alphaprime)_{I\union J}] 
                &= \Pr_{x\sim \hatx}[x_I=\alpha_I] \cdot \Pr_{x\sim \hatx}[x_J=\alphaprime_J \mid 
                x_I=\alpha_I] \\
                &= \Pr_{x\sim\hatx}[x_I=\alpha_I] \cdot \Pr_{x\sim \hatx_{\alpha}}[x_J=\alphaprime_J] \\
		 	&\geq m^{-\delta(|I|+|J|)},
		 \end{align}
	   meaning that $\hatx$ is also not dense on $I\cup J$, which contradicts the maximality of $I$.
		
		By \autoref{item:|I|} and \autoref{item:Xerror-dense} there is a strip $S \in 
		\mathsf{Strips}(T)$ with associated pointer $\alpha$, consisting of the rows %
		$x \in T^X$ for which $x_I=\alpha_I$.
		 Note that $\hatx_\alpha \subseteq S$, and since $\hatx_\alpha$ is $\alpha$-predense, the distinguished row 
		 $x^S$ of  strip $S$ cannot be strictly below all rows in $\hatx_{\alpha}$. However, this implies that some row 
		 $x \in \hatx_{\alpha}$ is secured by $S$. This is a contradiction, as $x \in \hatx_{\alpha} \subseteq X_{\err}^T$ 
		 where $X_{\err}^T$ contains only rows of $T$ that are not secured by any strip in $\mathsf{Strips}(T)$. 
	\end{proof}

\begin{claim}
		For any triangle $T$, the density of $Y_{\err}^T$ in $\{0,1\}^{mn}$ is at most $2^{-m^{\delta}/(16n) + (w+1)\log(2mn)}$.
\end{claim}
	\begin{proof}
	Each strip $S$ for $T$ is determined by a pointer $\alpha$ with $|\fix(\alpha)|\leq w$, of which 
	there are at most $\sum_{i=0}^w {n \choose i}m^i \leq (mn)^{w+1}$ many choices. 
 Moreover, for a fixed $\alpha$, there are at most $2^{|\fix(\alpha)|}\leq 2^w$ many 
 $\beta\in\{0,1,\star\}^{n}$ with $\fix(\beta) = \fix(\alpha)$, and hence at most this many $Y_\beta$ 
 which are not large, \ie each 
 contributing at most $2^{mn-m^\delta/(16n)}$ many columns. 
 Hence, $|Y^T_{\err}\mid < (2mn)^{w+1} 2^{mn-m^\delta/(16n)}$, and so the density of $Y^T_{\err}$ can be upper bounded by 
 $2^{-m^{\delta}/(16n) + (w+1)\log(2mn)}$.
	\end{proof}
This completes the proof of the \triLemma.

\section{Concluding Remarks}
\label{sec:conclusion}

This work opens up many exciting avenues for future research; we end by discussing the ones that 
we 
find most intriguing. 

\paragraph{Supercritical Trade-offs for Non-monotone Circuits.} We show that 
supercritical trade-offs exist for monotone circuits. What about for non-monotone circuits? Given 
that unconditional lower bounds for general circuits are beyond the reach of current techniques, it is 
interesting to prove the existence of such trade-offs under standard cryptographic assumptions, 
such as the existence of one-way functions. 

\paragraph{Supercritical Trade-offs for Perfect Matching and Tseitin.} 
Having established truly supercritical trade-offs for monotone circuits and cutting planes, 
we find it natural to ask for more examples of this phenomenon. As mentioned in the 
introduction, it is possible that the perfect matching problem exhibits such a trade-off for monotone circuits, and for 
cutting planes the Tseitin formulas are a candidate. The latter would also resolve the 
following question.

\paragraph{Separating Stabbing and Cutting Planes.} The quasi-polynomial size cutting planes 
proof of the Tseitin formulas was obtained by showing that a known upper bound on the Tseitin 
formulas in a proof system known as \emph{stabbing planes}~\cite{BFIKPPR18Stabbing} could be 
efficiently 
translated into cutting planes. In fact, as was shown in \cite{FGIPRTW21BranchAndCut}, any 
stabbing planes proof with sufficiently small coefficients can be translated into cutting planes. 
However, this transformation causes a blow-up in depth that is proportional to the size of the 
original proof. For example, the depth $O(\log^2n)$ stabbing planes proofs of the Tseitin formulas 
become quasi-polynomial-depth cutting planes proofs. 
Can one show that this blow-up is inevitable by giving a formula which has small stabbing planes 
proofs with low depth, however exhibits a supercritical size-depth trade-off for cutting planes? 

\paragraph{Further Applications of Variable Compression.} We give an %
application of variable compression in proof complexity. Is it possible to apply this technique to other problems? For example, can \emph{Pebbling formulas} and their associated graphs be compressed? New compressions for the Cop-Robber game would also be of interest.

\section*{Acknowledgements}

The authors are grateful for helpful discussions with
Christoph~Berkholz,
Jonas~Conneryd,
Daniel~Neuen,
and
Alexander~Razborov.
We would also like to thank the participants of the Oberwolfach
workshop \emph{Proof Complexity and Beyond} in March 2024
and of
the Dagstuhl workshop 24421 \emph{SAT and Interactions} for their
feedback.

Susanna F. de Rezende received funding from 
the
Knut and Alice Wallenberg grant \mbox{KAW 2021.0307}, ELLIIT,
and
the Swedish Research Council grant \mbox{2021-05104}. 
Noah Fleming was funded by NSERC.
Duri Andrea Janett and Jakob Nordström received funding from
the Independent Research Fund Denmark grant \mbox{9040-00389B},
and Jakob Nordström was also supported by the
Swedish Research Council grant \mbox{2016-00782}. 
Shuo Pang was funded by the European Union MSCA Postdoctoral
Fellowships 2023 project 101146273 NoShortProof. Views expressed are
the authors' and do not reflect the European Union or the Research
Executive Agency.

\bibliography{refArticles,refBooks,refOther,refLocal}

\bibliographystyle{alpha}

\appendix

\section{%
Proof of the Weisfeiler-Leman Result from the Cop-Robber Game}\label{appendix}

The goal of this appendix is to outline a proof of the trade-off for Weisfeiler-Leman (\autoref{thm:main_wl}) from our lower bound for the compressed Cop-Robber game (\autoref{thm:round}). 
The proof consists of two translations: first from \autoref{thm:round} %
to a round lower bound on a variant of the Cop-Robber game defined below, and then from there to WL-algorithms via the CFI construction \cite{CFI92OptimalLowerBound}. Both steps are standard if there is no compression. 
With compression, some details should be clarified. 
We start by recalling some definitions.

\subsection{Preliminaries on the Cai--F\"urer--Immerman Construction}
We recall the CFI-construction \cite{CFI92OptimalLowerBound,Furer01WeisfeilerLeman}. 
A graph $G=(V,E)$ is given with the vertices ordered, together with a function $\funcdescr{f}{E}{\mathbb{F}_2}$. 
We define a colored graph $\CFI{G}{f}$ as follows. 
For each $v\in V$, assume its degree is $d$, there is a group of $2^{d-1}$ vertices in $\CFI{G}{f}$ identified by $(v,\vec{a})$ for each 
$\vec{a}=(a_1,\dots,a_d)\in\mathbb{F}_2$ such that $a_1+\dots+a_d=0$ in $\mathbb{F}_2$, and these $2^{d-1}$ vertices are all colored $v$. 
The vertex order induces an adjacency list form of $G$. 
For any edge $e=\{u,v\}$ in $G$, assuming that $u$ is the $i$th neighbor of $v$ and $v$ is the $j$th neighbor of $u$, we add an edge between $(u,\vec{a})$ and $(v,\vec{b})$ if $a_i+b_j=f(e)$. Finally, the color of vertex $(v,\vec{a})$ is $v$.

The color-preserving automorphisms of $\CFI{G}{f}$ are characterized by \emph{twistings} \cite{CFI92OptimalLowerBound,gradel2019rank}. 
For a set of edges $F$ we denote by $\Vec{F}$ the set of directed edges $\setdescr{(u,v),(v,u)}{\set{u,v}\in F}$. We call $T\subseteq \vec{E}{(G)}$ a \emph{$G$-twisting} (or \emph{twisting} in short) if $\setsize{T\intersection (\set{v} \times V)}$ is even for every $v\in \vertices{G}$. An edge $\set{u,v}\in \edges{G}$ is \emph{twisted by $T$}, if $(u,v)\in T$ and $(v,u)\notin T$, or vice-versa. A vertex $v \in \vertices{G}$ is \emph{fixed by $T$} if $T\intersection (\set{v}\times V )= \emptyset$.

Next, we construct the compressed CFI graphs \cite{GLNS23CompressingCFI}. 
Assume $G$ is a vertex ordered graph and $\vertexeq$ is a compatible vertex equivalence relation on $\vertices{G}$ (recall this means $v\vertexeq v'$ implies that $v,v'$ are non-adjacent and have the same degree). 
Assume in addition that $\vertexeq$ satisfies the following condition: if $u\vertexeq u'$ and $v\vertexeq v'$, and $v$ is the $i$th neighbor of $u$, then $v'$ is the $i$ the neighbor of $u'$.\footnote{For $\basegraph$ with the lexicographical vertex order on $[k]\times[2\earlength+\midlength]$ and the vertex equivalence relation $\vertexeq$ in \autoref{def:compression}, this condition is satisfied.} 
We call a function $f: \edges{G}\to\mathbb{F}_2$ \emph{compressible} if for all $u\vertexeq u'$ and $v\vertexeq v'$ where both $\set{u,v}$ and $\set{u',v'}$ are edges, $f(\set{u,v})=f(\set{u',v'})$. 
Given such a $G$, such a vertex equivalence $\vertexeq$, and a compressible $f$, 
the \emph{compressed CFI graph $\CFI{G}{f}/\vertexeq$} is the \emph{quotient} graph of $\CFI{G}{f}$ by the vertex equivalence relation under which $(u,\vec{a})$ is equivalent to $(v,\vec{b})$ if and only if $u\vertexeq v$ and $\vec{a}=\vec{b}$. 
We denote this vertex equivalence on $\CFI{G}{f}$ still by $\vertexeq$. 
The color of vertices in $\CFI{G}{f}/\vertexeq$ can be defined in various canonical ways, e.g., coloring $(u,\vec{a})/\vertexeq$ by the minimum $v\in \vertices{G}$ in class $u/\vertexeq$. 

The color-preserving automorphisms of $\CFI{G}{f}/\vertexeq$ are characterized by \emph{compressed twistings}, which are $G$-twistings $T$ where, for any $(u_1,v_1)$ and $(u_2,v_2)$ such that $u_1\vertexeq u_2$ and $v_1$ has the same order in the neighbor-list of $u_1$ as $v_2$ does in that of $u_2$, it holds that $(u,v)\in T\Leftrightarrow (u',v')\in T$. 

For the cylinder graph $\basegraph$, we use the lexicographical vertex order on $[k]\times[2\earlength+\midlength]$. 
The associated adjacency-list form of $\basegraph$ induces an edge equivalence from $\vertexeq$ in \autoref{def:compressiongeneral}, which is the same as $\edgeeq$ in \autoref{def:compression}.

Finally, the argument uses a game variant where the Robber is on the edges. 
Given a graph compression $(G,\vertexeq,\edgeeq)$ as in \autoref{def:compressiongeneral}, the {\it compressed $k$-Cop-edge-Robber} game proceeds as in \autoref{def:cgame}, except that the Robber now stays on an edge, and the rule \ref{G2} becomes that the Robber should provide a compressible twisting $T$ which twists only the edge he currently occupies and the edge he moves to, and $T$ fixes every vertex that is $\vertexeq$-equivalent to a Cop position. 
The game ends if both endpoints of the Robber's edge are $\vertexeq$-equivalent to some Cop positions.

\subsection{Proof of Translation from the Cop-Robber Game to the Weisfeiler--Leman 
Algorithm}
We restrict our attention to the compression $(\basegraph,\vertexeq,\edgeeq)$ in \autoref{thm:round}. 
For ease of notation, we will call the game in \autoref{def:cgame} the {\it compressed $k$-VR} (indicating vertex Robber), 
and the above game variant the {\it compressed $k$-ER} (indicating edge Robber). 
We say a vertex $v\in \vertices{G}$ is \emph{singleton-class}\ if it forms a singleton $\vertexeq$-class, i.e., $\setsize{v_\vertexeq} = 1$. 
Similarly, we say an edge $e\in \edges{G}$ is \emph{singleton-class}\ if $\setsize{ e_\edgeeq} = 1$.

\autoref{lem:VRsim} says that our Robber strategy in \autoref{subsec:proof} can be carried over to the compressed game with edge Robber. 

\begin{lemma}%
\label{lem:VRsim}
    Assume that the Robber has a winning strategy in the $R$-round $k$-VR on $(\basegraph,\vertexeq,\edgeeq)$, where the Robber always occupies a vertex $v\in V$ that is singleton-class and such that there are two edges $e,e'\in E$ incident to $v$ that are singleton-class. Then there is a winning strategy for the Robber in the $R$-round $k$-ER on $G$ and $\vertexeq$, 
    where the Robber always occupies as singleton-class edge incident to a singleton-class vertex. 
\end{lemma}
\begin{proof}
    We fix a strategy $S$ for Robber that wins the $R$-round $k$-VR, and by simulating it, 
    we will construct a winning strategy for the Robber in the $R$-round $k$-ER.
    During the simulation, we maintain the invariant that the Cops occupy the same vertices in both graphs, that the Robber in $k$-VR is on a singleton-class vertex $w_1$ with two incident singleton-class edges $e_1$, $e_1'$, and that the Robber in $k$-ER is on either $e_1$ or $e_1'$. 
    In the initial round of the game the invariants hold by assumption, where we can assume without loss of generality that the Robber is placed on such an edge in the beginning.
    
    Assume that the invariants hold, that the game has lasted for $R'<R$ rounds, and that the Cops are playing step 1 (one Cops is picked up, and a destination $x\in \vertices{G}$ selected). In the $k$-VR, the same Cops is lifted up, and $x$ is signaled to the Robber. 
    As $R'<R$, the Robber in the $k$-VR has a compressible move $M\subseteq \edges{G}$ according to the strategy $S$ from $w_1$ to $w_2$, where $w_2$ is singleton-class and has two incident singleton-class edges $e_2$, $e_2'$. 
    Without loss of generality, we may assume that the Robber is on $e_1$ in the $k$-ER. We want to provide a $\vertexeq$-compressible $G$-twisting $T$ that only twists edges $e_1$ and $e_2$ (or possibly $e_2'$ instead of $e_2$), and fixes every vertex in a $\vertexeq$-class occupied by a Cops. 

    First, we assume that $w_1$ and $w_2$ are not adjacent. For all $v\in \vertices{M}\setminus \set{w_1,w_2}$, we include in $T$ all edges in $M$ incident to $v$ in the outgoing direction from $v$, i.e., for any edge $\set{v,v'}\in M$, we include $(v,v')$ in $T$. For $w_1$ and $w_2$, we do the same, except for the edges $e_1$ and $e_2$. If $e_1$ (or $e_2$) is in $M$, we do not include it in $T$ in the outgoing direction from $w_1$ or $w_2$, respectively. Otherwise (that is, $e_1\notin M$ or $e_2\notin M$), we include $e_1$ or $e_2$ in the outgoing direction from $w_1$ or $w_2$ in $T$, respectively. $T$ is a twisting, as all $v\in V$ satisfy that $T_v=T\intersection (\set{v}\times V)$ is of even size. (For $v\notin \vertices{M}$, $T_v=\emptyset$, for $v\in \vertices{M}\setminus \set{w_1,w_2}$, $\setsize{ T_v}= \deg_M(v)$, which is even by \ref{G2}(c), and for $w_1,w_2$, $\setsize{ T_v} = \deg_M(v) \pm 1$, which is even by \ref{G2}(c).) $T$ is $\vertexeq$-compressible, as $M$ is closed under $\edgeeq$ by \ref{G2}(a), and the directed edges in exactly one of $T$ and $\vec{M}$ %
    are singleton-class. Now $T$ twists exactly $e_1$ and $e_2$ by construction, and all the vertices in a $\vertexeq$-class occupied by a Cops are fixed due to \ref{G2}(b). 

    Let us now turn to the case where $w_1$ and $w_2$ are adjacent. Recall that the Robber is on $e_1$. If $e_1 \neq \set{w_1, w_2}$, let $e'\in \set{e_2, e_2'}$ be such that $e'\neq \set{w_1, w_2}$, then we can use the same twisting $T$ as above, twisting only $e_1$ and $e'$. %
    So assume that $e_1 = \set{w_1, w_2}$. Without loss of generality, we assume that $e_2 \neq e_1$ (otherwise switch the roles of $e_2$ and $e_2'$). To construct $T$, we treat $v\in \vertices{M}\setminus\set{w_2}$ as above, and immediately get that $T_v$ is of even size for $v\neq w_2$. For $w_2$, we include in $T$ all incident edges in $M$ except $e_1$ and $e_2$, in the outgoing direction from $w_2$. If $e_1\in M$, we include it in $T$ in the outgoing direction from $w_2$. If $e_2\notin M$, we include it in $T$ in the outgoing direction from $w_2$. 
    To see that $\setsize{ T_{w_2} }$ is even, note that if $e_2\in M$, $\setsize{ T_{w_2} }= \deg_M (w_2)-1$, and otherwise $\setsize{ T_{w_2} }= \deg_M (w_2)+1$. $T$ is $\vertexeq$-compressible, twists only $e_1$ and $e_2$, and fixes the vertices in the $\vertexeq$-class of the Cops for the same reasons as above.

    In the $k$-ER, the Robber moves according to the twisting $T$ constructed above. Finally, a Cops is placed on $x$ in step 3 of the $k$-ER. The same happens in the $k$-VR. As the Robber is not caught in the $k$-VR, the Robber is also not caught in the $k$-ER,
    Since the invariants are maintained throughout the simulation, %
    this concludes the proof of this lemma.
\end{proof}

\begin{proof}[Proof of \autoref{thm:main_wl}]%
    For clarity, we will use $N$ for the number of vertices in each graph in the pair, leaving $n$ as the parameter in the graph compression in \autoref{def:parameters}.
    We follow the standard chain of reasoning as in \cite{GLNS23CompressingCFI}, pointing out necessary changes. 
    Given $k,c$ as in the theorem statement,
    we take $n$ to be large enough such that the conditions of \autoref{def:parameters} are satisfied.
    We apply the CFI construction with $f,g\colon \edges{\basegraph}\rightarrow \F_2$, where $f$ is the all-zero function and $g$ is 1 only on one edge adjacent to vertex $(1,1)$. 
    This gives us graphs $G_{N}\coloneqq\textsc{CFI}(\basegraph,f)/\vertexeq$ and $H_{N}\coloneqq\textsc{CFI}(\basegraph,g)/\vertexeq$,
    where $N$ denotes the vertex set size of both. 
    We observe that $N\leq 2^{4-1}\left(2\earlength+k^2(2n)^{c+1}\right)<2^{c+5}k^2n^{c+1}$. 

    Using \autoref{lem:VRsim} and \autoref{thm:round}, we get their Lemma 29 
    and consequently Theorem~30 with $k+1$ replaced by $k+c$, which says that $(k+c-1)$-WL requires at least $(\midlength-2\earlength)/(8(k+c))\geq n^k/(32k)$ rounds to distinguish $G_N$ and $H_N$. %
    This is at least $\left(2^{-(c+10)}k^{-3}N\right)^{k/(c+1)}$. 
\end{proof} 
\section{Proof of {\fullImg}}\label{sec-Appendix-monotone}

\newtheorem*{fullimagelemma_app}{Full Image Lemma}

In this appendix, we give a self-contained proof of the \fullImg, restated below.

\begin{fullimagelemma_app}\label{lem:fullImage_app}
For any positive integers $m,n$, rectangle $X\times Y\subseteq \{m\}^n\times \{0,1\}^{mn}$, and parameter $\delta\in(0,1)$, assume $m^\delta \geq \frac{4}{\ln 2}n$,  
\begin{enumerate}[label=(\arabic*)]
    \item \label{item:condX,app} $H_{\infty}(\textbf{X}_I)\geq \delta|I|\log m-1$ for any nonempty $I\subseteq[n]$,   
    \item \label{item:condY,app} $H_{\infty}(Y) \geq {mn} - {m^{\delta}/(8n)}$.
\end{enumerate}
There exists $x^*\in X$ such that $\IND(\{x^*\},Y)=\{0,1\}^n$.
\end{fullimagelemma_app}

We claim no originality here, as the argument is the same as in \cite{lmmpz2022liftingWithSunflowers}. Except that we substitute their use of strong sunflower lemmas with a simpler result from \cite{fknp2021thresholds} proved by Janson's inequality.

Given a set $U$, a set sequence $(S_1,\ldots,S_l)$ is \textit{$\kappa$-spread over $U$} if each $S_i$ is a subset of $U$ and for any $W\subseteq U$, the number of elements in the sequence that contains $W$ is at most $l\kappa ^{-|W|}$. 

\begin{proposition}[Lemma 3.2 of \cite{fknp2021thresholds}]\label{prop:satisfying}
    Suppose $(S_1,\ldots,S_l)$ is a sequence of size-$r$ sets that is $\kappa$-spread over $U$. For any $p\in(0,1)$, if $W$ is a random subset of $U$ where each element is included independently with probability $p$, then
    \begin{equation}\label{eq:satisfying}
        \Pr_{W}\left[ (\forall i\in[l])\ S_i\not\subseteq W \right] \leq 
        \exp\left(-\frac{p\kappa}{r}\exp(-\frac{r-1}{p\kappa})\right).
    \end{equation}
In particular, if $p\kappa\geq(r-1)/\ln2$, then the bound in \eqref{eq:satisfying} can be replaced by $\exp(-p\kappa/(2r))$.
\end{proposition}
\begin{proof}[Proof of \autoref{prop:satisfying}]
Denote by $\xi_i$ the indicator variable of the event $S_i\subseteq W$. Let   
\begin{align}
\mu&:=\Sum_{i=1}^l \E[\xi_i]\label{eq:mu}\\
\Lambda&:=\Sum_{(i,j):\ S_i\cap S_j\neq\emptyset} \E[\xi_i\cdot\xi_j] 
= 
\Sum_{i\in [l]}\ \Sum_{j:\ S_j\cap S_i\neq\emptyset} \E[\xi_i\cdot\xi_j] \label{eq:Lambda}
\end{align}
Then $\mu=l\cdot p^{r}$, and we can upper bound $\Lambda$ as follows. 
For each inner sum in \eqref{eq:Lambda}, we group the terms according to $a:=|S_j\cap S_i|$, where in each group there are $\binom{r}{a}$ many choices of $S_j\cap S_i$, and for each choice there are at most $l\kappa^{-a}$ many possible $j$ due to spreadness. 
So 
\begin{equation}\label{eq:Lambda_bound}
    \Lambda \leq l\cdot \Sum_{a=1}^r \binom{r}{a}(l\kappa^{-a})p^{2r-a} = 
\mu^2\left((1+\frac{1}{p\kappa})^r-1\right) \leq \mu^2\cdot \frac{r}{p\kappa}\exp(\frac{r-1}{p\kappa})
\end{equation}
where the last step uses $(1+x)^r-1\leq x\cdot re^{(r-1)x}$. 
By Janson's inequality applied to $(\xi_1,\dots,\xi_l)$, 
\[\Pr_{W}\left[ (\forall i\in[l])\ \xi_i=0 \right] \leq \exp(-\frac{\mu^2}{\Lambda}) 
\leq \exp\left(-\frac{p\kappa}{r}\exp(-\frac{r-1}{p\kappa})\right).\qedhere\]
\end{proof}

\begin{remark}The probability bound in \eqref{eq:satisfying} was improved by \cite[Lemma 4]{rao2020coding} to $r\cdot\exp(-p\kappa/\Cabs)$ if $p<\frac{1}{2}$, where $\Cabs$ is an absolute constant. (There, a slightly different notion of $r$-spreadness and distribution of $W$ is used, but the argument is adaptable.) 
We use \eqref{eq:satisfying} for simplicity, which is sufficient for our purpose (we don't intend to optimize gadget size).
\end{remark}

\begin{proof}[Proof of {\fullImgApp}] 
Given $X\subseteq [m]^n$ as in the lemma, we consider the maximum possible size of  $Y'\!\subseteq\!\{0,1\}^{mn}$ subject to the following condition:
\begin{equation}\label{eq:badY}
\text{($\forall x\in X$) ($\exists z_x\in\{0,1\}^n$) $z_x \notin \IND(\{x\}\times Y')$.}
\end{equation} 
As the first step, we show that for the purpose of maximizing $|Y'|$, we can assume $z_x=\vec{1}$ for all $x\in X$ in \eqref{eq:badY}. 
Then, as the second step, we use \autoref{prop:satisfying} to show $|Y'|<2^{mn}2^{-m^{\delta}/(8n)}$ in that case. 
Given our assumption \ref{item:condY,app} on $|Y|$, this means \eqref{eq:badY} cannot hold for $Y'\leftarrow Y$, so the lemma follows. 

For the first step, take $mn$ boolean variables $p_{i,j}$ ($i\in[n]$, $j\in[m]$). 
For each $x\in[m]^n$, we take a clause $C_x:=\bigvee\limits_{i=1}^n p_{i,x(i)}^{1-z_x(i)}$, where 
$p^1:=p$ and $p^0:=\stdnot{p}$.  
Then condition \eqref{eq:badY} equivalently says that each element in $Y$ is a satisfying assignment of the CNF $C:=\bigwedge\limits_{x\in X} C_x$. 
We have the following: 

\begin{fact}[\cite{lmmpz2022liftingWithSunflowers}, Claim 8]\label{fact:CNFmon} For any CNF $F$, let $F^{mon}$ be its monotonization by negating each positive literal. Then $\left\vert \sat(F) \right\vert\leq \left\vert \sat(F^{mon}) \right\vert$, where $\sat(\cdot)$ means the set of satisfying assignments.
\end{fact}

To see this, note that $F^{mon}$ can be obtained from a sequence $F_0:=F, F_1,\ldots,F_N=F^{mon}$ where $F_i$ monotonize $F_{i-1}$ at the $i$th variable, and  
we only need to show that $|\sat(F_{i-1})|\leq |\sat(F_i)|$. 
For this aim, we fix an $i\in[n]$ and view $\{0,1\}^n$ as a collection of pairs $P_y=\{(y,0),(y,1)\}$ over $y\in\{0,1\}^{[n]\backslash\{i\}}$. 
For each $y$, if $P_y\subseteq \sat(F_{i-1})$ then the partial assignment $y$ already satisfies each clause in $F_{i-1}$, so $P_y\subseteq \sat(F_i)$; 
if exactly one of $(y,0)$ and $(y,1)$ satisfies $F_{i-1}$, then $(y,0)$ satisfies $F_i$. 
Thus for each $y$, $\left\vert \sat(F_{i-1})\cap P_y\right\vert \leq \left\vert \sat(F_{i})\cap P_y\right\vert$. 
Consequently, $|\sat(F_{i-1})|\leq |\sat(F_i)|$, and 
\autoref{fact:CNFmon} follows.
\smallskip

\autoref{fact:CNFmon} says that for the purpose of maximizing $|Y|$, we can take $z_x=\vec{1}$ for all $x$ in condition \eqref{eq:badY}. 
So below we fix $z_x=\vec{1}$ and upper bound $|Y|$. 
The argument uses a translation of language as follows. 
Take a ground set $U:=[mn]$, understood as the union of $n$ disjoint sets each having size $m$; 
we call these $m$ sets $m$ `groups' for clarity. 
We will consider set systems over $U$. 
Each $x\in[m]^n$ corresponds to a subset $\underline{x}\subseteq U$ which contains one element per group. 
Each $y\in Y$, when read as a function from $U$ to $\{0,1\}$, corresponds to the subset $\underline{y}\coloneqq y^{-1}(1)\subseteq U$. 
Then, assuming $z_x=\vec{1}$ for all $x$, condition \eqref{eq:badY} becomes
\begin{equation}\label{eq:badY,2}
     (\forall y\in Y)(\forall x\in X)\ \underline{x}\not\subseteq\underline{y}.
\end{equation}  
The following fact provides the last bit of the translation, proved by a direct inspection of definitions.
\begin{fact}\label{fact:spreadness} 
For any $A\subseteq[m]^n$ and $\kappa\geq 1$, $\underline{A}=\{\underline{a}\mid a\in A\}$ is $\kappa$-spread if and only if $H_{\infty}({\bf A}_I)\geq |I|\log \kappa$ for all nonempty $I\subseteq [n]$.
\end{fact} 
By \autoref{fact:spreadness} and the assumption \ref{item:condX,app} on $X$, the set family $\underline{X}$ is ${\frac{1}{2}}m^\delta$-spread. 
So we can apply \autoref{prop:satisfying} with parameters $r:=n$, $\kappa:=\frac{1}{2}m^{\delta}$, $p:={1\over 2}$, where it holds that $p\kappa = m^{\delta}/4 \geq r/\ln 2$, so the ``in particular'' part of the proposition applies. 
As the result, \eqref{eq:badY,2} implies that $|Y|< 2^{mn}2^{-m^{\delta}/(8n)}$. Then by our first step, \eqref{eq:badY} also implies $|Y|< 2^{mn}2^{-m^{\delta}/(8n)}$. But $|Y|\geq 2^{mn}2^{-m^{\delta}/(8n)}$ by the assumption \ref{item:condY,app}, so \eqref{eq:badY} cannot hold, i.e., $\exists x^*\in X$ such that $\IND(\{x^*\}\times Y)=\{0,1\}^n$. 
The {\fullImgApp} follows.
\end{proof} 
\end{document}